\PassOptionsToPackage{pagebackref,draft=false}{hyperref}
\documentclass[a4paper,onecolumn, superscriptaddress,10pt, accepted=2022-08-29, issue=7, volume=5, shorttitle=papers/compositionality-5-7]{compositionalityarticle}
\pdfoutput=1

	\usepackage{geometry}			

	\usepackage{graphicx}
	\usepackage[utf8]{inputenc}
	\usepackage[english]{babel}
	\usepackage[T1]{fontenc}

	\usepackage{amssymb}			
	\usepackage{amsmath}			
		\makeatletter
			\g@addto@macro\normalsize{%
			  \setlength\abovedisplayskip{8pt plus 3pt minus 4pt}
			  \setlength\belowdisplayskip{8pt plus 3pt minus 4pt}
			  \setlength\abovedisplayshortskip{6pt plus 2pt minus 3pt}
			  \setlength\belowdisplayshortskip{6pt plus 2pt minus 3pt}
			}
		\makeatother
	\usepackage{amsthm}				
		\makeatletter
			\def\thm@space@setup{\thm@preskip = 9pt plus 2pt minus 4pt
			\thm@postskip = \thm@preskip}
		\makeatother
	\usepackage{array,arydshln,multirow}	
	\usepackage[dvipsnames]{xcolor}
	\usepackage{bbm}
	\usepackage{bm}				
	\usepackage{breakurl}
	\usepackage{calc}
	\usepackage{centernot}			
	\usepackage{comment}
	\usepackage{csquotes}
	\usepackage[shortlabels]{enumitem}

		\setenumerate{itemsep=2pt,topsep=2pt}
		\setenumerate[1]{label=(\roman*)}
	\usepackage{etoolbox}		
	\usepackage{mathtools}
	\usepackage[numbers,sort&compress]{natbib}
	\usepackage{scalerel}
	\usepackage{soul}				
	\usepackage{subcaption}			
	\usepackage{tikz}
		\usetikzlibrary{tikzmark,calc,backgrounds,positioning}
	\usepackage{tikz-cd} 				
		\tikzcdset{
			cells = {/tikz/inner xsep = 1.5ex, /tikz/inner ysep = 1ex}, 
			labels = {font = \footnotesize}
		}
	\usepackage{tikzit}				
		\usetikzlibrary{decorations.pathreplacing,arrows,shapes.geometric}

\tikzstyle{morphism}=[fill=white, draw=black, shape=rectangle]
\tikzstyle{medium box}=[fill=white, draw=black, shape=rectangle, minimum width=0.8cm, minimum height=0.9cm]
\tikzstyle{large morphism}=[fill=white, draw=black, shape=rectangle, minimum width=1.7cm, minimum height=1cm]
\tikzstyle{bn}=[fill=black, draw=black, shape=circle, inner sep=1.5pt]
\tikzstyle{state}=[fill=white, draw=black, regular polygon, regular polygon sides=3, minimum width=0.8cm, shape border rotate=180, inner sep=0pt]
\tikzstyle{medium state}=[fill=white, draw=black, regular polygon, regular polygon sides=3, minimum width=1.3cm, inner sep=0pt, shape border rotate=180]
\tikzstyle{large state}=[fill=white, draw=black, regular polygon, regular polygon sides=3, minimum width=2.2cm, shape border rotate=180, inner sep=0pt]
\tikzstyle{wn}=[fill=white, draw=black, shape=circle, inner sep=1.5pt]
\tikzstyle{blue morphism}=[fill=white, draw={rgb,255: red,15; green,0; blue,150}, shape=rectangle, text={rgb,255: red,15; green,0; blue,150}, tikzit category=blue]
\tikzstyle{blue state}=[fill=white, draw={rgb,255: red,15; green,0; blue,150}, shape=circle, regular polygon, regular polygon sides=3, minimum width=0.8cm, shape border rotate=180, inner sep=0pt, text={rgb,255: red,15; green,0; blue,150}, tikzit category=blue]
\tikzstyle{blue node}=[fill={rgb,255: red,15; green,0; blue,150}, draw={rgb,255: red,15; green,0; blue,150}, shape=circle, tikzit category=blue, inner sep=1.5pt]
\tikzstyle{blue}=[text={rgb,255: red,15; green,0; blue,150}, tikzit draw={rgb,255: red,191; green,191; blue,191}, tikzit category=blue, tikzit fill=white, inner sep=0mm]
\tikzstyle{red node}=[fill={rgb,255: red,150; green,0; blue,2}, draw={rgb,255: red,150; green,0; blue,2}, shape=circle, inner sep=1.5pt]
\tikzstyle{Purple node}=[fill={rgb,255: red,150; green,0; blue,150}, draw={rgb,255: red,150; green,0; blue,150}, shape=circle, inner sep=1.5pt]
\tikzstyle{red}=[text={rgb,255: red,150; green,0; blue,2}, inner sep=0mm, tikzit fill=white, tikzit draw={rgb,255: red,191; green,191; blue,191}]
\tikzstyle{purple}=[text={rgb,255: red,150; green,0; blue,150}, inner sep=0mm, tikzit fill=white, tikzit draw={rgb,255: red,191; green,191; blue,191}]
\tikzstyle{wide state}=[fill=white, draw=black, shape=isosceles triangle, minimum width=0.8cm, shape border rotate=270, inner sep=1.4pt, minimum height=0.5cm, isosceles triangle apex angle=80]
\tikzstyle{white morphism}=[fill=white, draw=white, shape=rectangle, tikzit draw={rgb,255: red,139; green,139; blue,139}]
\tikzstyle{node}=[fill=none, draw=black, shape=circle, inner sep=1.5pt, tikzit fill=white]
\tikzstyle{node_fix}=[fill=none, draw=black, shape=circle, tikzit fill=white, minimum size=15pt, inner sep=0.5pt]
\tikzstyle{node_med}=[fill=none, draw=black, shape=circle, tikzit fill=white, minimum size=17pt, inner sep=1pt]
\tikzstyle{triangle_fix}=[regular polygon, regular polygon sides=3, fill=none, draw=black, minimum size=25pt, inner sep=0.5pt, tikzit fill=white]
\tikzstyle{green point}=[fill={rgb,255: red,33; green,153; blue,51}, draw=black, shape=circle, inner sep=1.5pt]
\tikzstyle{small node}=[fill=black, draw=black, shape=circle, inner sep=1pt]
\tikzstyle{grey}=[text={rgb,255: red,64; green,64; blue,64}, inner sep=0mm, shape=circle, tikzit fill=white, tikzit draw={rgb,255: red,64; green,64; blue,64}]

\tikzstyle{arrow}=[->]
\tikzstyle{dashed box}=[-, dashed]
\tikzstyle{blue arrow}=[-, draw={rgb,255: red,15; green,0; blue,150}, tikzit category=blue]
\tikzstyle{mapsto}=[{|->}]
\tikzstyle{double wire}=[-, double]
\tikzstyle{curly brace}=[decorate,decoration={brace,amplitude=10pt}]
\tikzstyle{curly brace 2}=[decorate,decoration={brace,amplitude=5pt}]
\tikzstyle{fill_ar}=[->, >=stealth']
\tikzstyle{fill_ar2}=[<->, >=stealth']
\tikzstyle{white}=[-, draw=white, tikzit draw={rgb,255: red,141; green,141; blue,141}]
\tikzstyle{white_ar}=[tikzit draw={rgb,255: red,141; green,141; blue,141}, draw=white, ->, >=stealth']
\tikzstyle{ellipse}=[-Ellipse]
\tikzstyle{segment}=[{|-|}]
\tikzstyle{thick_ar}=[->, thick]
\tikzstyle{orange edge}=[-, draw={rgb,255: red,215; green,143; blue,0}, thick]
\tikzstyle{purple edge}=[-, draw={rgb,255: red,125; green,0; blue,250}, thick]
\tikzstyle{thick edge}=[-, draw={rgb,255: red,145; green,145; blue,145}, fill={rgb,255: red,220; green,220; blue,220}, tikzit draw=black, tikzit fill={rgb,255: red,198; green,198; blue,198}]
\tikzstyle{thick grey}=[-, draw={rgb,255: red,145; green,145; blue,145}, fill={rgb,255: red,249; green,245; blue,183}, tikzit draw=black, tikzit fill={rgb,255: red,249; green,245; blue,183}]
\tikzstyle{protected}=[-, preaction={{ultra thick,white,draw}}]
\tikzstyle{blue fill}=[-, draw=none, fill={rgb,255: red,233; green,244; blue,255}, tikzit draw={rgb,255: red,0; green,125; blue,178}]
	\usepackage{todonotes}

	\usepackage{varwidth}			
	\usepackage{xparse}				
	
	\usepackage{bookmark}			
	
	\usepackage[noabbrev]{cleveref}

	\newlength{\myparskip}
	\newlength{\scale}
	\newcommand{\red}[1]{\textcolor{red}{#1}}

	\newcommand{\rob}[1]{{\color{OliveGreen!80!black} #1}}
	
		\newtheorem{theorem}{Theorem}
		\newtheorem{corollary}[theorem]{Corollary}
		\newtheorem{lemma}[theorem]{Lemma}
			\newenvironment{lemma'}[1]{\renewcommand{\thetheorem}{\ref*{#1}$'$}\addtocounter{theorem}{-1}\begin{lemma}}{\end{lemma}}		
		\newtheorem{proposition}[theorem]{Proposition}
		
	\theoremstyle{definition}

		\newtheorem{definition}[theorem]{Definition}
		\newtheorem{example}[theorem]{Example}
			\newenvironment{example'}[1]{\renewcommand{\thetheorem}{\ref*{#1}$'$}\addtocounter{theorem}{-1}\begin{example}}{\end{example}}		
		
		\newtheorem{framework}{Framework}
			\newenvironment{framework'}[1]{\renewcommand{\theframework}{\ref*{#1}$'$}\addtocounter{framework}{-1}\begin{framework}}{\end{framework}}		
		\newtheorem*{namedthm}{\namedthmname}	
			\newcounter{namedthm}
			\makeatletter
			\newenvironment{named}[1]
				{\def\namedthmname{#1}%
					\refstepcounter{namedthm}%
					\namedthm\def\@currentlabel{#1}}
				{\endnamedthm}
			\makeatother

		\newtheorem{remark}[theorem]{Remark}
	\newtheoremstyle{question}{\myparskip}{\myparskip}{\color{BrickRed}\normalfont}{0pt}{\bfseries}{.}{5pt plus 1pt minus 1pt}{}
	\theoremstyle{question}
		
	\newtheoremstyle{answer}{\myparskip}{\myparskip}{\color{PineGreen}\normalfont}{0pt}{\bfseries}{.}{5pt plus 1pt minus 1pt}{\thmname{#1}\thmnote{ \bfseries #3}}
	\theoremstyle{answer}
		
	\theoremstyle{remark}

	\newtheoremstyle{theoremrepeat}{\myparskip}{\myparskip}{\itshape}{0pt}{\bfseries}{.}{5pt plus 1pt minus 1pt}{\thmname{#1}\thmnote{ \bfseries #3}}
	\theoremstyle{theoremrepeat}

	\DeclarePairedDelimiter{\abs}{\lvert}{\rvert}
	\DeclarePairedDelimiter{\norm}{\lVert}{\rVert}
	\DeclarePairedDelimiterXPP{\pnorm}[2]{}{\lVert}{\rVert}{_{#1}}{#2}
	\makeatletter
		\let\oldabs\abs
		\def\abs{\@ifstar{\oldabs}{\oldabs*}}
		\let\oldnorm\norm
		\def\norm{\@ifstar{\oldnorm}{\oldnorm*}}
		\let\oldpnorm\pnorm
		\def\pnorm{\@ifstar{\oldpnorm}{\oldpnorm*}}
	\makeatother

	\newsavebox{\numbox}%
	\newsavebox{\slashbox}%
	\newsavebox{\denbox}%
	\DeclareDocumentCommand{\newfaktor}{m O{0.5} m O{-0.5}}{
		\savebox{\numbox}{\ensuremath{#1}}
		\savebox{\slashbox}{\ensuremath{\diagup}}
		\savebox{\denbox}{\ensuremath{#3}}
		\raisebox{#2\ht\slashbox}{\usebox{\numbox}}
		\mkern-5mu%
		\rotatebox{-44}{\rule[#4\ht\denbox]{0.4pt}{-#4\ht\denbox+#2\ht\numbox+\ht\numbox}}
		\mkern-4mu%
		\raisebox{#4\ht\denbox}{\usebox{\denbox}}
	}
	
	\providecommand{\given}{}			
	\newcommand{\SetSymbol}[1][]{%
		\nonscript\;\,#1\vert
		\allowbreak
		\nonscript\;\,
		\mathopen{}
	}
	\DeclarePairedDelimiterX{\Set}[1]{\{}{\}}{%
		\renewcommand{\given}{\SetSymbol[\delimsize]}
		#1
	}
	\makeatletter
		\let\oldSet\Set
		\def\Set{\@ifstar{\oldSet}{\oldSet*}}
	\makeatother

	\makeatletter
		\renewcommand*\env@matrix[1][\arraystretch]{%
			\edef\arraystretch{#1}%
			\hskip -\arraycolsep
			\let\@ifnextchar\new@ifnextchar
			\array{*\c@MaxMatrixCols c}
		}
	\makeatother

	\newcommand{\Id}{\mathsf{Id}}

	\newcommand*{\coloniff}{\mathrel{\vcentcolon\Longleftrightarrow}}

	\def\Tr{\operatorname{Tr}}

	\renewcommand{\geq}{\geqslant}
	\renewcommand{\leq}{\leqslant}

	\newcommand{\bea}{\begin{eqnarray}}
	\newcommand{\eea}{\end{eqnarray}}
	\newcommand{\be}{\begin{equation}}
	\newcommand{\ee}{\end{equation}}
	\newcommand{\ba}{\begin{equation}\begin{aligned}}
	\newcommand{\ea}{\end{aligned}\end{equation}}
	
	\newcommand{\bit}{\begin{itemize}}
	\newcommand{\eit}{\end{itemize}\par\noindent}
	\newcommand{\ben}{\begin{enumerate}}
	\newcommand{\een}{\end{enumerate}\par\noindent}
	\newcommand{\beq}{\begin{equation}}
	\newcommand{\eeq}{\end{equation}\par\noindent}
	\newcommand{\beqa}{\begin{eqnarray*}}
	\newcommand{\eeqa}{\end{eqnarray*}\par\noindent}
	\newcommand{\beqn}{\begin{eqnarray}}
	\newcommand{\eeqn}{\end{eqnarray}\par\noindent}

	\newcommand{\bra}[1]{\langle #1|}
	\newcommand{\ket}[1]{|#1\rangle}
	
	\newcommand{\down}{\mathord{\downarrow}}
	\newcommand{\up}{\mathord{\uparrow}}
	\newcommand{\downD}{\mathord{\downarrow_D}}
	\newcommand{\upD}{\mathord{\uparrow_D}}
	\newcommand{\downwrt}[1]{\mathord{\downarrow_{#1}}}
	\newcommand{\upwrt}[1]{\mathord{\uparrow_{#1}}}

	\newcommand{\cost}[2]{#1\textup{-}\mathsf{cost}_{#2}}
	\newcommand{\yield}[2]{#1\textup{-}\mathsf{yield}_{#2}}
	\newcommand{\fmax}{f\textup{-}\mathsf{max}}
	\newcommand{\fmin}{f\textup{-}\mathsf{min}}
	
	\newcommand{\xmin}[1]{#1\textup{-}\mathsf{min}}
	
	\newcommand{\calmin}{\mathsf{cva}\textup{-}\mathsf{min}}
	\newcommand{\fWmax}{f_W\textup{-}\mathsf{max}}
	\newcommand{\fWmin}{f_W\textup{-}\mathsf{min}}
	\newcommand{\R}[1]{{R}^{\scriptscriptstyle{(#1)}}}

	\newcommand{\Rfreebasic}[1]{\mathcal{R}^{\scriptscriptstyle{(#1)}}_{\rm free}}
	\newcommand{\Rfree}[1]{\left( \mathcal{R}_{\rm free} \right)^{\sqcuptimes k}}
	
	\newcommand{\Rcons}[1]{\mathcal{R}^{\scriptscriptstyle{(#1)}}_{\rm cons}}
	\newcommand{\Rconsfree}[1]{\mathcal{R}^{\scriptscriptstyle{(#1)}}_{\rm cons,free}}
	
	\newcommand{\Interesting}[2]{#1\textup{-}\mathsf{Interesting}\left( #2 \right) }
	\newcommand{\freeconv}{\succeq}
	\newcommand{\reals}{\overline{\mathbb{R}}}
	\newcommand{\ordreals}{(\reals,\geq)}
	\newcommand{\RT}[1]{({#1}_{\rm free}, {#1}, \boxtimes)}

	\newcommand{\res}{R}
	
	\newcommand{\resfree}{\res_{\rm free}}
	\newcommand{\ordres}{(\res,\freeconv)}

	\newcommand{\tabitem}{\hspace{.2em}\textbullet\hspace{.4em}}
	
	\newcommand{\enhgeq}{\geq_{\rm enh}}
	\newcommand{\deggeq}{\geq_{\rm deg}}
	\newcommand{\enhconv}{\freeconv_{\rm enh}}
	\newcommand{\degconv}{\freeconv_{\rm deg}}

	\DeclareRobustCommand{\bitcoin}{{%
		\normalfont\sffamily
		\raisebox{-.05ex}{\makebox[.1\width][l]{-\kern-.2em-}}B%
	}}
	
	\makeatletter
		\newcommand{\raisedtimes}[1]{\mathpalette\do@raisedtimes{#1}\relax}
		\newcommand{\do@raisedtimes}[2]{\raisebox{#2}{$\m@th#1\times$}}
		\newcommand{\scaledtimes}[2]{\scaleobj{#1}{\raisedtimes{#2}}}
		
		\newcommand{\sqcuptimes}{\mathrel{\mathpalette\do@sqcuptimes\relax}}
		\newcommand{\do@sqcuptimes}[2]{%
			\ooalign{%
				\hidewidth$#1\m@th\scaledtimes{0.7}{2pt}$\hidewidth\cr
				$#1\m@th\sqcup$\cr
			}%
		}
		\newcommand{\bigsqcuptimes}{\mathop{\mathpalette\do@bigsqcuptimes\relax}}
		\newcommand{\do@bigsqcuptimes}[2]{%
			\ooalign{%
				\hidewidth$#1\m@th\scaledtimes{1.1}{0pt}$\hidewidth\cr
				$#1\m@th\bigsqcup$\cr
			}%
		}
	\makeatother

	\newcommand{\hyphenationsetting}{%
		\emergencystretch=0pt	
		\tolerance=2000			
		\pretolerance=1000		
		\righthyphenmin=4		
		\lefthyphenmin=4			
	}
	
	\hyphenationsetting

\begin{document}

	\title{Monotones in General Resource Theories}
	\date{}
	\author{Tomáš Gonda}
	\email{tomas.gonda@uibk.ac.at}
	\affiliation{University of Innsbruck, Innrain 52, 6020 Innsbruck, Austria}
	\author{Robert W. Spekkens}
	\email{rspekkens@perimeterinstitute.ca}
	\affiliation{Perimeter Institute for Theoretical Physics, 31 Caroline St N, Waterloo, ON N2L 2Y5, Canada}
	
	\maketitle

	\begin{abstract}
		\noindent
		A central problem in the study of resource theories is to find functions that are nonincreasing under resource conversions---termed monotones---in order to quantify resourcefulness.  
		Various constructions of monotones appear in many different concrete resource theories.
		How general are these constructions?
		What are the necessary conditions on a resource theory for a given construction to be applicable?
		To answer these questions, we introduce a broad scheme for constructing monotones. 
		It involves finding an order-preserving map from the preorder of resources of interest to a distinct preorder for which nontrivial monotones are previously known or can be more easily constructed; these monotones are then pulled back through the map. 
		In one of the two main classes we study, the preorder of resources is mapped to a preorder of sets of resources, where the order relation is set inclusion, such that monotones can be defined via maximizing or minimizing the value of a function within these sets.
		In the other class, the preorder of resources is mapped to a preorder of tuples of resources, and one pulls back monotones that measure the amount of distinguishability of the different elements of the tuple (hence its information content). 
		Monotones based on contractions arise naturally in the latter class, and, more surprisingly, so do weight and robustness measures. 
		In addition to capturing many standard monotone constructions, our scheme also suggests significant generalizations of these.
		In order to properly capture the breadth of applicability of our results, we present them within a novel abstract framework for resource theories in which the notion of composition is independent of the types of the resources involved (i.e., whether they are states, channels, combs, etc.).
	\end{abstract}

	\tableofcontents
	
	\setlength{\parskip}{\myparskip}				
	\setlength{\parindent}{0pt}					
		

	\section{Introduction}
\label{sec:Introduction}
	
	In physics there is a long tradition of taking a pragmatic perspective on physical phenomena, and more specifically of focussing on how certain physical states or processes can constitute \emph{resources}.
	A prominent example is the study of heat engines and the advent of thermodynamics.
	Here, one seeks to determine what work can be achieved given access to a heat bath and systems out of thermal equilibrium, such as a compressed gas.
	It was eventually understood that resources of thermal nonequilibrium could be of informational character. 
	One of the most famous examples is the Szilard engine \cite{Szilard1929}, which uses information about the state of a system in order to perform work.	
	Resource-theoretic thinking has expanded with the development of information theory.
	The pioneering work of Claude Shannon \cite{Shannon1948} is centered around questions regarding the convertibility of communication resources.	
	It is not surprising, therefore, that with the rising prominence of information theory in physics in recent years, the application of resource-theoretic ideas in physics has also been on the rise.  
	
	The most relevant of these efforts for us is the development of quantum information theory. 
	Once it was understood that entangled quantum states constitute a resource for communication tasks, the property of entanglement began to be studied as a resource.
	Specifically, in circumstances wherein quantum communication is expensive while local operations and classical communication (LOCC) are free, one can conceptualize the distinction between entanglement and lack thereof  
		in terms of whether a state can be generated by LOCC.
	Furthermore, one can define an ordering over all states wherein one state is above (and hence more entangled than) another if the first can be converted to the second by LOCC.
	The resource theory of entanglement is now very well developed and continues to find a variety of applications in quantum information theory~\cite{Horodecki2009,Devetak2008}.
	We return to this theory in \Cref{ex:rt_quantum} \ref{ex:rt_entanglement} and other examples throughout the text.

	The success of entanglement theory inspired researchers to study other properties of quantum states and channels as resources relative to a set of operations considered to be free.
	For example, symmetry-breaking states can be characterized as resources relative to symmetric operations \cite{Marvian2013} and states of thermal nonequilibrium can be characterized as 
		resources relative to thermal operations \cite{Brandao2013}. 
	These resource theories have led to surprising and important conceptual insights, such as a generalization of Noether's theorem \cite{Marvian2014} and a refinement of our understanding of the second law of 
		thermodynamics \cite{Brandao2015}.
	Many other examples of the use of resource theories within quantum information theory have followed \cite{Chitambar2018}. 

	In the approach to conceptualizing resources that we have just outlined, a choice of free operations defines a preorder relation on resource objects given by the existence of a conversion between the resources in question 
		by the free operations.
	Measures of resourcefulness can then be defined in terms of monotones, that is, real-valued functions that respect the order relation.
	In earlier works, researchers sometimes took the goal of the resource-theoretic endeavour to be the identification of \emph{the} correct measure of the resource, and worked under the mistaken impression that there was 
		``one measure to rule them all''. 
	However, if the preorder is not a total order, then there will in general be many independent measures.
	In our view, the preorder is the fundamental structure, while any particular measure is typically a coarse-grained and incomplete description thereof.

	One useful tool for characterizating the preorder is an algorithm which can solve the decision problem associated to a particular ordering relation, that is, one which takes as input a description of any pair of 
		resources and outputs whether or not the first can be converted to the second by the free operations. 
	However, even if one has identified such an algorithm, it might still be difficult to answer simple questions about the global structure of the preorder.
	For instance, the following properties of the preorder cannot easily be determined in this way (as noted in~\cite{Wolfe2019}): its height (cardinality of the largest chain), its width (cardinality of the largest antichain), 
		whether or not it is totally ordered (i.e., fails to have any incomparable elements), whether or not it is weak (i.e., the incomparability relation is transitive), and whether or not it is locally finite 
		(i.e., with finite number of inequivalent elements between any two ordered elements).
	Monotones can provide a better route to answering such questions. 
	In order to learn all of these properties, it sometimes suffices to find a few nontrivial monotones~\cite{Wolfe2019}. 
	Another point in favour of looking for monotones arises when these quantify the usefulness of resources for a given task in a direct manner \cite{Takagi2019,Ducuara2019}. 

	Consequently, constructing useful monotones is indeed a critical part of developing a concrete resource theory.
	In this article, we approach the problem of devising schemes for constructing useful monotones within an abstract framework for resource theories, designed to be sufficiently general to capture a large variety of 
		concrete resource theories of interest. 
	The details of abstract framework we use here (see \Cref{sec:ResourceTheory}) are chosen with two desiderata in mind, besides the obvious requirement that the framework should be applicable to the study of resource theories.
	On the one hand, the framework should only contain structure relevant to the monotone-construction-schemes we study.
	That is why we deliberately do not use the language of quantum theory (or process theories more broadly) and we do not assume the existence of convex combination (unless necessary).
	On the other hand, we do not wish to obscure the concepts presented by needless abstraction. 
	A more general framework that aims to broaden the scope of future investigations of resource theories is presented in the PhD thesis of one of the authors \cite{Gonda2021} instead.
	There, one can also find an extensive analysis of the connection of our framework for resource theories presented here to those of ordered commutative monoids~\cite{Coecke2016,Fritz2017} and partitioned process theories~\cite{Coecke2016}.
	
	One of the key simplifications we introduce is that we omit any explicit sequential information about resource interactions.
	For example, in the context of partitioned process theories, we thus model combining resources neither in terms of the parallel composition of processes ($\otimes$) nor in terms of their sequential composition ($\circ$).
	This is because both are \emph{non-commutative} operations and therefore carry sequential information.
	A commutative operation that fits our purposes is one called ``universal combination'' \cite{Coecke2016} and denoted by $\boxtimes$.
	It is called universal because it specifies the combinations of processes in arbitrary fashion.
	Motivated by this connection, we refer to a general commutative binary operation on resources as a \emph{universal combination} and denote it by $\boxtimes$ throughout this article even if it no longer applies to resources that are processes in some process theory.
	
	Among the many process combinations, there are some that generate higher-order processes \cite{wilson2022mathematical}.
	For instance, one may combine two channels, $f$ and $g$, into a higher-order process that implements the following map of channels to channels:
	
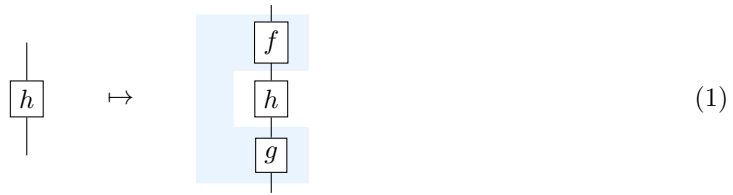
\begin{equation}
		\begin{tikzpicture}
	\begin{pgfonlayer}{nodelayer}
		\node [style=morphism] (0) at (-2.5, 0) {$h$};
		\node [style=none] (1) at (0, 0) {$\mapsto$};
		\node [style=none] (2) at (-2.5, 1.5) {};
		\node [style=none] (3) at (-2.5, -1.5) {};
		\node [style=morphism] (4) at (4, 0) {$h$};
		\node [style=none] (5) at (4, 1) {};
		\node [style=none] (6) at (4, -1) {};
		\node [style=morphism] (7) at (4, 1.5) {$f$};
		\node [style=morphism] (8) at (4, -1.5) {$g$};
		\node [style=none] (9) at (4, 2.5) {};
		\node [style=none] (10) at (4, -2.5) {};
		\node [style=none] (11) at (5, 2.25) {};
		\node [style=none] (12) at (2, 2.25) {};
		\node [style=none] (13) at (2, -2.25) {};
		\node [style=none] (14) at (5, -2.25) {};
		\node [style=none] (15) at (5, -0.75) {};
		\node [style=none] (16) at (3, -0.75) {};
		\node [style=none] (17) at (3, 0.75) {};
		\node [style=none] (18) at (5, 0.75) {};
	\end{pgfonlayer}
	\begin{pgfonlayer}{edgelayer}
		\draw [style=blue fill] (16.center)
			 to (17.center)
			 to (18.center)
			 to (11.center)
			 to (12.center)
			 to (13.center)
			 to (14.center)
			 to (15.center)
			 to cycle;
		\draw (3.center) to (0);
		\draw (0) to (2.center);
		\draw (6.center) to (4);
		\draw (4) to (5.center);
		\draw (10.center) to (8);
		\draw (8) to (6.center);
		\draw (5.center) to (7);
		\draw (7) to (9.center);
	\end{pgfonlayer}
\end{tikzpicture}
	\end{equation}
	The specific higher-order process is a 1-comb in this case and it is indicated by the blue background in the above diagram.
	Universal combination thus naturally leads to many different types of resources. 
	Specifically, even if we are interested in the resources in the form of quantum states (processes with trivial input), we will incorporate not just states, but also channels and quantum combs \cite{Chiribella2008} within the same theory.
	
	The first motivation for such type-independence is that one \emph{is} generally interested in more than one type of resource.
	This becomes explicit if we interpret the \emph{tasks} we want to use the resources for just as another type of a \emph{resource}.
	Whereas many monotones in quantum resource theories are commonly defined for states only, the resourcefulness of channels \cite{Liu2019,liu2020operational,gour2019quantify} and combs is also of interest. 
	For instance, situations where combs play a central role include adaptive strategies in quantum games \cite{Gutoski2007} and algorithms for a measurement-based quantum computer \cite{Raussendorf2003}.
	Resource-theoretic studies of these therefore require a framework that can handle more than just states and channels.

	Furthermore, type-independence is important because it allows one to consider conversions between different types of resources (e.g., see \cite{schmid2020type}).
	In the context of quantum entanglement, for example, one may be interested in understanding the possibilities for interconversion between entangled states and quantum channels (exemplified by the teleportation protocol).
	If one seeks to understand the interplay between states and channels in terms of measures of resourcefulness, it is necessary that these measures be applicable to both types of resources.  
	This motivates our choice that type-independence is built into the abstract framework for resource theories we use---schemes for constructing monotones within such a framework necessarily yield type-independent monotones.
	
	In spite of the attractiveness of resource theories with universal composition, a restriction on the types of resources being considered is often helpful when we aim to answer concrete questions of limited scope.
	For example, if we are only ever interested in transitions between states, it suffices to use a resource theory of states that may be modelled as a partitioned process theory \cite[Section 3.2]{Coecke2016} or a quantale module \cite[Example 3.13]{Gonda2021}, the latter of which is closer to the analysis presented here. 
	The key point 
		is that answers to the questions that we investigate here do not hinge on the assumption of universal composition.
	Indeed, one can develop them in the more general setting of quantale modules \cite{Gonda2021} as well as other approaches.
	The minimal framework we use allows us to describe the results in a clear manner devoid of unnecessary structures, but the ideas can equally well be applied in the context of type-restricted resource theories such as 
		resource theories of states.  
	In particular, any reasonable model of a resource theory that gives rise to a resource order relation can accommodate the ideas we propound here.
	
	We introduce a general scheme for constructing monotones, which we term the \ref{broad scheme} and which applies to any such notion of a resource theory with an order relation 
		among the resources.
	The idea is that the problem of finding monotones for the preordered set of resources can be conceptualized as a two-step procedure that consists of identifying
	\begin{itemize}
		\item an order-preserving map from the set of resources to a \emph{distinct} preordered set, and
		\item a monotone for the latter preorder which is then translated to the preorder of interest by precomposition.
	\end{itemize}
	Details are provided in \Cref{sec:Order theory}.
	We refer to the three elements in this \ref*{broad scheme} as the mediating order-preserving map, the mediating preordered set, and the root monotone.  
	Various monotone constructions then correspond to different choices of these three elements. 
	To cast an existing monotone construction into the mold of the \ref*{broad scheme}, one merely identifies what choices of each of these elements yield the given construction. 
	Alternative constructions can be obtained immediately by simply varying any of the elements.
	In this way, the \ref*{broad scheme} provides a means of classifying the set of all such constructions as well as generalizing the known ones.
	
	\begin{example}[translating measures of nonuniformity to measures of entanglement]
	\label{ex:Nielsen's Theorem}
		As was shown by Nielsen in \cite{Nielsen1999}, entanglement properties are connected to majorization that is well-understood \cite{Marshall1979}.
		Specifically, the ordering of pure quantum states in the resource theory of entanglement (see \Cref{ex:rt_quantum} \ref{ex:rt_entanglement}) is equivalent to the (reverse) ordering of probability distributions by majorization. 
		Majorization monotones are given by Schur-convex functions.
		It follows that all such measures of ``nonuniformity'' of distributions (see \Cref{ex:rt_classical} \ref{ex:rt_nonuniformity}) can be translated to measures of entanglement via the mediating map found by Nielsen, as pointed out in \cite{Coecke2010} where the mediating map was termed a qualitative measure of entanglement.
	\end{example}
	\begin{example}[translating measures of distinguishability to measures of asymmetry]
	\label{ex:Bridge Lemma}
		In Lemma 1 of reference \cite{Marvian2013}, it is shown that there is a mediating order isomorphism from the partial order of quantum states, ordered with respect to convertibility under covariant operations (see \Cref{ex:rt_quantum} \ref{ex:rt_asymmetry}), to the partial order of tuples of quantum states (specifically, the orbits of a quantum state under the symmetry group), ordered with respect to convertibility under arbitrary operations.
		In this way, one can translate (root) monotones from the resource theory of distinguishability (see \Cref{sec:Monotones from Information Theory}) to the resource theory of asymmetry.
	\end{example}

	The art of constructing useful monotones, therefore, is not simply the art of finding specific order-preserving functions for the resource theory of interest.  
	It is also the art of identifying (mediating) order-preserving maps from the theory under investigation to a distinct preorder for which the problem of finding good monotones is easier, or for which many good monotones are 
		already known. 

	The \ref*{broad scheme} provides a useful perspective on monotone constructions based on resource yield (see Section~\ref{sec:yield and cost for subset}).
	Yield constructions for some given function do not specify the value of the function on the resource itself, but rather find the highest value that the function can take among all resources achievable from the resource in question 
		by the free operations.
	The perspective is that there is a mediating order isomorphism between the partial order of resources and the partial order of free images of resources, where the free image of a resource is the set of all resources that can be 
		obtained from it for free. 
	Root monotones can then be easily defined on the latter partial order, for instance in terms of the supremum of the given function over these sets. 
	A similar interpretation is possible for the dual notion of a cost construction.
	This perspective suggests many natural generalizations of such cost and yield constructions, which are the subject of \Cref{sec:yield and cost} of this article.

	If the mediating preorder itself arises from some resource theory, such as in \Cref{ex:Nielsen's Theorem}, then the \ref*{broad scheme} describes a method for translating monotones from one resource theory to another. 
	We explore similar situations in greater detail in \Cref{sec:Translating Monotones}.
	In particular, in \Cref{sec:Monotones from Information Theory}, we consider schemes wherein the mediating preordered set describes tuples of resources from the original resource theory and their convertibility 
		under processes that act identically on each element of the tuple.
	Such a mediating preorder can be understood in terms of a resource theory of distinguishability \cite{Wang2019}. 
	The order-preserving map from the partial order of asymmetric states to the partial order of group orbits of these states, described in Example \ref{ex:Bridge Lemma}, is an instance of this method. 
	Among other instances of this scheme we show how standard monotone constructions, such as weight, robustness and relative entropy based measures, can be reconceived in this way, and this recasting leads us to propose natural generalizations thereof. 
	
	An overview of the most directly applicable ideas for constructing monotones from \Cref{sec:yield and cost,sec:Translating Monotones} is presented as a table in \Cref{sec:monotone_overview}.
	Given basic familiarity with resource theories and some of the notation introduced in \Cref{sec:ResourceTheory}, the table can be also understood prior to reading \Cref{sec:yield and cost,sec:Translating Monotones}.

	Finally, in \Cref{sec:Ordering monotones}, we introduce an ordering among monotones that captures their relative informativeness about the preorder of resources. 
	Such an ordering relation could conceivably also aid in the project of classifying monotone constructions, and we make some progress toward this goal by proving 
		\Cref{thm:yield and cost from more informative functions}, which asserts that ``more informative functions generate more informative monotones''.
	\subsection*{Acknowledgements}	
		
		We would like to thank Elie Wolfe, Denis Rosset, Tobias Fritz, and Bob Coecke for useful discussions.
		This research was supported by a Discovery grant of the Natural Sciences and Engineering Research Council of Canada and by the Perimeter Institute for Theoretical Physics.
		Research at Perimeter Institute is supported in part by the Government of Canada through the Department of Innovation, Science and Economic Development Canada and by the Province of Ontario through the Ministry of Colleges and Universities.
	\section{General Resource Theories}
\label{sec:ResourceTheory}
	
	Before presenting our working abstract definition of a resource theory, we outline the features we would like it to have. 
	First of all, a resource theory should describe a collection $R$ of resources---the objects of study.
	Secondly, it should describe a way of combining these together, which we model by an associative, binary operation $\boxtimes$ and call the \textbf{universal combination}. 
	Finally, it must incorporate a structure that specifies what conversions between the resources are possible and under what conditions.
	This is achieved by specifying a subset $R_{\rm free}$ of resources deemed to be \emph{free}.
		
	\begin{figure}[tbh!]
		\centering
		\begin{tikzpicture}
	\begin{pgfonlayer}{nodelayer}
		\node [style=morphism] (100) at (-5.25, 0) {$f$};
		\node [style=none] (101) at (-5.25, 1.25) {};
		\node [style=none] (102) at (-5.25, -1.25) {};
		\node [style=none] (103) at (-3.75, 0) {$\boxtimes$};
		\node [style=morphism] (104) at (-2.25, 0) {$g$};
		\node [style=none] (105) at (-2.25, 1.25) {};
		\node [style=none] (106) at (-2.25, -1.25) {};
		\node [style=none] (107) at (0, 0) {$=$};
		\node [style=morphism] (108) at (2.25, 0) {$f$};
		\node [style=none] (109) at (2.25, 1.25) {};
		\node [style=none] (110) at (2.25, -1.25) {};
		\node [style=morphism] (112) at (3.75, 0) {$g$};
		\node [style=none] (113) at (3.75, 1.25) {};
		\node [style=none] (114) at (3.75, -1.25) {};
		\node [style=none] (115) at (1.5, 1) {};
		\node [style=none] (116) at (1.5, -1) {};
		\node [style=none] (117) at (4.5, -1) {};
		\node [style=none] (118) at (4.5, 1) {};
		\node [style=none] (119) at (5.75, 0) {$\cup$};
		\node [style=morphism] (120) at (8, 0.75) {$f$};
		\node [style=none] (121) at (8, 2) {};
		\node [style=morphism] (123) at (8, -0.75) {$g$};
		\node [style=none] (124) at (8, 0) {};
		\node [style=none] (125) at (8, -2) {};
		\node [style=none] (126) at (7, 1.75) {};
		\node [style=none] (127) at (7, -1.75) {};
		\node [style=none] (128) at (9, -1.75) {};
		\node [style=none] (129) at (9, 1.75) {};
		\node [style=none] (130) at (10.25, 0) {$\cup$};
		\node [style=morphism] (131) at (12.5, 0.75) {$g$};
		\node [style=none] (132) at (12.5, 2) {};
		\node [style=morphism] (133) at (12.5, -0.75) {$f$};
		\node [style=none] (134) at (12.5, 0) {};
		\node [style=none] (135) at (12.5, -2) {};
		\node [style=none] (136) at (11.5, 1.75) {};
		\node [style=none] (137) at (11.5, -1.75) {};
		\node [style=none] (138) at (13.5, -1.75) {};
		\node [style=none] (139) at (13.5, 1.75) {};
		\node [style=none] (141) at (17.5, 0.5) {};
		\node [style=none] (142) at (17.5, -0.5) {};
		\node [style=morphism] (143) at (17.5, 1.5) {$f$};
		\node [style=morphism] (144) at (17.5, -1.5) {$g$};
		\node [style=none] (145) at (17.5, 2.5) {};
		\node [style=none] (146) at (17.5, -2.5) {};
		\node [style=none] (147) at (18.5, 2.25) {};
		\node [style=none] (148) at (16, 2.25) {};
		\node [style=none] (149) at (16, -2.25) {};
		\node [style=none] (150) at (18.5, -2.25) {};
		\node [style=none] (151) at (18.5, -0.75) {};
		\node [style=none] (152) at (16.75, -0.75) {};
		\node [style=none] (153) at (16.75, 0.75) {};
		\node [style=none] (154) at (18.5, 0.75) {};
		\node [style=none] (155) at (14.75, 0) {$\cup$};
		\node [style=none] (156) at (22.5, 0.5) {};
		\node [style=none] (157) at (22.5, -0.5) {};
		\node [style=morphism] (158) at (22.5, 1.5) {$g$};
		\node [style=morphism] (159) at (22.5, -1.5) {$f$};
		\node [style=none] (160) at (22.5, 2.5) {};
		\node [style=none] (161) at (22.5, -2.5) {};
		\node [style=none] (162) at (23.5, 2.25) {};
		\node [style=none] (163) at (21, 2.25) {};
		\node [style=none] (164) at (21, -2.25) {};
		\node [style=none] (165) at (23.5, -2.25) {};
		\node [style=none] (166) at (23.5, -0.75) {};
		\node [style=none] (167) at (21.75, -0.75) {};
		\node [style=none] (168) at (21.75, 0.75) {};
		\node [style=none] (169) at (23.5, 0.75) {};
		\node [style=none] (170) at (19.75, 0) {$\cup$};
		\node [style=none] (171) at (-6, 1) {};
		\node [style=none] (172) at (-6, -1) {};
		\node [style=none] (173) at (-4.5, -1) {};
		\node [style=none] (174) at (-4.5, 1) {};
		\node [style=none] (175) at (-3, 1) {};
		\node [style=none] (176) at (-3, -1) {};
		\node [style=none] (177) at (-1.5, -1) {};
		\node [style=none] (178) at (-1.5, 1) {};
	\end{pgfonlayer}
	\begin{pgfonlayer}{edgelayer}
		\draw [style=blue fill] (175.center)
			 to (176.center)
			 to (177.center)
			 to (178.center)
			 to cycle;
		\draw [style=blue fill] (173.center)
			 to (174.center)
			 to (171.center)
			 to (172.center)
			 to cycle;
		\draw [style=blue fill] (164.center)
			 to (165.center)
			 to (166.center)
			 to (167.center)
			 to (168.center)
			 to (169.center)
			 to (162.center)
			 to (163.center)
			 to cycle;
		\draw [style=blue fill] (138.center)
			 to (139.center)
			 to (136.center)
			 to (137.center)
			 to cycle;
		\draw [style=blue fill] (115.center)
			 to (116.center)
			 to (117.center)
			 to (118.center)
			 to cycle;
		\draw (102.center) to (100);
		\draw (100) to (101.center);
		\draw (106.center) to (104);
		\draw (104) to (105.center);
		\draw (110.center) to (108);
		\draw (108) to (109.center);
		\draw (114.center) to (112);
		\draw (112) to (113.center);
		\draw [style=blue fill] (126.center)
			 to (127.center)
			 to (128.center)
			 to (129.center)
			 to cycle;
		\draw (120) to (121.center);
		\draw (125.center) to (123);
		\draw (123) to (124.center);
		\draw (124.center) to (120);
		\draw (131) to (132.center);
		\draw (135.center) to (133);
		\draw (133) to (134.center);
		\draw (134.center) to (131);
		\draw [style=blue fill] (152.center)
			 to (153.center)
			 to (154.center)
			 to (147.center)
			 to (148.center)
			 to (149.center)
			 to (150.center)
			 to (151.center)
			 to cycle;
		\draw (146.center) to (144);
		\draw (144) to (142.center);
		\draw (141.center) to (143);
		\draw (143) to (145.center);
		\draw (161.center) to (159);
		\draw (159) to (157.center);
		\draw (156.center) to (158);
		\draw (158) to (160.center);
	\end{pgfonlayer}
\end{tikzpicture}
		\caption{Universal combination of two processes in a process theory.
			For simplicity, the input and output wires of $f$ and $g$ are of the same type here.
			We can read the right-hand side as a set of five diagrams, the $\cup$ symbol indicates union of sets and each diagram in the union is thought of as a singleton set of diagrams.
			The light blue background indicates the type of the resulting process---a channel with two inputs and two outputs in the first case, a channel with a single input and single output in the next two and a 1-comb for the last two diagrams.}
		\label{fig:Universal composition 2}
	\end{figure}
	\begin{example}[resource theory of universally combinable processes]
		\label{ex:RT of states 1}
		In a resource theory of universally combinable processes, defined in \cite{Coecke2016}, $R$ corresponds to the set of all processes of interest and $\boxtimes$ describes the combination of these 
			via wirings.
		In Figure~\ref{fig:Universal composition 2}, we illustrate the universal combination of two processes, $f$ and $g$ (see Example~\ref{ex:Star operation for universally combinable processes} for more details).
		A concrete resource theory is then given by specifying a set of free processes $R_{\rm free}$ that is closed under the $\boxtimes$ operation.
		Recall that there are various types of processes: 
		States are processes with trivial inputs, channels are processes that have an input and an output, and combs are processes with many inputs and many outputs, interleaved over `time' (such as the last two elements 
			in Figure~\ref{fig:Universal composition 2}). 
		In a theory of universally combinable processes, the set of resources $R$ in general includes all of these types of processes. 
		Similarly, the set of \emph{free} resources $R_{\rm free}$ may include representatives of each type.
		
		Many concrete \emph{quantum} resource theories have been studied in recent years~\cite{Chitambar2018}. 
		For each of these, one can hope to define a choice of $R_{\rm free}$ which corresponds to the free operations that define the resource theory.
		Note, however, that most previous work did not take universal combination as the relevant notion of composition and often took the resources of interest to be processes of a specific type (most commonly states).
		Processes that model the conversions between resources are then typically of a different type (most commonly channels).
		The framework we use here presumes universal combination and therefore does not require any distinction to be made a priori between the resources and the processes achieving conversions---each object is 
			conceptualized \emph{both} as a resource in its own right \emph{and} as something that can be used to achieve a conversion of one resource into another.
	\end{example}

	Our philosophy is that universal combination is a natural notion of combination in a resource theory, which is why we introduce a binary operation $\boxtimes$ that satisfies properties of universal combination.
	However, concrete resource theories that have been studied previously typically consider only specific types of resources and are therefore \emph{not} universally combinable in the sense of \Cref{def:resource theory} below.
	Rather, they are instances of frameworks in which there \emph{is} a conceptual distinction between the objects that constitute resources and those that constitute transformations.
	Such frameworks include resource theories of states arising from partitioned process theories \cite{Coecke2016} as well as resource theories defined through quantale modules \cite{Gonda2021}.
	Nonetheless, such concrete type-specific resource theories can in general be embedded into a universally combinable one. 
	Specifically, the embedding makes sense if (in the case of a resource theory of states) both 
	\begin{itemize}
		\item the set of free operations is the restriction of a set of free resources closed under some notion of universal combination to channels, and
		\item the set of free states is the restriction of the same set of free resources closed under the same notion of universal combination to states.
	\end{itemize}
	In such cases, the definition of the concrete type-specific resource theory is judged to be consistent with our approach. 
	This notion of embedding implies that results derived in the framework of universally combinable resource theories can be carried over to address questions in concrete type-specific resource theories. 
	See \cite[Chapter 4]{Gonda2021} for more details on these generalizations.

	Generically, processes can be wired together in multiple ways (see \Cref{fig:Universal composition 2}).
	In order to capture this feature, the object $r \boxtimes s$ that represents the universal combination of resources $r$ and $s$ is not another resource.
	Rather, it is a \emph{collection} of resources each describing a particular way of combining $r$ and $s$.
	We can thus view $r \boxtimes s$ as an element of $\mathcal{P}(R)$, the power set of the underlying set $R$ of all individual resources. 
	
	The interpretation of $r \boxtimes s$ with respect to the individual resources $r$ and $s$ is that an agent has access to \emph{both} $r$ and $s$ (and they can be combined in various ways).
	On the other hand, the union of $\{r\}$ and $\{s\}$---the set $\{r,s\} \in \mathcal{P}(R)$---represents an agent having access to \emph{either} resource $r$ or resource $s$, but not both.
	Because of this interpretation, we require that the $\boxtimes$ operation distributes over unions, just like conjunction distributes over disjunction.
	That is, for any two sets of resources $S,T \in \mathcal{P}(R)$, we have
	\begin{equation}
		\label{eq:Combination of sets of resources}
		S \boxtimes T = \bigcup_{s\in S, t \in T} s \boxtimes t,
	\end{equation}
	where $s \boxtimes t$ is a shorthand notation for $\{s\} \boxtimes \{t\}$.
	As mentioned beforehand, in this article $s \boxtimes t$ corresponds to \emph{all} valid compositions of $s$ with $t$, whence $s \boxtimes t = t \boxtimes s$ for all $s,t \in R$.
	That is, universal combination is \emph{commutative}. 
	Furthermore, we assume that there is a \textbf{neutral set} of resources denoted by $0 \in \mathcal{P}(R)$ that satisfies
	\begin{equation}
		0 \boxtimes S = S = S \boxtimes 0
	\end{equation}
	for all $S \in \mathcal{P}(R)$.
	In the resource theories of universally combinable quantum processes from example~\ref{ex:RT of states 1}, the neutral set $0$ consists of all identity processes.
	More generally, we would interpret the neutral set as consisting of resources that cannot be used for non-trivial conversions.

	Notice that the neutral set $0$ differs from the empty set $\emptyset \in \mathcal{P}(R)$, which satisfies
	\begin{equation}
		\emptyset \boxtimes S = \emptyset = S \boxtimes \emptyset
	\end{equation}
	and which is the bottom element of the complete Boolean lattice $(\mathcal{P}(R), \cup)$.
	If $r \boxtimes s = \emptyset$ holds, then the interpretation is that the resources $r$ and $s$ are mutually incompatible---there is no way to combine them.
	For example, in the context of deterministic computation, if $r$ and $s$ denote two states of the same register, then they cannot coexist and therefore they cannot be combined.
	On the other hand, if $r \boxtimes s \subseteq 0$ holds, then we would say that any way to combine $r$ and $s$ produces a resource in the neutral set, thus effectively discarding them.
	
	Altogether, we get a commutative monoid $(\mathcal{P}(R), \boxtimes, 0)$ with a monoidal operation $\boxtimes$ that distributes over unions.
	A natural way to understand this structure is in terms of a quantale \cite{rosenthal1990quantales}.
	However, in this article we avoid using the language of quantales since it is not necessary as far as the results presented here are concerned.
	
	When defining a resource theory, we 
		also identify a distinguished subset of resources that are \emph{free}, denoted by $R_{\rm free}$.
	These resources are free in the sense that one can access them in unlimited supply without restrictions, whence we impose that combining free resources together cannot yield a non-free resource.
	That is, we require that
	\begin{equation}
		\label{eq:Compatibility of free set and star 3}
		R_{\rm free} \boxtimes R_{\rm free} \subseteq R_{\rm free}.
	\end{equation}
	holds, leading to the following definition.
	\begin{definition}
		\label{def:resource theory}
		A \textbf{universally combinable resource theory}\footnotemark{} 
			(in this article also referred to as a \textbf{resource theory} for short) $\mathcal{R} = \RT{R}$ consists of a set of resources $R$, a subset of free resources $R_{\rm free} \subseteq R$ and a binary operation $\boxtimes  \colon \mathcal{P}(R) \times \mathcal{P}(R) \to \mathcal{P}(R)$ such that
	\footnotetext{Following \cite{Gonda2021}, we use the terminology ``universally combinable resource theory''  in order to distinguish resource theories considered here from a broader class of resource theories (cf.\ Section 3.3 of \cite{Gonda2021}).}%
		\begin{enumerate}
			\item $(\mathcal{P}(R), \boxtimes, 0)$ is a commutative monoid with a submonoid $(\mathcal{P}(R_{\rm free}), \boxtimes, 0)$, and
			
			\item the operation $\boxtimes$ distributes over the union operation $\cup$ as described by equation~(\ref{eq:Combination of sets of resources}).
		\end{enumerate}
	\end{definition}
	Since in \Cref{def:resource theory} we impose that neutral resources are free (i.e., $0 \subseteq R_{\rm free}$) we can express the closure property \eqref{eq:Compatibility of free set and star 3} as $R_{\rm free} \boxtimes R_{\rm free} = R_{\rm free}$.

	The order relation between resources is induced by the choice of the submonoid of free resources.
	The allowed conversions are those that arise via a composition with elements of the free set $R_{\rm free}$.
	Given a resource theory $\RT{R}$, we define the \textbf{resource ordering}, denoted $\succeq$, by
	\begin{equation}
		\label{eq:Order of resources}
		r \succeq s  \quad \iff\quad  s \in R_{\rm free} \boxtimes r
	\end{equation}
	for any $r, s \in R$, where $R_{\rm free} \boxtimes r$ is a shorthand notation for $R_{\rm free} \boxtimes \{r\}$.
	The order relation captures whether $r$ can be converted to $s$ by means of composition with free resources.
	It can be used to determine the value of resources with respect to the choice of the partition of $R$ into free and non-free resources.
	If $r$ can be converted to $s$ for free, i.e., if $r \succeq s$ holds, then we say that $r$ is better than (or equivalent to) $s$ as a resource in the resource theory $\RT{R}$.
	With this order relation, the set of resources becomes a preordered set $\ordres$.
	
	Similarly, we can define the ordering of \emph{sets} of resources by
	\begin{equation}
		\label{eq:Order of sets of resources}
		S \succeq T \quad \iff\quad T \subseteq R_{\rm free} \boxtimes S
	\end{equation}
	for any $S,T \in \mathcal{P}(R)$.
	Again, $(\mathcal{P}(R), \succeq)$ is a preordered set. 
	
	In the framework for resource theories in terms of ordered commutative monoids \cite{Fritz2017}, one requires a compatibility between the order relation and the monoidal operation as one of the axioms.
	Here, we can derive a corresponding property from Definition~\ref{def:resource theory}.
	\begin{lemma}[compatibility of $\boxtimes$ and $\succeq$]
		\label{thm:Compatibility of star and order}
		Let $\RT{R}$ be a resource theory with the corresponding order relation $\succeq$ defined as in~\eqref{eq:Order of sets of resources}.
		For any three subsets $S,T,U$ of $R$, we have
		\begin{equation}
			S \succeq T \implies S \boxtimes U \succeq T \boxtimes U.
		\end{equation}
	\end{lemma}
	
	\begin{proof}
		By the definition of $\succeq$, we have $S \succeq T \iff R_{\rm free} \boxtimes S \supseteq T $, which implies 
		\begin{equation}
			(R_{\rm free} \boxtimes S) \boxtimes U \supseteq T \boxtimes U.
		\end{equation}
		Via the associativity of $\boxtimes$, we can then conclude that $S \boxtimes U \succeq T \boxtimes U$ must hold whenever $S$ is above $T$ according to the order relation $\succeq$.
	\end{proof}
	
	With respect to the preorder $\ordres$, two resources $r$ and $s$ are said to be \emph{equivalent} (denoted $r \sim s$) if both $r \succeq s$ and $s \succeq r$ hold. 
	Similarly, two sets of resources $S$ and $T$ are said to be equivalent if both $S \succeq T$ and $T \succeq S$ are true, which we denote by $S \sim T$.
	However, as we will see later, there are multiple relevant order relations on $\mathcal{P}(R)$, each of which defines a distinct notion of equivalence for sets of resources.
	One should therefore be careful to attach the right interpretation to $S \sim T$, depending on the order relation used.

	Equivalent resources can be freely converted one to another.
	Therefore, from the point of view of resource convertibility, there is no need to distinguish them unless distinguishing them provides a more convenient representation.
	When we remove this degeneracy we obtain a quotient resource theory $\mathcal{R} / {\sim}$ wherein the resources are equivalence classes of resources in $\mathcal{R}$.
	The quotient need not be a universally combinable resource theory in the strict sense of \Cref{def:resource theory}.
	It does, however, fit into a more general framework in terms of quantale modules \cite{Gonda2021}, which we can think of as replacing the complete Boolean lattice $(\mathcal{P}(R), \cup)$ in   \Cref{def:resource theory} with an arbitrary suplattice.
		
	\subsection{Examples of General Resource Theories}
	\label{sec:Examples of General Resource Theories}
		
		A resource theory satisfying \Cref{def:resource theory} embodies the idea of resources without arbitrary restrictions on the allowed combinations $\boxtimes$.
		It only has a restriction in the form of the set of free resources that generate the resource ordering, representing a restriction on the capabilities of certain agents or on the abundance of certain resources.
		As such, it can naturally accommodate resource theories of universally combinable processes \cite{Coecke2016}, which are defined in the same spirit.
		
		
		\begin{example}[universally combinable processes]
			\label{ex:Star operation for universally combinable processes}
			In a resource theory of universally combinable processes, the set $R$ corresponds to the set of all diagrams with no loops and for any two processes $f$ and $g$, their universal combination $f \boxtimes g$ is the set of all processes one can obtain by ``wiring'' $f$ and $g$ together as depicted in \Cref{fig:Universal composition 2}.
			In order to understand it better, we can consider processes, not all of whose inputs and outputs coincide---e.g.\ by virtue of having multiple output wires:

			\begin{equation}\label{eq:universal_composition_4_v3}
				\begin{tikzpicture}
	\begin{pgfonlayer}{nodelayer}
		\node [style=none] (96) at (1.5, 3.625) {};
		\node [style=none] (97) at (1.5, 1.625) {};
		\node [style=none] (98) at (5.25, 1.625) {};
		\node [style=none] (99) at (5.25, 3.625) {};
		\node [style=none] (12) at (0, 2.625) {$=$};
		\node [style=none] (22) at (6.5, 2.625) {$\cup$};
		\node [style=morphism] (23) at (9, 1.875) {$\;\; g \;\;$};
		\node [style=none] (24) at (9, 0.875) {};
		\node [style=none] (25) at (8.5, 1.875) {};
		\node [style=none] (26) at (9.5, 1.875) {};
		\node [style=none] (27) at (8.5, 2.375) {};
		\node [style=none] (28) at (9.5, 2.375) {};
		\node [style=morphism] (29) at (8.5, 3.375) {$f$};
		\node [style=none] (30) at (8.5, 4.375) {};
		\node [style=none] (31) at (11.5, 2.625) {$\cup$};
		\node [style=morphism] (40) at (2.25, 2.625) {$f$};
		\node [style=none] (41) at (2.25, 3.875) {};
		\node [style=morphism] (43) at (4, 2.625) {$\;\; g \;\;$};
		\node [style=none] (44) at (4, 1.375) {};
		\node [style=none] (47) at (3.5, 2.625) {};
		\node [style=none] (49) at (3.5, 3.875) {};
		\node [style=none] (51) at (4.5, 2.625) {};
		\node [style=none] (52) at (4.5, 3.875) {};
		\node [style=none] (54) at (2.25, 1.375) {};
		\node [style=none] (55) at (2.25, 1.875) {};
		\node [style=none] (56) at (9.5, 2.875) {};
		\node [style=none] (57) at (9.5, 4.375) {};
		\node [style=none] (58) at (8.5, 2.875) {};
		\node [style=morphism] (59) at (14, 1.875) {$\;\; g \;\;$};
		\node [style=none] (60) at (14, 0.875) {};
		\node [style=none] (61) at (14.5, 1.875) {};
		\node [style=none] (62) at (13.5, 1.875) {};
		\node [style=none] (63) at (14.5, 2.375) {};
		\node [style=none] (64) at (13.5, 2.375) {};
		\node [style=morphism] (65) at (14.5, 3.375) {$f$};
		\node [style=none] (66) at (14.5, 4.375) {};
		\node [style=none] (67) at (13.5, 2.875) {};
		\node [style=none] (68) at (13.5, 4.375) {};
		\node [style=none] (69) at (14.5, 2.875) {};
		\node [style=none] (70) at (1.5, -3.375) {$\cup$};
		\node [style=morphism] (71) at (9, -1.875) {$f$};
		\node [style=none] (72) at (9, -0.875) {};
		\node [style=none] (73) at (9, -2.875) {};
		\node [style=morphism] (74) at (9, -4.875) {$\;\; g \;\;$};
		\node [style=none] (75) at (9, -5.875) {};
		\node [style=none] (76) at (8.5, -4.875) {};
		\node [style=none] (77) at (9.5, -4.875) {};
		\node [style=none] (78) at (8.5, -3.875) {};
		\node [style=none] (79) at (9.5, -3.875) {};
		\node [style=none] (81) at (10.25, -4.125) {};
		\node [style=none] (82) at (7.75, -5.625) {};
		\node [style=none] (83) at (10.25, -5.625) {};
		\node [style=none] (84) at (7.75, -1.125) {};
		\node [style=none] (86) at (10.25, -2.625) {};
		\node [style=none] (87) at (10.25, -1.125) {};
		\node [style=none] (88) at (12.75, 4.125) {};
		\node [style=none] (89) at (12.75, 1.125) {};
		\node [style=none] (90) at (15.25, 1.125) {};
		\node [style=none] (91) at (15.25, 4.125) {};
		\node [style=none] (92) at (7.75, 4.125) {};
		\node [style=none] (93) at (7.75, 1.125) {};
		\node [style=none] (94) at (10.25, 1.125) {};
		\node [style=none] (95) at (10.25, 4.125) {};
		\node [style=morphism] (100) at (-6.25, 2.625) {$f$};
		\node [style=none] (101) at (-6.25, 3.875) {};
		\node [style=none] (102) at (-6.25, 1.375) {};
		\node [style=none] (103) at (-4.75, 2.625) {$\boxtimes$};
		\node [style=morphism] (104) at (-2.75, 2.625) {$\;\;g\;\;$};
		\node [style=none] (106) at (-2.75, 1.375) {};
		\node [style=none] (107) at (-7, 3.625) {};
		\node [style=none] (108) at (-7, 1.625) {};
		\node [style=none] (109) at (-5.5, 1.625) {};
		\node [style=none] (110) at (-5.5, 3.625) {};
		\node [style=none] (111) at (-4, 3.625) {};
		\node [style=none] (112) at (-4, 1.625) {};
		\node [style=none] (113) at (-1.5, 1.625) {};
		\node [style=none] (114) at (-1.5, 3.625) {};
		\node [style=none] (115) at (-3.25, 3.875) {};
		\node [style=none] (116) at (-2.25, 3.875) {};
		\node [style=none] (117) at (-3.25, 2.625) {};
		\node [style=none] (118) at (-2.25, 2.625) {};
		\node [style=morphism] (119) at (4, -2.625) {$\;\; g \;\;$};
		\node [style=none] (120) at (4, -5.125) {};
		\node [style=none] (121) at (4.5, -2.625) {};
		\node [style=none] (122) at (3.5, -2.625) {};
		\node [style=morphism] (125) at (4, -4.125) {$f$};
		\node [style=none] (126) at (4.5, -1.625) {};
		\node [style=none] (128) at (3.5, -1.625) {};
		\node [style=none] (130) at (2.75, -1.875) {};
		\node [style=none] (131) at (2.75, -4.875) {};
		\node [style=none] (132) at (5.25, -4.875) {};
		\node [style=none] (133) at (5.25, -1.875) {};
		\node [style=none] (134) at (16.5, 2.625) {$\cup$};
		\node [style=none] (135) at (6.5, -3.375) {$\cup$};
		\node [style=none] (136) at (11.5, -3.375) {$\cup$};
		\node [style=morphism] (137) at (14, -4.875) {$f$};
		\node [style=none] (138) at (14, -3.875) {};
		\node [style=none] (139) at (14, -5.875) {};
		\node [style=none] (141) at (12.75, -5.625) {};
		\node [style=none] (142) at (15.25, -5.625) {};
		\node [style=none] (143) at (15.25, -4.125) {};
		\node [style=morphism] (144) at (14, -1.875) {$\;\; g \;\;$};
		\node [style=none] (145) at (14, -2.875) {};
		\node [style=none] (146) at (13.5, -1.875) {};
		\node [style=none] (147) at (14.5, -1.875) {};
		\node [style=none] (148) at (13.5, -0.875) {};
		\node [style=none] (149) at (14.5, -0.875) {};
		\node [style=none] (150) at (12.75, -1.125) {};
		\node [style=none] (151) at (15.25, -1.125) {};
		\node [style=none] (153) at (15.25, -2.625) {};
		\node [style=none] (154) at (8.25, -4.125) {};
		\node [style=none] (155) at (8.25, -2.625) {};
		\node [style=none] (156) at (13.25, -2.625) {};
		\node [style=none] (157) at (13.25, -4.125) {};
	\end{pgfonlayer}
	\begin{pgfonlayer}{edgelayer}
		\draw [style=blue fill] (143.center)
			 to (157.center)
			 to (156.center)
			 to (153.center)
			 to (151.center)
			 to (150.center)
			 to (141.center)
			 to (142.center)
			 to cycle;
		\draw [style=blue fill] (82.center)
			 to (83.center)
			 to (81.center)
			 to (154.center)
			 to (155.center)
			 to (86.center)
			 to (87.center)
			 to (84.center)
			 to cycle;
		\draw [style=blue fill] (90.center)
			 to (91.center)
			 to (88.center)
			 to (89.center)
			 to cycle;
		\draw [style=blue fill] (96.center)
			 to (97.center)
			 to (98.center)
			 to (99.center)
			 to cycle;
		\draw [style=blue fill] (92.center)
			 to (93.center)
			 to (94.center)
			 to (95.center)
			 to cycle;
		\draw (24.center) to (23);
		\draw (25.center) to (27.center);
		\draw (26.center) to (28.center);
		\draw (29) to (30.center);
		\draw (40) to (41.center);
		\draw (44.center) to (43);
		\draw (51.center) to (52.center);
		\draw (54.center) to (55.center);
		\draw [in=-90, out=90] (28.center) to (56.center);
		\draw (56.center) to (57.center);
		\draw [in=-90, out=90] (27.center) to (58.center);
		\draw (58.center) to (29);
		\draw (60.center) to (59);
		\draw (61.center) to (63.center);
		\draw (62.center) to (64.center);
		\draw (65) to (66.center);
		\draw [in=-90, out=90] (64.center) to (67.center);
		\draw (67.center) to (68.center);
		\draw [in=-90, out=90] (63.center) to (69.center);
		\draw (69.center) to (65);
		\draw (73.center) to (71);
		\draw (71) to (72.center);
		\draw (75.center) to (74);
		\draw (76.center) to (78.center);
		\draw (77.center) to (79.center);
		\draw [style=blue fill] (111.center)
			 to (112.center)
			 to (113.center)
			 to (114.center)
			 to cycle;
		\draw [style=blue fill] (109.center)
			 to (110.center)
			 to (107.center)
			 to (108.center)
			 to cycle;
		\draw (102.center) to (100);
		\draw (100) to (101.center);
		\draw (106.center) to (104);
		\draw (118.center) to (116.center);
		\draw (117.center) to (115.center);
		\draw (55.center) to (40);
		\draw (47.center) to (49.center);
		\draw [style=blue fill] (132.center)
			 to (133.center)
			 to (130.center)
			 to (131.center)
			 to cycle;
		\draw (120.center) to (125);
		\draw (119) to (125);
		\draw (122.center) to (128.center);
		\draw (121.center) to (126.center);
		\draw (139.center) to (137);
		\draw (137) to (138.center);
		\draw (145.center) to (144);
		\draw (146.center) to (148.center);
		\draw (147.center) to (149.center);
	\end{pgfonlayer}
\end{tikzpicture}
			\end{equation}	
			We can then infer the universal combination $f \boxtimes g \boxtimes h$ of three processes, as long as we also specify the universal composition of a channel and a comb:
			\begin{equation}\label{eq:universal_composition_5}
				\begin{tikzpicture}
	\begin{pgfonlayer}{nodelayer}
		\node [style=none] (32) at (0, 0) {$=$};
		\node [style=morphism] (33) at (-6.5, 0) {$f$};
		\node [style=none] (34) at (-6.5, 1.25) {};
		\node [style=none] (35) at (-6.5, -1.25) {};
		\node [style=none] (36) at (-4.5, 0) {$\boxtimes$};
		\node [style=morphism] (37) at (-2.5, 1.65) {$g$};
		\node [style=none] (38) at (-2.5, 2.9) {};
		\node [style=none] (39) at (-2.5, 0.4) {};
		\node [style=none] (40) at (-7.25, 1) {};
		\node [style=none] (41) at (-7.25, -1) {};
		\node [style=none] (42) at (-5.75, -1) {};
		\node [style=none] (43) at (-5.75, 1) {};
		\node [style=none] (44) at (-3.25, 2.65) {};
		\node [style=none] (45) at (-3.25, 0.65) {};
		\node [style=none] (46) at (-1.75, 0.65) {};
		\node [style=none] (47) at (-1.75, 2.65) {};
		\node [style=morphism] (48) at (-2.5, -1.625) {$h$};
		\node [style=none] (49) at (-2.5, -0.375) {};
		\node [style=none] (50) at (-2.5, -2.875) {};
		\node [style=none] (51) at (-3.25, -0.625) {};
		\node [style=none] (52) at (-3.25, -2.625) {};
		\node [style=none] (53) at (-1.75, -2.625) {};
		\node [style=none] (54) at (-1.75, -0.625) {};
		\node [style=morphism] (55) at (3.5, 1.65) {$ f \boxtimes g $};
		\node [style=none] (56) at (3.5, 2.9) {};
		\node [style=none] (57) at (3.5, 0.4) {};
		\node [style=none] (62) at (2, 2.65) {};
		\node [style=none] (63) at (2, 0.65) {};
		\node [style=none] (68) at (5, 0.65) {};
		\node [style=none] (69) at (5, 2.65) {};
		\node [style=morphism] (70) at (3.5, -1.625) {$h$};
		\node [style=none] (71) at (3.5, -0.375) {};
		\node [style=none] (72) at (3.5, -2.875) {};
		\node [style=none] (73) at (2.75, -0.625) {};
		\node [style=none] (74) at (2.75, -2.625) {};
		\node [style=none] (75) at (4.25, -2.625) {};
		\node [style=none] (76) at (4.25, -0.625) {};
		\node [style=none] (77) at (6.25, 0) {$\cup$};
		\node [style=morphism] (78) at (9, -1.625) {$ f \boxtimes h $};
		\node [style=none] (79) at (9, -2.875) {};
		\node [style=none] (80) at (9, -0.375) {};
		\node [style=none] (81) at (7.5, -2.625) {};
		\node [style=none] (82) at (7.5, -0.625) {};
		\node [style=none] (87) at (10.5, -0.625) {};
		\node [style=none] (88) at (10.5, -2.625) {};
		\node [style=morphism] (89) at (9, 1.65) {$g$};
		\node [style=none] (90) at (9, 0.4) {};
		\node [style=none] (91) at (9, 2.9) {};
		\node [style=none] (92) at (8.25, 0.65) {};
		\node [style=none] (93) at (8.25, 2.65) {};
		\node [style=none] (94) at (9.75, 2.65) {};
		\node [style=none] (95) at (9.75, 0.65) {};
	\end{pgfonlayer}
	\begin{pgfonlayer}{edgelayer}
		\draw [style=blue fill] (81.center)
			 to (88.center)
			 to (87.center)
			 to (82.center)
			 to cycle;
		\draw [style=blue fill] (63.center)
			 to (68.center)
			 to (69.center)
			 to (62.center)
			 to cycle;
		\draw [style=blue fill] (93.center)
			 to (94.center)
			 to (95.center)
			 to (92.center)
			 to cycle;
		\draw [style=blue fill] (54.center)
			 to (51.center)
			 to (52.center)
			 to (53.center)
			 to cycle;
		\draw [style=blue fill] (44.center)
			 to (45.center)
			 to (46.center)
			 to (47.center)
			 to cycle;
		\draw [style=blue fill] (42.center)
			 to (43.center)
			 to (40.center)
			 to (41.center)
			 to cycle;
		\draw (35.center) to (33);
		\draw (33) to (34.center);
		\draw (39.center) to (37);
		\draw (37) to (38.center);
		\draw (50.center) to (48);
		\draw (48) to (49.center);
		\draw (57.center) to (55);
		\draw (55) to (56.center);
		\draw [style=blue fill] (76.center)
			 to (73.center)
			 to (74.center)
			 to (75.center)
			 to cycle;
		\draw (72.center) to (70);
		\draw (70) to (71.center);
		\draw (80.center) to (78);
		\draw (78) to (79.center);
		\draw (91.center) to (89);
		\draw (89) to (90.center);
	\end{pgfonlayer}
\end{tikzpicture}
			\end{equation}
			In \Cref{eq:universal_composition_5}, the box labelled as $f \boxtimes g$, for instance, does not refer to an individual diagram.
			Instead, it is used as a shorthand for the collection of diagrams (both channels and combs) as specified by $f \boxtimes g$ in \Cref{fig:Universal composition 2}.
		\end{example}
		
		As we mentioned before, resource theories studied at present rarely follow the exact structure of a universally combinable resource theory according to \Cref{def:resource theory}.
		\Cref{ex:counterexample} below explicates this fact in a more concrete manner.
		
		\begin{remark}[\Cref{def:resource theory} excludes type-restricted theories]
			\label{ex:counterexample}
			Note that in a resource theory of universally combinable processes, the channel given by the sequential composition $f \circ h \circ g$ \emph{is} an element of $f \boxtimes g \boxtimes h$.
			This fact can be established by a successive application of the equation shown in \Cref{fig:Universal composition 2}. 
			However, if we impose the restriction that the set of resources $R$ only includes channels (and states as a special case) but not combs, then we reach a contradiction with the assumption that $\boxtimes$ is an associative, commutative and binary operation.
			For instance, restricting the right-hand side of \Cref{fig:Universal composition 2} to exclude the last two, higher-order, processes leads one to conclude that $f \circ h \circ g$ \emph{is not} an element of $(f \boxtimes g) \boxtimes h$.
			This is despite the fact that it would be an element of $f \boxtimes (g \boxtimes h)$.
		\end{remark}
		
		Nevertheless, type-restricted resource theories can be modelled similarly if we introduce a distinction between objects representing resources and objects representing their transformations.
		The resulting structure is one of quantale modules. 
		It allows one to translate the results described in this article to a setting closer to the practice of resource theories \cite{Gonda2021}.
		
		This article is inspired by the study of quantum information theory, which is why most of our examples are of quantum resource theories and related ones.
		While we invite the reader to think of examples they are familiar with and use them to understand the concepts presented, we also provide short background material on quantum resource theories now.
		
		\begin{example}[quantum resource theories]\label{ex:rt_quantum}
			Whenever the underlying process theory that generates $R$ is quantum theory \cite{coecke2018picturing}, we speak of a quantum resource theory \cite{Chitambar2018}.
			More precisely, quantum states are density operators---unit trace positive semi-definite operators on a given complex Hilbert space. 
			Pure states are those of rank 1, so that all other states arise as convex mixtures thereof.
			First-order quantum processes are given by completele positive trace preserving (CPTP) maps.
			The input and output types of these can be labelled by the relevant Hilbert spaces and are associated to a concrete physical system.
			Quantum processes thus map quatum states of one system to quantum states of another system.
			We can think of states themselves as processes with trivial input system $I$ which is given by the $1$-dimensional Hilbert space.
			
			There are many quantum resource theories that correspond to distinct choices of the set of free resources $R_{\rm free}$.
			Let us mention two that feature repeatedly in our examples.
			\begin{compactenum}
				\item \label{ex:rt_entanglement} One of the most studied is the resource theory of quantum entanglement \cite{Horodecki2009}. 
					There, each resource type (i.e., a wire label in diagrams as above) specifies two systems of two distinct agents respectively.
					One thus studies the entanglement between these two parties as a property of \emph{bipartite} states (and other processes).
					
					Traditionally, free resources are generated through composition of local operations and classical communications (LOCC) \cite{Chitambar2014}.
					Local operation refers to an arbitrary bipartite quantum process that factorizes into a tensor product of two independent processes for the two parties.
					A communication channel is a process that maps system of one party to the other party.
					It is classical if it is implemented by a (monopartite) channel $\mathcal{E}$ that acts on a standard basis of operators in the ``bra-ket'' notation via
					\begin{equation}
						\mathcal{E} \bigl( \ket{i}\bra{j} \bigr) = \delta_{ij} \sum_{k} \epsilon_{ik} \ket{k} \bra{k} 
					\end{equation}
					or any other one that can be obtained from it by applying arbitrary quantum pre- and post-processing.					
						
				\item \label{ex:rt_asymmetry} 
					Given an unitary action of a group $G$ on each quantum system (i.e., each Hilbert space considered), one can define a quantum resource theory of asymmetry \cite{Marvian2012,Marvian2014}.
					A quantum channel $\mathcal{E} \colon A \to B$ is free in this resource theory if it is covariant with respect to these group actions, i.e., if it satisfies:
					\begin{equation}
						\mathcal{E} \bigl( U_g \rho U_g^{\dagger} \bigr) = V_g \mathcal{E} (\rho) V_g^{\dagger}
					\end{equation}
					for every state $\rho \colon I \to A$ and every $g \in G$, where $U_g$ and $V_g$ are the relevant representations of $G$ on $A$ and $B$ respectively.
					Consequently, free states are the $G$-invariant ones.
			\end{compactenum}
		\end{example}
		
		\begin{example}[classical resource theories]\label{ex:rt_classical}
			Besides quantum resource theories, we also use examples where the physical systems exhibit classical, stochastic, behavior.
			Quantum states are (in the discrete case) replaced by probability distribution and quantum processes by stochastic maps.
			While there are many interesting classical resource theories, in this article we only refer to the following ones.
			\begin{compactenum}
				\item \label{ex:rt_athermality} Classical resource theory of athermality \cite{Janzing2000,Brandao2013} describes resources of thermal non-equilibrium. 
					Free states are those which are in thermal equilibrium with a given heat bath at a fixed temperature, i.e., thermal states.
					Free processes can be characterized as those stochastic maps that preserve thermal states. 
					The resulting order of resources can be interpreted as describing how far a given state is from the equilibrium, and thus how useful it is in thermodynamic protocols for the purposes of work extraction for example.
						
				\item \label{ex:rt_nonuniformity} A special case of the above in the case of infinite temperature bears the name of a resource theory of nonuniformity \cite{Gour2015,Horodecki2003}. 
					In the limit of infinite temperature, all thermal states become uniform probability distributions.
					Thus the free processes with identical input and output systems are the doubly stochastic maps, and the resource ordering is the famous majorization preorder \cite{Marshall1979}. 
			\end{compactenum}
		\end{example}
		
	\subsection{Useful Order-Theoretic Notions}
	\label{sec:Order theory}
			
		Note that Lemma~\ref{thm:Compatibility of star and order} can be viewed as saying that the map $S \mapsto S \boxtimes U$ is order-preserving with respect to $(\mathcal{P}(R), \succeq)$.
		Order-preserving functions are the key structure-preserving maps between ordered sets.
		\begin{definition}
			\label{def:order preserving}
			Let $(\mathcal{A}, \succeq_{\mathcal{A}})$ and $(\mathcal{B},\succeq_{\mathcal{B}})$ be two arbitrary preordered sets. 
			A function $M \colon \mathcal{A} \to \mathcal{B}$ is \textbf{order-preserving} if the implication
			\begin{equation}
				a_1 \succeq_{\mathcal{A}} a_2 \implies M(a_1) \succeq_{\mathcal{B}} M(a_2)
			\end{equation}
			holds for all $a_1, a_2 \in \mathcal{A}$.
		\end{definition}
		We can use order-preserving functions to learn about the preordered set of resources $\ordres$.
		One of the most common practices is to find so-called resource monotones.
		They are order-preserving maps from $\ordres$ to the totally ordered set of extended real numbers $\ordreals$, by which we mean the set $\mathbb{R} \cup \{-\infty,\infty\}$ ordered as usual.
		\begin{definition}
			Let $\RT{R}$ be a resource theory.
			A \textbf{resource monotone} (or \textbf{monotone} for short) is a function $M \colon R \to \reals$ such that for all $r, s \in R$, we have
			\begin{equation}
				\label{eq:monotone definition}
				s \in R_{\rm free} \boxtimes r \quad \implies \quad M(r) \geq M(s).
			\end{equation}
			For the preorder $\ordres$ defined by \eqref{eq:Order of resources}, condition~\eqref{eq:monotone definition} can be expressed as \mbox{$r \succeq s \implies  M(r) \geq M(s)$}, so that such an 
				$M$ is indeed an order-preserving map from $\ordres$ to $\ordreals$.
		\end{definition}
		
		\begin{example}\label{ex:monotones}
			Since resource theories are traditionally studied with a particular focus of a certain type of resources, which are most commonly states, examples of monotones from the literature are typically type-specific---their domain consists of only certain resources.
			\begin{compactenum}
				\item \label{ex:monotones_entanglement} In the resource theory of quantum entanglement, there is a number of famous monotones with domain restricted to pure bipartite quantum states $\psi_{AB}$.
					One is the Schmidt rank defined as the rank of its reduced density matrix $\psi_A \coloneqq \Tr_A(\psi_{AB})$ obtained by partial trace.
					The entanglement rank (introduced as Schmidt number in \cite{Terhal2000}) of an arbitrary bipartite quantum state $\rho$ is then defined as the largest integer $k$ such that in any convex decomposition
					\begin{equation}
						\rho = \sum_{j} \lambda_j \psi_j
					\end{equation}
					of $\rho$ in terms of pure states, there is a $j$ such that the Schmidt rank of $\psi_j$ is at least $k$.
						
				\item \label{ex:monotones_nonuniformity} An extensive overview of monotones for states in the classical resource theory of nonuniformity can be found in \cite{Gour2015}.
					For instance, since uniform distributions maximize Shannon entropy $H$, it is not surprising that it appears as a monotone under the name Shannon nonuniformity.
					More precisely Shannon nonuniformity $I_H$ is a monotone defined as
					\begin{equation}
						I_H (p) \coloneqq \ln(d) - H(p)
					\end{equation}
					where $p$ is a probability distribution over sample space with $d$ elements.
			\end{compactenum}
		\end{example}
	
		In this work, we study ways in which resource monotones can be constructed.
		We look at examples of common constructions of monotones appearing in the literature on resource theories and identify more general procedures, which they are instances of.
		This helps us organize various monotones, understand the connections between them, and obtain generally applicable methods for generating new interesting monotones in any resource theory of interest.
		All of the monotone constructions we discuss fall within the following general scheme.
		\begin{named}{Broad Scheme}
			\label{broad scheme}
			We identify a preordered set $(\mathcal{A}, \succeq_{\mathcal{A}})$ and two order-preserving maps $\sigma_1$ and $\sigma_2$:
			\begin{align}
				\label{eq:general}
				\sigma_1 &\colon \ordres \to (\mathcal{A}, \succeq_{\mathcal{A}})  &  \sigma_2 &\colon (\mathcal{A}, \succeq_{\mathcal{A}}) \to \ordreals
			\end{align}
			Composing the two order-preserving maps gives a monotone $\ordres \to \ordreals$.
			$\sigma_2$ is called the \textbf{root monotone}, $(\mathcal{A}, \succeq_{\mathcal{A}})$ is called the \textbf{mediating preordered set} ($\mathcal{A}$ being the mediating set and $\succeq_{\mathcal{A}}$ the mediating preorder) and $\sigma_1$ is called the \textbf{mediating order-preserving map}.
			The target monotone $\sigma_2 \circ \sigma_1$ is said to be pulled back from the monotone $\sigma_2$ through the mediating map $\sigma_1$.
		\end{named}
		Broadly speaking, the aim of this work is thus to illuminate which choices of $(\mathcal{A}, \succeq_{\mathcal{A}})$, $\sigma_1$, and $\sigma_2$ lead to monotones that are either prevalent in the literature or interesting for other reasons.

		A concept that we will find useful is that of downward and upward closed sets.
		We make use of these repeatedly.
		\begin{definition}
			Let $(\mathcal{A}, \succeq_{\mathcal{A}})$ be a preordered set.
			A set $D \subseteq \mathcal{A}$ is \textbf{downward closed} with respect to $\succeq_{\mathcal{A}}$ if for all $a \in \mathcal{A}$ and all $d \in D$ the implication
			\begin{equation}
				d \succeq_{\mathcal{A}} a  \quad \implies \quad  a \in D
			\end{equation}
			holds.
			The set of all downward closed subsets of $\mathcal{A}$ is denoted by $\mathcal{DC}(\mathcal{A})$.

			On the other hand, a set $U \subseteq \mathcal{A}$ is \textbf{upward closed} with respect to $\succeq_{\mathcal{A}}$ if for all $a \in \mathcal{A}$ and all $u \in U$ the implication
			\begin{equation}
				a \succeq_{\mathcal{A}} u \quad \implies \quad a \in U
			\end{equation}
			holds. 
			The set of all upward closed subsets of $\mathcal{A}$ is denoted by $\mathcal{UC}(\mathcal{A})$.
	 	\end{definition}
	 	Note that in a resource theory, a set of resources $D \subseteq R$ is downward closed if and only if $R_{\rm free} \boxtimes D = D$, where the preorder $\succeq$ is the resource ordering defined by (\ref{eq:Order of resources}).
	 		
 		\begin{figure}[!tb]
		\begin{center}
			\begin{subfigure}[b]{.4\textwidth}\centering
				\begin{tikzpicture}[align=center,thick,>=stealth,node distance=\scale and 0.5*\scale,
						resource/.style={circle,fill=black,draw=none},
						freeconv/.style={->,thick}
					]
					\node[resource]	(a)		at (0,0)					{};
					\node[resource]	(1b)		[below left=of a]			{};
					\node[resource]	(b1)		[below right=of a]		{};
					\node[resource]	(2c)		[below left=of 1b]		{};
					\node[resource]	(c)		[below left=of b1]		{};
					\node[resource]	(c2)		[below right=of b1]		{};
					\node[resource]	(3d)		[below left=of 2c]		{};
					\node[resource]	(1d)		[below left=of c]			{};
					\node[resource]	(d1)		[below left=of c2]		{};
					\node[resource]	(d3)		[below right=of c2]		{};
					
					\draw[freeconv]	(a) -- (1b);
					\draw[freeconv]	(a) -- (b1);
					\draw[freeconv]	(1b) -- (2c);
					\draw[freeconv]	(1b) -- (c);
					\draw[freeconv]	(b1) -- (c);
					\draw[freeconv]	(b1) -- (c2);
					\draw[freeconv]	(2c) -- (3d);
					\draw[freeconv]	(2c) -- (1d);
					\draw[freeconv]	(c) -- (1d);
					\draw[freeconv]	(c) -- (d1);
					\draw[freeconv]	(c2) -- (d1);
					\draw[freeconv]	(c2) -- (d3);
					
					\node[text=white]			at (b1)								{$r$};
					
					\begin{pgfonlayer}{background}
						\foreach \nodename in {a,1b,b1,2c,c,c2,3d,1d,d1,d3} {\coordinate (\nodename') at (\nodename);} 
					
						\path[fill=SeaGreen,draw=Emerald,line width=1.17*\scale,line cap=round,line join=round]  (b1') to (1d') to (d3') to (b1') -- cycle;
						\path[fill=SeaGreen,draw=SeaGreen,line width=1.06*\scale,line cap=round,line join=round]  (b1') to (1d') to (d3') to (b1') -- cycle;
						
						\path[fill=Bittersweet!80,draw=Brown,line width=0.94*\scale,line cap=round,line join=round]  (c2') to (d1') to (1d') to (2c') to (3d') to (d3') to (c2') -- cycle;
						\path[fill=Bittersweet!80,draw=Bittersweet!80,line width=0.82*\scale,line cap=round,line join=round]  (c2') to (d1') to (1d') to (2c') to (3d') to (d3') to (c2') -- cycle;
					\end{pgfonlayer}
				\end{tikzpicture}
				\caption{Example of a downward closed set (brown region) and the free image $\down(r)$ of a particular resource $r$ (turquoise) in a simple preordered set.}\label{fig:DC}
			\end{subfigure}\hspace{0.05\textwidth}
			\begin{subfigure}[b]{.4\textwidth}\centering
				\centering
				\begin{tikzpicture}[align=center,thick,>=stealth,node distance=\scale and 0.5*\scale,
						resource/.style={circle,fill=black,draw=none},
						freeconv/.style={->,thick}
					]
					\node[resource]	(a)		at (0,0)					{};
					\node[resource]	(1b)		[below left=of a]			{};
					\node[resource]	(b1)		[below right=of a]		{};
					\node[resource]	(2c)		[below left=of 1b]		{};
					\node[resource]	(c)		[below left=of b1]		{};
					\node[resource]	(c2)		[below right=of b1]		{};
					\node[resource]	(3d)		[below left=of 2c]		{};
					\node[resource]	(1d)		[below left=of c]			{};
					\node[resource]	(d1)		[below left=of c2]		{};
					\node[resource]	(d3)		[below right=of c2]		{};
					
					\draw[freeconv]	(a) -- (1b);
					\draw[freeconv]	(a) -- (b1);
					\draw[freeconv]	(1b) -- (2c);
					\draw[freeconv]	(1b) -- (c);
					\draw[freeconv]	(b1) -- (c);
					\draw[freeconv]	(b1) -- (c2);
					\draw[freeconv]	(2c) -- (3d);
					\draw[freeconv]	(2c) -- (1d);
					\draw[freeconv]	(c) -- (1d);
					\draw[freeconv]	(c) -- (d1);
					\draw[freeconv]	(c2) -- (d1);
					\draw[freeconv]	(c2) -- (d3);
					
					\node[text=white]			at (b1)								{$r$};
					
					\begin{pgfonlayer}{background}
						\foreach \nodename in {a,1b,b1,2c,c,c2,3d,1d,d1,d3} {\coordinate (\nodename') at (\nodename);} 
						
						\path[fill=Bittersweet!80,draw=Brown,line width=1.17*\scale,line cap=round,line join=round]  (c2') to (a') to (2c') to (1b') to (b1') to (c2') -- cycle;
						\path[fill=Bittersweet!80,draw=Bittersweet!80,line width=1.06*\scale,line cap=round,line join=round]  (c2') to (a') to (2c') to (1b') to (b1') to (c2') -- cycle;
						
						\path[fill=SeaGreen,draw=Emerald,line width=0.94*\scale,line cap=round,line join=round]  (b1') to (a') to (b1') -- cycle;
						\path[fill=SeaGreen,draw=SeaGreen,line width=0.82*\scale,line cap=round,line join=round]  (b1') to (a') to (b1') -- cycle;
					\end{pgfonlayer}
				\end{tikzpicture}
				\caption{Example of an upward closed set (brown) and the free preimage $\up(r)$ of a particular resource $r$ (turquoise) in a simple preordered set.}\label{fig:UC}
			\end{subfigure}
		\end{center}
		\end{figure}
	 	
		It follows from the definition above that (finite) unions and intersections of downward closed sets are downward closed, and likewise (finite) unions and intersections of upward closed sets are upward closed.
		There are two canonical examples of downward closed sets in any resource theory $\RT{R}$.
		They are $R$ and $R_{\rm free}$.
		Since both are closed under $\boxtimes$ and contain $R_{\rm free}$, these two sets of resources are always downward closed.
		Moreover, for any set of resources $S$, the set $R_{\rm free} \boxtimes S$ is also downward closed.
		Therefore, we can express $\mathcal{DC}(R)$ as
		\begin{equation}
			\mathcal{DC}(R) = \Set{R_{\rm free} \boxtimes S  \given  S \in \mathcal{P}(R)}.
		\end{equation}
		\begin{example}[downward closed sets]\label{examplesDCsets}
			Downward closed sets naturally appear in the study of many resource theories.
			\begin{compactenum}
				\item \label{ex:entrankDCsets}
					In the resource theory of bipartite quantum entanglement, the set of states with entanglement rank at most $k$ is a downward closed set.
					More precisely, there is a downward closed set (which naturally contains processes of various types, not only states) in a universally combinable resource theory of quantum entanglement, whose intersection
						with the set of all states gives precisely the states with entanglement rank at most $k$.
					These sets for different values of $k$ form a total order under set inclusion.
						
				\item \label{ex:3-wayentangledDCsets} 
					Consider the resource theory of \emph{multipartite} quantum entanglement with the free operations given by the appropriate LOCC processes again.
					A partition of the $m$ parties is said to have radius at most $k$ if each of its elements consists of at most $k$ parties.
					Then, for any given $k$, the set of states that are convex combinations of pure states separable with respect a partition of radius at most $k$ forms a downward closed set.
					In particular, the specific case of $m = 3$ and $k=2$ defines the set of states that are deemed to be \emph{not} intrinsically $3$-way entangled, and so this set is downward closed.
				
				\item \label{ex:symmetricDCsets}
					In a resource theory of $G$-asymmetry, 
						the set of processes covariant w.r.t.\ a subgroup of $G$ is a downward closed set (see Proposition 3 of \cite{Marvian2013}).
					These sets are related by a partial order under set inclusion that is anti-isomorphic to the partial order of the respective subgroups under set inclusion.
					
				\item Generalizing the downsets from \ref{ex:entrankDCsets}, given any monotone $f$, the set of resources with the value of $f$ bounded above by some constant $c$ is a downward closed set.
					For a fixed monotone and varying $c$, these form a total order under set inclusion.
					More generally, given an order-preserving map $(R, \succeq) \to (\mathcal{A}, \succeq_{\mathcal{A}})$, the preimage of any downward closed set in $\mathcal{A}$ is downward closed 
						in $R$.
			\end{compactenum}
		\end{example}
			 	
	 	\begin{definition}
	 		Given a resource theory $\RT{R}$, the \textbf{free image map} $\down \colon R \to \mathcal{DC}(R)$ is defined by
			\begin{equation}
				\label{eq:free image map}
				\down (r) = R_{\rm free} \boxtimes r 
			\end{equation}
			for any $r \in R$.
			Similarly, the \textbf{free preimage map} $\up \colon R \to \mathcal{UC}(R)$ is defined by
			\begin{equation}
				\up (r) = \Set{s \in R  \given  r \in R_{\rm free} \boxtimes s}.
			\end{equation}
			The maps $\down$ and $\up$ can be also extended to act on sets of resources by requiring compatibility with unions.
			That is, we have
			\begin{align}\label{eq:down_on_sets}
				\down (S) &\coloneqq \cup_{s \in S} \down(s)  &  \up (S) &\coloneqq \cup_{s \in S} \up(s)
			\end{align}
			for any $S \in \mathcal{P}(R)$.
			We can then express the fact that a set $D$ is downward closed by $\down (D) = D$, while the fact that a set $U$ is upward closed can be stated as $\up (U) = U$.
			In the language of order theory, we can identify $\down (S)$ as the \emph{downward closure} of $S \in \mathcal{P}(R)$ and $\up (S)$ as the \emph{upward closure} of $S$ with respect to the preordered set 
				$(\mathcal{P}(R), \succeq)$.
		\end{definition}

		Notice that both $\mathcal{DC}(R)$ and $\mathcal{UC}(R)$ have a natural ordering in terms of subset inclusion which makes $\up$ and $\down$ into order-preserving maps.  
		In particular, we have partially ordered sets $(\mathcal{DC}(R), \supseteq)$ and $(\mathcal{UC}(R), \subseteq)$.
		With this choice, both
		\begin{align}
			\down &\colon \ordres \to (\mathcal{DC}(R), \supseteq)  & &\text{and} &  \up &\colon \ordres \to (\mathcal{UC}(R), \subseteq)
		\end{align} 
		are order-preserving, which we show explicitly in \Cref{thm:down and up wrt D are order-preserving}.
	
		Consequently, monotones for the partial orders $(\mathcal{DC}(R), \supseteq)$ and $(\mathcal{UC}(R), \subseteq)$ can be pulled back to monotones for $\ordres$ via the 
			\ref{broad scheme}.
		We investigate such constructions of resource monotones in the following section, where the role of the mediating preordered set is associated with either $(\mathcal{DC}(R), \supseteq)$ or $(\mathcal{UC}(R), \subseteq)$.
	
	\section{Generalized Resource Yield and Generalized Resource Cost}
\label{sec:yield and cost}
	
	Now we are finally in good shape to start answering the main question posed in the abstract.
	Namely, how general are the monotone constructions one finds in the literature on resource theories?
	We have defined a somewhat minimal and abstract framework for resource theories, within which we can investigate this question.
	In this section, we use the structure of downward and upward closed sets to learn about the convertibility of resources through a generalization of yield-like and cost-like monotones.
	
	First of all, in Section~\ref{sec:yield and cost for all}, we describe a fairly trivial way of generating monotones for both posets $(\mathcal{DC}(R), \supseteq )$ and $(\mathcal{UC}(R), \subseteq )$, 
		given an arbitrary real-valued function on $R$.
	These can act as the root monotone in the~\ref{broad scheme} and thus define monotones on $R$ by precomposition with the corresponding mediating order-preserving map.
	Then, in Section~\ref{sec:yield and cost for subset}, we extend this root monotone construction to the case when we are given a real-valued function that is only defined on a subset of all resources.
	Lastly, in Section~\ref{sec:yield and cost wrt D}, we further generalize this construction by identifying other order-preserving maps that can be used instead of $\down$ and $\up$ as the mediating order-preserving map.

	\subsection{Yield and Cost Constructions Given a Function Defined on All Resources}
	\label{sec:yield and cost for all}
	
		Consider a (not necessarily order-preserving) function $f \colon R \to \reals$, and define two functions $\fmax$ and $\fmin$ by
		\begin{align*}
			\fmax \colon \mathcal{DC}(R) &\to \reals  &  \fmin \colon \mathcal{UC}(R) &\to \reals  \\
				S &\mapsto \sup f(S)  &  S &\mapsto \inf f(S)
		\end{align*}
		where $f(S)$ denotes the image of $S$ under $f$.
		As a function from the partially ordered set $(\mathcal{DC}(R), \supseteq)$ to the totally ordered set $\ordreals$, $\fmax$ is clearly order-preserving.
		Similarly, $\fmin$ is an order-preserving map between $(\mathcal{UC}(R), \subseteq)$ and $\ordreals$.
		
		With the maps $\down$ and $\up$ described in the previous section, we can pull $\fmax$ and $\fmin$ back to monotones on $R$.
		In particular, we get real-valued functions on $R$ defined by
		\begin{align}
			\label{eq:yield_basic}
			\yield{f}{} (r) &\coloneqq \fmax \bigl( \down (r) \bigr) = \sup \Set{ f(s)  \given  s \in \down (r) } = \sup \Set{ f(s)  \given  s \in R_{\rm free} \boxtimes r } \\
			\label{eq:cost_basic}
			\cost{f}{} (r) &\coloneqq \fmin \bigl( \up (r) \bigr) = \inf \Set{ f(s)  \given  s \in \up (r) } = \inf \Set{ f(s)  \given  r \in R_{\rm free} \boxtimes s },
		\end{align}
		which are both resource monotones. 
		The $\yield{f}{}$ of $r$ is the largest value of $f$ among the resources that can be obtained from $r$ for free.
		On the other hand, the $\cost{f}{}$ of $r$ is the smallest value of $f$ among the resources one can use to obtain $r$ for free.
		
		\begin{example}[monotones from dimension functions]\label{ex:dim_cost}
			Consider a resource theory of one-way quantum communication, in which free resources include arbitrary pre- and post-processings as well as one-way classical communication channels.
			Define the dimension, $\mathrm{dim}(\Phi)$, of a quantum channel to be the smaller of the Hilbert space dimensions of its input and output respectively.
			The $\cost{\mathrm{dim}}{}(\Phi)$, or \emph{dimension cost} of a channel $\Phi$, is then the smallest dimension of a channel from which $\Phi$ can be obtained by composition with free resources.
			The dimension cost of a channel is upper bounded by its dimension, but in general it can be strictly smaller. 
			Similar dimension cost and dimension yield monotones arise in any resource theory that can be embedded in vector spaces, and thus in any quantum resource theory.
		\end{example}

	\subsection{Yield and Cost Constructions Given a Function Defined on a Subset of Resources}
	\label{sec:yield and cost for subset}
	
		It is often useful to be able to evaluate resources in terms of their cost or yield with respect to a particular set of special resources that one could call a ``gold standard''.
		
		Let $W \subseteq R$ denote a subset of resources, and consider a partial function $f_W \colon R \to \reals$ with domain $W$.
		It is not hard to see that one can accomodate constructions from the previous section to this case by restricting all optimizations to be within $W$.
		Specifically, we can again define functions $\fWmax$ and $\fWmin$ as
		\begin{equation}
			\begin{aligned}
				\label{eq:Root monotones for yield and cost}
				\fWmax \colon \mathcal{DC}(R) &\to \reals  &  \fWmin \colon \mathcal{UC}(R) &\to \reals  \\
					S &\mapsto \sup f_W(S)  &  S &\mapsto \inf f_W(S),
			\end{aligned}
		\end{equation}
		where $\sup \emptyset \coloneqq -\infty$ and $\inf \emptyset \coloneqq \infty$ and $f_W(S)$ denotes the image of $S$ under $f_W$.
		As we show in \Cref{lem:function_extensions},
			both of these optimization maps are order-preserving.
		Therefore, we also get yield and cost monotones on $\ordres$ defined by
		\begin{align}
			\yield{f_W}{} (r) &\coloneqq \fWmax \bigl( \down (r) \bigr)   &  \cost{f_W}{} (r) &\coloneqq \fWmin \bigl( \up (r) \bigr)
		\end{align}
		which can be expressed as
		\begin{align}
			\label{eq:yield_subset}
			\yield{f_W}{} (r) &= \sup \Set{ f_W(s)  \given  s \in \down (r) \cap W } = \sup \Set{ f_W(s)  \given  s \in R_{\rm free} \boxtimes r, \; s \in W } \\
			\label{eq:cost_subset}
			\cost{f_W}{} (r) &= \inf \Set{ f_W(s)  \given  s \in \up (r) \cap W } = \inf \Set{ f_W(s)  \given  r \in R_{\rm free} \boxtimes s ,\; s \in W }.
		\end{align}
		The $\yield{f_W}{}$ of $r$ is now the largest value of $f_W$ among the resources \emph{within $W$} that can be obtained from $r$ for free, while the $\cost{f_W}{}$ of $r$ is the smallest value of $f_W$ among 
			the resources \emph{within $W$} that one can use to obtain $r$ for free.
		Allowing the domain of $f$ to be smaller than $R$ enables us to see many more monotone constructions as special cases of the generalized cost and yield measures.
		\begin{example}[yield and cost with respect to a chain]\label{ex:currencies}
			The ``currencies'' described in~\cite{Kraemer2016} are yield and cost monotones, wherein $W$ is a chain; i.e., a totally ordered set of resources.
			A concrete example of this type---from entanglement theory---is the cost of an entangled state measured in the number of e-bits (i.e., maximally entangled 2-qubit states) needed to produce it.
			It is called the single-shot entanglement cost.  
			In that case, $W$ is the set of $n$-fold tensor products of e-bits for different values of $n$ and $f_W$ just returns the integer $n$.
			Another example---from the classical resource theory of nonuniformity introduced in \Cref{ex:rt_classical} \ref{ex:rt_nonuniformity}---is the single-shot nonuniformity yield\footnotemark{} of a probability distribution defined with respect to $W$ that is the set of sharp states, and $f_W$ is the Shannon nonuniformity from \Cref{ex:monotones} \ref{ex:monotones_nonuniformity}.
			Sharp states defined as distributions of the form
			\begin{equation}
				p = 	\bigl( \nu, \nu, \ldots, \nu, 0, 0, \ldots, 0 \bigr)
			\end{equation}
			play a similar role to the e-bits in the sense that they include, for any fixed sample space, all states with maximal nonuniformity.
			\footnotetext{Both of the examples presented here also have their dual counterparts of course.
			They are called the single-shot entanglement yield and single-shot nonuniformity cost respectively.}%
		\end{example}
		\begin{example}[axiomatic definitions of thermodynamic entropy]
			In the axiomatic approach to thermodynamics \cite{Lieb1999}, Lieb and Yngvason define the canonical entropy $S$---an essentially unique monotone among equilibrium states---as a currency.
			Moreover, central to the study of non-equilibrium states in this context are the monotones $S_{-}$ and $S_{+}$, defined in \cite{Lieb2013} as the $\yield{S}{}$ and $\cost{S}{}$, where the domain of 
				$S$ (i.e., the set $W$ in our notation) is the set of equilibrium states.
			In \cite{Weilenmann2016}, this approach to thermodynamics was directly related to the manifestly resource-theoretic approach \cite{Janzing2000,Brandao2013}, and it was shown that $S$, $S_{-}$, and 
				$S_{+}$ correspond to versions of the Helmholtz free energy introduced in \cite{Horodecki2013}.
		\end{example}
			\begin{example}[entanglement rank as cost]\label{ex:ent_rank_cost}
				In the resource theory of bipartite quantum states with respect to LOCC operations, entanglement rank introduced in \Cref{ex:monotones} \ref{ex:monotones_entanglement} is a resource monotone.
				As explained in detail in \cite{gour2020optimal}, it can be expressed as a $\cost{f_W}{}$ where $W$ is the set of pure bipartite quantum states and $f_W$ is the Schmidt rank. 
			\end{example}
			\begin{example}[yield and cost for nonclassical correlations]\label{ex:qcorrelations}
				One of the resource theories where cost and yield monotones as presented here have been used explicitly is the resource theory of nonclassicality of common-cause boxes \cite{Wolfe2019}.
				The resources are bipartite classical channels (also known as ``boxes'' in this context) represented by a conditional probability distribution $P(X,Y | S,T)$ and often depicted as a process
				
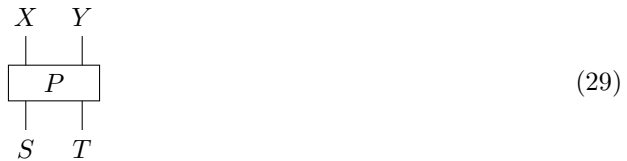
\begin{equation}
					\begin{tikzpicture}
	\begin{pgfonlayer}{nodelayer}
		\node [style=morphism] (0) at (0, 0) {$ \quad P \quad $};
		\node [style=none] (1) at (-0.75, 0) {};
		\node [style=none] (2) at (-0.75, 1.25) {};
		\node [style=none] (3) at (0.75, 1.25) {};
		\node [style=none] (4) at (0.75, 0) {};
		\node [style=none] (5) at (-0.75, -1.25) {};
		\node [style=none] (6) at (0.75, -1.25) {};
		\node [style=none] (7) at (-0.75, -1.75) {$S$};
		\node [style=none] (8) at (0.75, -1.75) {$T$};
		\node [style=none] (9) at (-0.75, 1.75) {$X$};
		\node [style=none] (10) at (0.75, 1.75) {$Y$};
	\end{pgfonlayer}
	\begin{pgfonlayer}{edgelayer}
		\draw (5.center) to (1.center);
		\draw (6.center) to (4.center);
		\draw (1.center) to (2.center);
		\draw (4.center) to (3.center);
	\end{pgfonlayer}
\end{tikzpicture}
				\end{equation}
				and thought of as a stochastic map $S \otimes T \to X \otimes Y$.
				$X$ and $Y$ represent the spaces of outcomes for the two parties, while $S$ and $T$ are the respective settings.
				Moreover, the boxes are required to be non-signalling, so that they can be conceivably interpreted as describing a common-cause relationship between the two parties (with no direct causal influence $S \to Y$ or $T \to X$ mediating their interaction).
				This is the set-up relevant for experiments demonstrating violations of Bell inequalities.
				The free boxes are those that can be explained with a common cause mediated by a classical variable.
				The resource theory then allows us to compare nonclassical behaviors by means of their resourcefulness relative to free operations (see \cite{Wolfe2019} for more details) that respect the causal structure of the Bell scenario
				
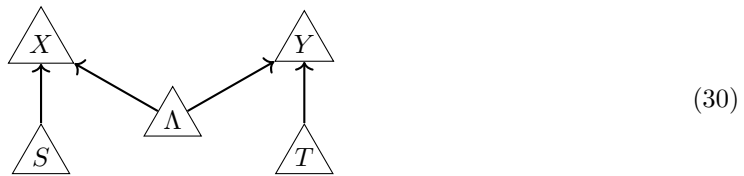
\begin{equation}
					\begin{tikzpicture}
	\begin{pgfonlayer}{nodelayer}
		\node [style={triangle_fix}] (0) at (-3.5, 1.5) {$X$};
		\node [style={triangle_fix}] (1) at (-3.5, -1.5) {$S$};
		\node [style={triangle_fix}] (2) at (3.5, -1.5) {$T$};
		\node [style={triangle_fix}] (3) at (3.5, 1.5) {$Y$};
		\node [style={triangle_fix}] (4) at (0, -0.5) {$\Lambda$};
	\end{pgfonlayer}
	\begin{pgfonlayer}{edgelayer}
		\draw [style={thick_ar}] (4) to (3);
		\draw [style={thick_ar}] (2) to (3);
		\draw [style={thick_ar}] (4) to (0);
		\draw [style={thick_ar}] (1) to (0);
	\end{pgfonlayer}
\end{tikzpicture}
				\end{equation}
				and only use classical variables $\Lambda$ as common causes.
				
				The maximal amount by which a given box violates a Bell inequality is a resource monotone, as it arises via the yield construction. 
				For instance, one may use the violation of the famous Clauser--Horne--Shimony--Holt (CHSH) inequality \cite{clauser1969proposed}.
				This gives a measure, which can be alternatively expressed also as a weight or robustness monotone as well \cite[Corollary 18]{Wolfe2019}.
				While Bell inequalities delineate the set of free resources precisely, thay are insufficient for the characterization of the resource ordering via such yield constructions.
				However, as shown in \cite[Section 7.1]{Wolfe2019}, even a single additional monotone can help uncover a range of properties of the resource ordering inaccessible by considering violations of Bell inequalities alone.
				The specific one used therein uses the cost construction.
				Specifically, the gold standard resources ($W$) are given by the chain of boxes interpolating between a Popescu--Rohrlich (PR) box \cite{popescu1994quantum} at the top of the order and a classical box that is a noisy version thereof. 
				Notably, both of the monotones from \cite{Wolfe2019}, yield relative to CHSH inequality violations and cost with respect to noisy PR boxes, have a closed-form expression as shown in section 6.3 there.
				They may thus aid in discovering explicit formulas for yield and cost monotones in a broader context.
			\end{example}
			\begin{example}[changing the type of resources being evaluated]
				\label{ex:Changing type}
				The general constructions of $\yield{f}{}$ and the $\cost{f}{}$ allow one to extend monotones defined for a particular type of resources to other types.
				A particularly useful example of such a translation is the extension of a monotones for states to monotones for higher-order processes within the same resource theory; such as channels, measurements, or combs.
				For instance, in entanglement theory, one can define a monotone for channels from a monotone for states, such as the cost of implementing a given channel (measured in the terms of number of e-bits used). 
				Let us elaborate on this procedure in a resource theory of universally combinable processes \cite{Coecke2016} as in \Cref{ex:Star operation for universally combinable processes}.
				Let $f \colon W \to \reals$ be a function whose domain $W$ is the set of all states in the process theory.
				Then we can express $\yield{f}{}$ for a particular channel $\phi$ as
				\begin{equation}\label{eq:state_to_channel1}
					\begin{tikzpicture}
	\begin{pgfonlayer}{nodelayer}
		\node [style=morphism] (0) at (-2.25, 0) {$\phi$};
		\node [style=none] (1) at (-2.25, 1) {};
		\node [style=none] (2) at (-2.25, -1) {};
		\node [style=none] (3) at (-1.5, 1.25) {};
		\node [style=none] (4) at (-1.5, -1.25) {};
		\node [style=none] (5) at (-3, 1.25) {};
		\node [style=none] (6) at (-3, -1.25) {};
		\node [style=none] (7) at (-5, 0) {$\yield{f}{}$};
		\node [style=none] (8) at (0, 0) {$=$};
		\node [style=none] (9) at (3.375, 2.5) {};
		\node [style=none] (10) at (3.375, -2.5) {};
		\node [style=none] (11) at (1.75, 0) {$\sup$};
		\node [style=morphism] (12) at (5.875, -0.25) {$\phi$};
		\node [style=none] (15) at (7.875, 2.25) {};
		\node [style=none] (16) at (7.875, -2.25) {};
		\node [style=none] (17) at (4.875, 2.25) {};
		\node [style=none] (18) at (4.875, -2.25) {};
		\node [style=none] (19) at (3.875, 0) {$f$};
		\node [style=wide state] (20) at (6.375, -1.625) {$\phantom{\rho}$};
		\node [style=none] (21) at (6.375, -1.725) {$\rho$};
		\node [style=none] (22) at (5.875, -1.625) {};
		\node [style=none] (23) at (6.875, -1.625) {};
		\node [style=none] (24) at (6.875, 1.375) {};
		\node [style=none] (25) at (5.875, 1.375) {};
		\node [style=none] (26) at (6.375, 2.375) {};
		\node [style=morphism] (27) at (6.375, 1.375) {$\;\; \psi \;\;$};
		\node [style=none] (28) at (14.625, 2.5) {};
		\node [style=none] (29) at (14.625, -2.5) {};
		\node [style=none] (30) at (9.625, 2.25) {};
		\node [style=none] (31) at (9.625, -2.25) {};
		\node [style=none] (32) at (12.625, 1) {$\psi \in R_{\rm free}$};
		\node [style=none] (33) at (12.625, -1) {$\rho \in R_{\rm free}$};
	\end{pgfonlayer}
	\begin{pgfonlayer}{edgelayer}
		\draw (2.center) to (0);
		\draw (0) to (1.center);
		\draw [bend right] (4.center) to (3.center);
		\draw [bend left] (6.center) to (5.center);
		\draw [style=curly brace] (10.center) to (9.center);
		\draw [bend right=15] (16.center) to (15.center);
		\draw [bend left=15] (18.center) to (17.center);
		\draw (22.center) to (12);
		\draw (23.center) to (24.center);
		\draw (12) to (25.center);
		\draw (27) to (26.center);
		\draw [style=curly brace] (28.center) to (29.center);
		\draw [style=dashed box] (31.center) to (30.center);
	\end{pgfonlayer}
\end{tikzpicture}
				\end{equation}
				The argument of $f$ in the optimization on the right-hand side is the most general state which can be obtained from $\phi$ for free.
				If furthermore $f$ is a monotone on its domain so that it satisfies $f(\psi \circ \tau) \leq f(\tau)$ for any state $\tau$ and any free channel $\psi$, we can simplify $\yield{f}{}$ to
				\begin{equation}\label{eq:state_to_channel2}
					\begin{tikzpicture}
	\begin{pgfonlayer}{nodelayer}
		\node [style=morphism] (0) at (-2.25, 0) {$\phi$};
		\node [style=none] (1) at (-2.25, 1) {};
		\node [style=none] (2) at (-2.25, -1) {};
		\node [style=none] (3) at (-1.5, 1.25) {};
		\node [style=none] (4) at (-1.5, -1.25) {};
		\node [style=none] (5) at (-3, 1.25) {};
		\node [style=none] (6) at (-3, -1.25) {};
		\node [style=none] (7) at (-5, 0) {$\yield{f}{}$};
		\node [style=none] (8) at (0, 0) {$=$};
		\node [style=none] (9) at (3.375, 2) {};
		\node [style=none] (10) at (3.375, -2) {};
		\node [style=none] (11) at (1.75, 0) {$\sup$};
		\node [style=morphism] (12) at (5.75, 0.375) {$\phi$};
		\node [style=none] (15) at (7.5, 1.75) {};
		\node [style=none] (16) at (7.5, -1.75) {};
		\node [style=none] (17) at (4.875, 1.75) {};
		\node [style=none] (18) at (4.875, -1.75) {};
		\node [style=none] (19) at (3.875, 0) {$f$};
		\node [style=wide state] (20) at (6.25, -1) {$\phantom{\rho}$};
		\node [style=none] (21) at (6.25, -1.1) {$\rho$};
		\node [style=none] (22) at (5.75, -1) {};
		\node [style=none] (23) at (6.75, -1) {};
		\node [style=none] (26) at (5.75, 1.5) {};
		\node [style=none] (28) at (14.25, 2) {};
		\node [style=none] (29) at (14.25, -2) {};
		\node [style=none] (30) at (9.25, 1.75) {};
		\node [style=none] (31) at (9.25, -1.75) {};
		\node [style=none] (33) at (12, 0) {$\rho \in R_{\rm free}$};
		\node [style=none] (34) at (6.75, 1.5) {};
	\end{pgfonlayer}
	\begin{pgfonlayer}{edgelayer}
		\draw (2.center) to (0);
		\draw (0) to (1.center);
		\draw [bend right] (4.center) to (3.center);
		\draw [bend left] (6.center) to (5.center);
		\draw [style=curly brace] (10.center) to (9.center);
		\draw [bend right=15] (16.center) to (15.center);
		\draw [bend left=15] (18.center) to (17.center);
		\draw (22.center) to (12);
		\draw [style=curly brace] (28.center) to (29.center);
		\draw [style=dashed box] (31.center) to (30.center);
		\draw (23.center) to (34.center);
		\draw (12) to (26.center);
	\end{pgfonlayer}
\end{tikzpicture}
				\end{equation}
				Likewise, we can express the $\cost{f}{}$ as
				\begin{equation}\label{eq:state_to_channel3}
					\begin{tikzpicture}
	\begin{pgfonlayer}{nodelayer}
		\node [style=morphism] (35) at (-2.25, 0) {$\phi$};
		\node [style=none] (36) at (-2.25, 1) {};
		\node [style=none] (37) at (-2.25, -1) {};
		\node [style=none] (38) at (-1.5, 1.25) {};
		\node [style=none] (39) at (-1.5, -1.25) {};
		\node [style=none] (40) at (-3, 1.25) {};
		\node [style=none] (41) at (-3, -1.25) {};
		\node [style=none] (42) at (-4.75, 0) {$\cost{f}{}$};
		\node [style=none] (43) at (0, 0) {$=$};
		\node [style=none] (44) at (1.75, 0) {$\inf$};
		\node [style=none] (45) at (3.5, 2) {};
		\node [style=none] (46) at (3.5, -2) {};
		\node [style=none] (47) at (4.25, 0) {$f$};
		\node [style=none] (49) at (6.75, 1.5) {};
		\node [style=none] (50) at (6.75, -1.5) {};
		\node [style=none] (51) at (5, 1.5) {};
		\node [style=none] (52) at (5, -1.5) {};
		\node [style=state] (53) at (5.875, -0.25) {$\phantom{\rho}$};
		\node [style=none] (54) at (5.875, -0.35) {$\rho$};
		\node [style=none] (58) at (5.875, 1) {};
		\node [style=none] (59) at (8.5, 1.75) {};
		\node [style=none] (60) at (8.5, -1.75) {};
		\node [style=morphism] (61) at (10.25, 0) {$\phi$};
		\node [style=none] (62) at (10.25, 1.5) {};
		\node [style=none] (63) at (10.25, -1.5) {};
		\node [style=none] (64) at (11.75, 0) {$=$};
		\node [style=morphism] (65) at (13.75, 0.5) {$\;\; \psi \;\;$};
		\node [style=state] (66) at (13.25, -1) {$\phantom{\rho}$};
		\node [style=none] (67) at (13.25, -1.1) {$\rho$};
		\node [style=none] (69) at (14.25, -1.5) {};
		\node [style=none] (70) at (13.75, 1.75) {};
		\node [style=none] (71) at (14.25, 0.5) {};
		\node [style=none] (72) at (13.25, 0.5) {};
		\node [style=none] (73) at (18.25, 0) {$\text{and} \;\; \psi \in R_{\rm free}$};
		\node [style=none] (74) at (21, 2) {};
		\node [style=none] (75) at (21, -2) {};
	\end{pgfonlayer}
	\begin{pgfonlayer}{edgelayer}
		\draw (37.center) to (35);
		\draw (35) to (36.center);
		\draw [bend right] (39.center) to (38.center);
		\draw [bend left] (41.center) to (40.center);
		\draw [style=curly brace] (46.center) to (45.center);
		\draw [bend right=15] (50.center) to (49.center);
		\draw [bend left=15] (52.center) to (51.center);
		\draw [style=dashed box] (60.center) to (59.center);
		\draw (63.center) to (61);
		\draw (61) to (62.center);
		\draw (69.center) to (71.center);
		\draw (65) to (70.center);
		\draw [style=curly brace] (74.center) to (75.center);
		\draw (53) to (58.center);
		\draw (66) to (72.center);
	\end{pgfonlayer}
\end{tikzpicture}
				\end{equation}
				These kinds of constructions have appeared in the works on resource theories of quantum channels \cite{Liu2019,liu2020operational,gour2019quantify}, where the proposed monotones for channels are defined via channel divergences (see Example \ref{ex:cost for pairs}).
				However, as Theorem 2 of \cite{gour2019quantify} shows, they can be also equivalently seen as originating from monotones for states (defined via state divergences).
				Channel divergences \cite{Leditzky2018} themselves arise from the generalized yield construction as they are defined by \cref{eq:state_to_channel2} in a resource theory of pairs of resources, i.e., a resource theory of distinguishability \cite{Wang2019}.
				More details on pairs (and other tuples) of resources and in what way they constitute a resource theory can be found in \Cref{sec:Encodings}.
			\end{example}

	\subsection{Yield and Cost Constructions Relative to a Downward Closed Set}
	\label{sec:yield and cost wrt D}

		Apart from varying the root monotone from the \ref{broad scheme} as we have done in Section~\ref{sec:yield and cost for subset}, one can also vary the mediating order-preserving map.
		In particular, $\down$ and $\up$ are not the only order-preserving functions from $\ordres$ to $(\mathcal{DC}(R),\supseteq)$ and $(\mathcal{UC}(R),\subseteq)$ respectively.
		
		\begin{definition}
			\label{def:D-image map}
			Consider a set of resources $D \subseteq R$.
			We define the \textbf{$\bm{D}$-image map} $\downD \colon R \to \mathcal{DC}(R)$ by
			\begin{equation}
				\downD (r) \coloneqq D \boxtimes r 
			\end{equation}
			for any $r \in R$, which can be also written as $\downD (r) = D \boxtimes r$.
			Similarly, the \textbf{$\bm{D}$-preimage map} $\upD \colon R \to \mathcal{UC}(R)$ is defined by
			\begin{equation} 
				\upD (r) \coloneqq \Set{s \in R  \given  r \in D \boxtimes s}.
			\end{equation}
			We can naturally extend them to act on sets of resources as well by compatibility with unions.
			That is, 
			\begin{equation}
				\downD (S) \coloneqq \bigcup_{s \in S} \downD(s)  \quad\text{ and }\quad  \upD (S) \coloneqq \bigcup_{s \in S} \upD(s)
			\end{equation}
			for any $S \in \mathcal{P}(R)$.
		\end{definition}
		It is interesting to note that unlike for $\down$ and $\up$, there is not neccesarily a preorder on $R$ for which the maps $\downD$ and $\upD$ are the downward and upward closure operations respectively.
		
		Composing image and preimage maps can be related to the image (preimage) maps of the composition of the relevant sets of resources as follows.
		\begin{lemma}
			\label{thm:composing image maps}
			Let $\RT{R}$ be a resource theory and let $S$ and $T$ be two subsets of $R$.
			The $S$-image and $T$-image maps from Definition~\ref{def:D-image map} satisfy
			\begin{equation}
				\downwrt{S} \circ \downwrt{T} = \downwrt{S \boxtimes T}.
			\end{equation}
			Similarly, the $S$-preimage and $T$-preimage maps satisfy
			\begin{equation}
				\upwrt{S} \circ \upwrt{T} = \upwrt{T \boxtimes S}.
			\end{equation}
		\end{lemma}
		\begin{proof}
			The first equality follows from the definition of $\downwrt{S}$ and the associativity of $\boxtimes$.
			Namely, recall that we have $\downwrt{S} (U) \coloneqq S \boxtimes U$ for any $U \in \mathcal{P}(R)$.
			
			We can prove the second equality as follows.
			Let $U$ be an arbitrary set of resources.
			If $v$ is an element of $\upwrt{S} \circ \upwrt{T} (U)$, then there must exist a $u \in U$ such that $u \in T \boxtimes w$ holds for some $w \in S \boxtimes v$.
			Therefore, we have $u \in T \boxtimes S \boxtimes v$ and thus $v \in \upwrt{T \boxtimes S} (U)$, so that 
			\begin{equation}
				\upwrt{S} \circ \upwrt{T} (U) \subseteq \upwrt{T \boxtimes S} (U)
			\end{equation}
			holds.
			
			On the other hand, if $v$ is an element of $\upwrt{T \boxtimes S} (U)$, then there exists a $u \in U$ which is also an element of $T \boxtimes S \boxtimes v$. 
			Thus, there is a $w \in S \boxtimes v$ which is also in $\upwrt{T}(u)$, which in turn implies $v \in \upwrt{S} \circ \upwrt{T} (U)$ and
			\begin{equation}
				\upwrt{S} \circ \upwrt{T} (U) \supseteq \upwrt{T \boxtimes S} (U).
			\end{equation}
			Consequently $\upwrt{S} \circ \upwrt{T} (U)$ is equal to $\upwrt{T \boxtimes S} (U)$ for all $U \in \mathcal{P}(R)$.
		\end{proof}
		
		One can show that whenever $D$ is a downward closed set of resources, both $\downD$ and $\upD$ are order-preserving.
		With this notation, we can also see that $\down$ coincides with $\downwrt{R_{\rm free}}$ and $\up$ coincides with $\upwrt{R_{\rm free}}$.	
		\begin{lemma}[mediating maps for the generalized yield and cost constructions]
			\label{thm:down and up wrt D are order-preserving}
			Let $D$ be a downward closed subset of $R$, i.e., $D \in \mathcal{DC}(R)$.
			The two maps, 
			\begin{align}
				\downD \colon \ordres &\to (\mathcal{DC}(R), \supseteq)  &  \upD \colon \ordres &\to (\mathcal{UC}(R), \subseteq),
			\end{align}
			are then both order-preserving.
		\end{lemma}
		\begin{proof}
			If $r,s \in R$ are two resources such that $r \succeq s$, then $s \in R_{\rm free} \boxtimes r = \down (r)$ by definition.
			Furthermore,
			\begin{equation}
				\begin{split}
					s \in \down (r) 	&\,\implies 	\downwrt{D} (s) \subseteq \downwrt{D} \circ \down (r) \\
									&\iff 		\downwrt{D} (s) \subseteq \downwrt{D \boxtimes R_{\rm free}} (r) \\
									&\iff 		\downwrt{D} (s) \subseteq \downwrt{D} (r).
				\end{split}
			\end{equation}
			The first equivalence follows from Lemma~\ref{thm:composing image maps} and the second one uses $D \boxtimes R_{\rm free} = D$.
			Therefore, $\downD$ is order-preserving.
			
			On the other hand, we also have
			\begin{equation}
				s \in \up (r)  \implies  \upwrt{D} (s) \subseteq \upwrt{D} \circ \up (r),
			\end{equation}
			and $\upwrt{D} \circ \up (r) = \upwrt{D} (r)$ by Lemma~\ref*{thm:composing image maps} and the fact that $D$ is downward closed.
			The $D$-preimage map $\upD$ is thus also order-preserving.
		\end{proof}	 
		As a consequence of the order-preserving property of $\downD$ and $\upD$, we have the following theorem.
		\begin{theorem}[generalized yield and cost constructions]
			\label{thm:yield and cost}
			Let $\RT{R}$ be a resource theory, let $D$ be a downward closed subset of $R$, and consider a partial function $f_W \colon R \to \reals$ with 
				domain $W \subseteq R$.
			The $f_W$-yield relative to the $D$-image map, $\yield{f_W}{D} \colon R \to \reals$, and the $f_W$-cost relative to the $D$-preimage map, 
				$\cost{f_W}{D} \colon R \to \reals$, defined as
			\begin{align}
				\yield{f_W}{D} (r) &\coloneqq \fWmax \bigl( \downD (r) \bigr)   &  \cost{f_W}{D} (r) &\coloneqq \fWmin \bigl( \upD (r) \bigr)
			\end{align}
			 are both resource monotones.
		\end{theorem}
		\begin{proof}
			Both $\cost{f_W}{D}$ and $\yield{f_W}{D}$ are constructed as a composition of two order-preserving maps, as proven by \Cref{thm:down and up wrt D are order-preserving} and \Cref{lem:function_extensions}.
			They are therefore order-preserving functions from $\ordres$ to $\ordreals$, i.e., resource monotones.
		\end{proof}
		By unpacking the definitions, we can express the generalized yield and cost monotones as
		\begin{align}
			\label{eq:yield_downset}
			\yield{f_W}{D} (r) &= \sup \Set*[\big]{ f_W(w) \given w \in D \boxtimes r \text{ and } w \in W } \\
			\label{eq:cost_downset}
			\cost{f_W}{D} (r) &=\inf \Set*[\big]{ f_W(w) \given r \in D \boxtimes w \text{ and } w \in W }.
		\end{align}
		The $\yield{f_W}{D}$ of $r$ is the largest value of $f_W$ among the resources within $W$ that can be obtained from $r$ by composing it with a resource in $D$.
		On the other hand, the $\cost{f_W}{D}$ of $r$ is the smallest value of $f_W$ among the resources within $W$ that one can compose with a resource in $D$ and obtain the resource $r$.
		
		There are two main reasons why one might want to use $\yield{f_W}{D}$ (or $\cost{f_W}{D}$) instead of $\yield{f_W}{} \coloneqq \yield{f_W}{R_{\rm free}}$ (or $\cost{f_W}{} \coloneqq \cost{f_W}{R_{\rm free}}$).
		On the one hand, a downward closed set $D$ different from $R_{\rm free}$ can be easier to work with either algebraically or numerically when evaluating the function explicitly.
		This is a common practice in many resource theories in which $R_{\rm free}$ is not straightforward to work with.
		For example, LOCC operations in entanglement theory get replaced by separable operations, noisy operations in nonuniformity theory get replaced by unital operations, and thermal operations in athermality theory 
			get replaced by Gibbs-preserving operations.
		
		On the other hand, $\yield{f_W}{D}$ and $\cost{f_W}{D}$ can give us new interesting monotones distinct from $\yield{f_W}{}$ and $\cost{f_W}{}$. 
		To our knowledge, none of the monotones introduced in the resource theory literature to date are of this kind.
		Here we give a simple toy example of how one could use these constructions for $D \neq R_{\rm free}$ in practice.
		\begin{example}[using $\yield{f_W}{D}$ for $D \neq R_{\rm free}$]
			\label{ex:Advantage of generalized yield}
			Imagine a quantum resource theory in which there are no free states.
			Such resource theories arise naturally when we consider multi-resource theories~\cite{Sparaciari2018} like the resource theory of work and heat~\cite{Sparaciari2017}.
			Since a channel can only be converted to a state by applying it to a state (either a state on the input system of the channel or a state on a larger system which contains it), there is no way to convert a channel to a state 
				for free in this case.
			Therefore, evaluating $\yield{f_W}{}$ for a function $f_W$, defined on states only, would lead to a trivial monotone for channels.
			One would not be able to use this construction to extend monotones for states to monotones for channels.
			However, one can instead use a downward closed set $D$ that does include some states, in which case $\yield{f_W}{D}$ becomes a non-trivial monotone for channels in the resource theory.
			A choice of the set $D$ that is guaranteed to be downward closed and include some states is $R_{\rm free} \boxtimes r$ for a particular state $r \in R$.
			The set $R_{\rm free} \boxtimes r$ can contain more states than just $r$ of course.
			In particular, it contains any other state one could obtain from $r$ for free. 
			An example of this sort of construction is illustrated in \Cref{fig:generalized yield}.
		\end{example}
		\begin{figure}[tbh!]
			\centering
			\begin{tikzpicture}
	\begin{pgfonlayer}{nodelayer}
		\node [style=morphism] (0) at (-2.25, 0) {$\phi$};
		\node [style=none] (1) at (-2.25, 1) {};
		\node [style=none] (2) at (-2.25, -1) {};
		\node [style=none] (3) at (-1.5, 1.25) {};
		\node [style=none] (4) at (-1.5, -1.25) {};
		\node [style=none] (5) at (-3, 1.25) {};
		\node [style=none] (6) at (-3, -1.25) {};
		\node [style=none] (7) at (-5, 0) {$\yield{f}{}$};
		\node [style=none] (8) at (0, 0) {$=$};
		\node [style=none] (9) at (3.375, 2.5) {};
		\node [style=none] (10) at (3.375, -2.5) {};
		\node [style=none] (11) at (1.75, 0) {$\sup$};
		\node [style=morphism] (12) at (5.875, -0.25) {$\phi$};
		\node [style=none] (15) at (7.875, 2.25) {};
		\node [style=none] (16) at (7.875, -2.25) {};
		\node [style=none] (17) at (4.875, 2.25) {};
		\node [style=none] (18) at (4.875, -2.25) {};
		\node [style=none] (19) at (3.875, 0) {$f$};
		\node [style=wide state] (20) at (6.375, -1.625) {$\phantom{\rho}$};
		\node [style=none] (21) at (6.375, -1.725) {$\rho$};
		\node [style=none] (22) at (5.875, -1.625) {};
		\node [style=none] (23) at (6.875, -1.625) {};
		\node [style=none] (24) at (6.875, 1.375) {};
		\node [style=none] (25) at (5.875, 1.375) {};
		\node [style=none] (26) at (6.375, 2.375) {};
		\node [style=morphism] (27) at (6.375, 1.375) {$\;\; \psi \;\;$};
		\node [style=none] (28) at (13.625, 2.5) {};
		\node [style=none] (29) at (13.625, -2.5) {};
		\node [style=none] (30) at (9.625, 2.25) {};
		\node [style=none] (31) at (9.625, -2.25) {};
		\node [style=none] (32) at (12.125, 1) {$\psi \in D$};
		\node [style=none] (33) at (12.125, -1) {$\rho \in D$};
	\end{pgfonlayer}
	\begin{pgfonlayer}{edgelayer}
		\draw (2.center) to (0);
		\draw (0) to (1.center);
		\draw [bend right] (4.center) to (3.center);
		\draw [bend left] (6.center) to (5.center);
		\draw [style=curly brace] (10.center) to (9.center);
		\draw [bend right=15] (16.center) to (15.center);
		\draw [bend left=15] (18.center) to (17.center);
		\draw (22.center) to (12);
		\draw (23.center) to (24.center);
		\draw (12) to (25.center);
		\draw (27) to (26.center);
		\draw [style=curly brace] (28.center) to (29.center);
		\draw [style=dashed box] (31.center) to (30.center);
	\end{pgfonlayer}
\end{tikzpicture}
			\caption{The $f_W$-yield relative to $D$-image map of a channel $\phi$ given a function $f_W$ defined on states only.
			Note that the right hand side is equal to $-\infty$ if $D$ contains no states.}
			\label{fig:generalized yield}
		\end{figure}
		Note that for any set of resources $S$, the set $S \boxtimes R_{\rm free}$ is downward closed.
		It need not be closed under~$\boxtimes$, in which case it is not a candidate for the set of free resources in a resource theory. 
		Nevertheless, we can use $S \boxtimes R_{\rm free}$ in generalized yield and cost constructions (Theorem~\ref{thm:yield and cost}) by defining the image and preimage maps with respect to it.
		If the discarding operation is a free resource (or else if $S$ contains $0$), we may to interpret taking the images and preimages with respect to $S \boxtimes R_{\rm free}$ is as follows.
		They specify what can be achieved by an agent who, in addition to having access to the free resources in unlimited supply, also has access to a resource from $S$.
		Of course, if $S \boxtimes R_{\rm free}$ \emph{is} closed under~$\boxtimes$, then we can think of $S \boxtimes R_{\rm free}$ as describing an agent's access to both $R_{\rm free}$ and 
			$S$ in unlimited supply. 
		In such case $(S \boxtimes R_{\rm free}, R, \boxtimes)$ is a resource theory. 
	
	\section{Translating Monotones Between Resource Theories}
\label{sec:Translating Monotones}

	Let us now change the mediating preordered sets $(\mathcal{A},\succeq_{\mathcal{A}})$ that we consider in constructing monotones via the \ref{broad scheme}.
	In Section~\ref{sec:yield and cost}, we looked at $\mathcal{DC}(R)$ and $\mathcal{UC}(R)$ as possible choices, and we made use of the fact that the ordering on each is defined in terms of
		subset inclusion in $\mathcal{P}(R)$. 
	In the present section, we investigate what can be said about the case when the mediating set $\mathcal{A}$ arises from a resource theory 
		$\mathcal{Q} \coloneqq \RT{Q}$ as the power set $\mathcal{P}(Q)$.
	The root monotones will therefore be functions $\mathcal{P}(Q) \to \reals$.
	We will be interested in obtaining target monotones for a resource theory $\mathcal{R} \coloneqq \RT{R}$, possibly different from $\mathcal{Q}$.
	
	There are multiple choices of the mediating preorder $\succeq_{\mathcal{A}}$ that one could consider.
	We investigate two order relations on $\mathcal{P}(Q)$ in particular.
	They mirror the two choices of $(\mathcal{DC}(R), \supseteq)$ and $(\mathcal{UC}(R), \subseteq)$ in \Cref{sec:yield and cost}.
	They are defined below as $\succeq_{\rm enh}$ and $\succeq_{\rm deg}$.
	
	\subsection{Translating Monotones within a given Resource Theory}
	\label{sec:Translating Monotones from a Resource Theory to Itself}
	
		First of all, let us look at a particularly simple case in which the root and target resource theories are identical, i.e., $\mathcal{R} = \mathcal{Q}$.
		In Section~\ref{sec:Monotones for Sets of Resources}, we construct root monotones $\mathcal{P}(R) \to \reals$ given an arbitrary monotone $(R, \succeq_{R}) \to \ordreals$ and a choice of an ordering of sets of resources.
		Then, in Section~\ref{sec:Mediating map from a resource theory to itself}, we describe mediating order-preserving maps $R \to \mathcal{P}(R)$.
		Together, these give us a method to generate new monotones for $\mathcal{R}$ from existing ones.
		Later on, we extend this method to a translation of monotones from $\mathcal{Q}$ to $\mathcal{R}$ when the two resource theories are not identical.
		
		\subsubsection{Root Monotones for Sets of Resources from a Monotone for Individual Resources}
		\label{sec:Monotones for Sets of Resources}
		
			
			Given a resource ordering $(R,\freeconv)$ we are thus interested in comparing not just elements of $R$, but also its subsets---elements of $\mathcal{P}(R)$.
			There is not a canonical way to do so.
			Whether two subsets $S,T$ of $(R, \geq)$ should be ordered relative to each other or not depends on the intended interpretation of the ordering. 
			
			\textbf{Enhancement preorder.}			
			
			We could say that $S$ should be above $T$, denoted by $S \freeconv_{\rm enh} T$, if every element of $T$ lies below an element of $S$.
			In the following toy example (in which nodes are elements of $R$ and arrows depict order relations)
			
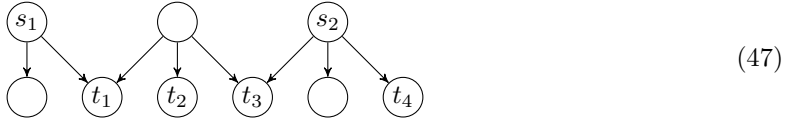
\begin{equation}\label{eq:enh_deg}
				\begin{tikzpicture}
	\begin{pgfonlayer}{nodelayer}
		\node [style={node_fix}] (0) at (-4, 1) {$s_1$};
		\node [style={node_fix}] (1) at (0, 1) {};
		\node [style={node_fix}] (2) at (4, 1) {$s_2$};
		\node [style={node_fix}] (3) at (-2, -1) {$t_1$};
		\node [style={node_fix}] (4) at (0, -1) {$t_2$};
		\node [style={node_fix}] (17) at (6, -1) {$t_4$};
		\node [style={node_fix}] (19) at (2, -1) {$t_3$};
		\node [style={node_fix}] (20) at (4, -1) {};
		\node [style={node_fix}] (21) at (-4, -1) {};
	\end{pgfonlayer}
	\begin{pgfonlayer}{edgelayer}
		\draw [style={fill_ar}] (0) to (3);
		\draw [style={fill_ar}] (1) to (4);
		\draw [style={fill_ar}] (2) to (17);
		\draw [style={fill_ar}] (2) to (19);
		\draw [style={fill_ar}] (1) to (3);
		\draw [style={fill_ar}] (1) to (19);
		\draw [style={fill_ar}] (2) to (20);
		\draw [style={fill_ar}] (0) to (21);
	\end{pgfonlayer}
\end{tikzpicture}
			\end{equation}
			we would have that $\{s_1,s_2\}$ is above $\{t_3,t_4\}$, but not above $\{t_1,t_2,t_3\}$.
			Recall that an element of $\mathcal{P}(R)$ is interpreted as specifying a set of resources, one of which an agent can choose to access and manipulate.
			With respect to this interpretation,
			\begin{quote}
				$S \freeconv_{\rm enh} T$ means: For every element of $T$ that an agent could make use of, there is an element of $S$ which is at least as valuable (according to the resource \mbox{ordering $\succeq$}).
			\end{quote}
			This notion of ordering of subsets of $(R,\geq)$ can be expressed via the existence of an enhancement map, defined as follows.
			\begin{definition}\label{def:enh}
				Let $(R,\geq)$ be a preordered set with two subsets $S,T$.
				A function $\mathsf{enh} \colon T \to S$ is termed an \textbf{enhancement} if we have
				\begin{equation}
					\mathsf{enh}(t) \geq t \quad \forall \, t \in T.
				\end{equation}
			\end{definition}	
			
			\begin{definition}\label{def:enh_ord}
				Given a preordered set $(R,\freeconv)$, define the \textbf{enhancement preorder} $\enhconv$ on $\mathcal{P}(R)$ by
				\begin{equation}
					S \enhconv T  \quad \coloniff \quad  \text{there exists an enhancement } T \to S.
				\end{equation}
			\end{definition}
			
			Given $S \enhconv T$, an agent that has access to resources in $S$ can obtain access to resources in $T$ by ignoring elements of $S$ outside the image of $\mathsf{enh}$ and applying the relevant free conversion to elements of $S$ within the image of $\mathsf{enh}$.
			In this way we obtain a preordered set $(\mathcal{P}(R), \enhconv)$ from the preordered set $(R, \freeconv)$, such that the latter is isomorphic to the ordering of singletons in $\mathcal{P}(R)$.

			The enhancement preorder can be equivalently expressed in terms of the downward closure operator. 
			
			\begin{lemma}\label{lem:enh_down}
				Let $S,T$ be two sets of resources in a resource theory with resource ordering $(R,\freeconv)$.
				Then the following are equivalent:
				\begin{enumerate}
					\item $S \enhconv T$,
					\item $\down (S) \supseteq \down (T)$, where the action of $\down$ is given as in \eqref{eq:down_on_sets}, and
					\item $S \succeq T$, where the ordering of sets of resources is given by \eqref{eq:Order of sets of resources}.
				\end{enumerate}
			\end{lemma}
			\begin{proof}
				Conditions (i) and (iii) are equivalent because $S \succeq T$ is defined as $T \subseteq \mathcal{R}_{\rm free} \boxtimes S$ which means that for each $t \in T$, there exists an $s_t \in S$ such that $t \in \mathcal{R}_{\rm free} \boxtimes s_t$, i.e., the map $t \mapsto s_t$ is an enhancement of type $T \to S$.
			
				If $S \enhconv T$ holds, then there is an enhancement $\mathsf{enh} \colon T \to S$.
				Denoting its image within $S$ by $\mathsf{enh}(T)$, we have
				\begin{equation}
					\down (S) \supseteq \down \bigl( \mathsf{enh}(T) \bigr) = \down \down \bigl( \mathsf{enh}(T) \bigr) \supseteq \down (T),
				\end{equation}
				since $\down \bigl( \mathsf{enh}(T) \bigr) \supseteq T$ follows from the definition of an enhancement.
				Thus condition (i) implies condition (ii).
				
				Finally, (ii) implies (iii) because $T$ is a subset of $ \mathcal{R}_{\rm free} \boxtimes T = \down (T)$.
			\end{proof}
			
			\begin{corollary}\label{lem:down_isotone}
				The map $\down \colon (\mathcal{P}(R), \enhconv) \to (\mathcal{P}(R), \supseteq)$ is order-preserving.
			\end{corollary}
			
			\textbf{Degradation preorder.}
			
			On the other hand, we could also say that $S$ should be above $T$, denoted by $S \freeconv_{\rm deg} T$, if every element of $S$ lies above an element of $T$.
			Then we would have that $\{s_1,s_2\}$ is above $\{t_1,t_2,t_3\}$, but not above $\{t_3,t_4\}$ in diagram \eqref{eq:enh_deg}.
			\begin{definition}\label{def:deg}
				Let $(R,\geq)$ be a preordered set with two subsets $S,T$.
				A function $\mathsf{deg} \colon S \to T$ is termed a \textbf{degradation} if we have
				\begin{equation}
					s \geq \mathsf{deg}(s) \quad \forall \, s \in S.
				\end{equation}
			\end{definition}
			
	
			\begin{definition}\label{def:deg_ord}
				Given a preordered set $(R, \freeconv)$, define the \textbf{degradation preorder} $\degconv$ on $\mathcal{P}(R)$ by
				\begin{equation}
					S \degconv T  \quad \coloniff \quad  \text{there exists a degradation } S \to T.
				\end{equation}
			\end{definition}
			
			Degradation ordering does not necessarily have a meaningful resource-theoretic interpretation, at least not in the same sense as the enhancement ordering does. 
			Therefore, we do not think of it as expressing that $S$ is more valuable than $T$, even though such an interpretation may be viable in specific contexts, one of which is as follows. 
			
			Imagine Alice and Bob play the following game. 
			Alice has to choose a set of resources $S \subseteq R$, while Bob receives a resource $x \in R$ from the referee.
			If Bob can recover an element of $S$ from $x$ (for free), he wins.
			Otherwise, Alice wins.
			The relation $S \degconv T$ means that $S$ is not worse than than $T$ from Alice's point of view, for any distribution of referee's choices.
			The fact that Alice can always win this game by choosing the empty set justifies why $\emptyset$ is a maximal element of $\degconv$.
					
			Regardless of its interpretation, the preordered set $(\mathcal{P}(R), \degconv)$ is useful for constructing resource monotones via the \ref{broad scheme}.
			One way to understand it is via a dual version of \Cref{lem:enh_down}.
			
			\begin{lemma}\label{lem:deg_up}
				Let $S,T$ be two subsets of a preordered set $\ordres$.
				Then we have
				\begin{equation}
					S \degconv T  \quad \iff \quad  \up (S) \subseteq \up (T)
				\end{equation}
				where the action of $\up$ is given as in \eqref{eq:down_on_sets}.
			\end{lemma}
			\begin{proof}
				If $S \degconv T$ holds, then there is a degradation $\mathsf{deg} \colon S \to T$.
				Denoting its image within $T$ by $\mathsf{deg}(S)$, we have
				\begin{equation}
					\up (T) \supseteq \up \bigl( \mathsf{deg}(S) \bigr) = \up \up \bigl( \mathsf{deg}(S) \bigr) \supseteq \up (S),
				\end{equation}
				since $\up \bigl( \mathsf{deg}(S) \bigr) \supseteq S$ follows from the definition of a degradation.
				
				Conversely, if $\up (S) \subseteq \up (T)$ holds, then we also have $S \subseteq \up (T)$.
				That is, for every $s \in S$, there exists some $t_s \in T$ such that $s \freeconv t_s$.
				Thus, the function given by $s \mapsto t_s$ is a degradation of type $S \to T$.
			\end{proof}
			
			\begin{corollary}\label{lem:up_isotone}
				The map $\up \colon (\mathcal{P}(R), \degconv) \to (\mathcal{P}(R), \subseteq)$ is order-preserving.\footnotemark{}
				\footnotetext{Note the opposite direction of the ordering in the codomain relative to prior occurences.
				If we were to keep the convention from earlier, we would instead say that $\up$ is an order-reversing function.}
			\end{corollary}
			
			\textbf{Root monotones for enhancement and degradation preorders.}
			
			Another way to get some intuition about the enhancement and degradation preorders is to look at what they correspond to in the specific example of a total order of singletons.
			\begin{example}\label{ex:enh_reals}
				Specifically, the enhancement preorder for sets of extended real numbers is given by the comparison of their suprema.
				That is, for $S, T \in \mathcal{P}(\reals)$ we have 
				\begin{equation}
					S \enhgeq T \quad \iff \quad  \sup S \geq \sup T.
				\end{equation}
				On the other hand, the degradation preorder is given by infima.
				In this case, both $\enhgeq$ and $\deggeq$ are thus total preorders.
			\end{example}
			As an immediate consequence of this example, we conclude that 
			\begin{equation}
				\begin{split}
					\mathrm{sup} &\colon \bigl( \mathcal{P}(\reals), \enhgeq \bigr) \to \bigl(\reals , \geq \bigr) \\
					\mathrm{inf} &\colon \bigl( \mathcal{P}(\reals), \deggeq \bigr) \to \bigl(\reals , \geq \bigr)
				\end{split}
			\end{equation}
			are both monotones.
			
			\begin{lemma}\label{lem:sup_isotone}
				Given a partial, order-preserving function $f_W \colon (R, \freeconv) \to (Q, \succeq)$ with upward closed domain of definition $W \in \mathcal{UC}(R)$, the function $f_W \colon \mathcal{P}(R) \to \mathcal{P}(Q)$ that maps each $S$ to its image under $f_W$ is order-preserving with respect to the corresponding enhancement preorders.
			\end{lemma}
			\begin{proof}
				Consider $S,T \in \mathcal{P}(R)$ and an enhancement $\mathsf{enh} \colon T \to S$.
				If $f_W (T)$ is non-empty, we can construct an enhancement $f_W (T) \to f_W (S)$ as follows.
				Given $q \in f_W(T)$, pick an arbitrary element $t \in T$ such that $f_W(t) = q$ holds.
				Since $t$ is in the upset $A$, $f_W$ is also defined for its enhancement so that we can let the image of $q$ be $f_W(\mathsf{enh} (t)) \in f_W (S)$.
				Since $f_W$ is order-preserving, we have
				\begin{equation}
					f_W \bigl( \mathsf{enh} (t) \bigr) \geq f_W(t) = r.
				\end{equation}
				Thus, any such function is an enhancement and we conclude that $S \succeq_{\rm enh} T$ implies the desired relation ${f_W (S) \geq_{\rm enh} f_W (T)}$.
			\end{proof}
			\begin{lemma}\label{lem:sup_isotone_deg}
				Given a partial, order-preserving function $h_Y \colon (R, \freeconv) \to (Q, \succeq)$ with downward closed domain of definition $Y \in \mathcal{DC}(R)$, the function $h_Y \colon \mathcal{P}(R) \to \mathcal{P}(Q)$ is order-preserving with respect to the corresponding degradation preorders.
			\end{lemma}
			\begin{proof}
				The proof is analogous to that of \Cref{lem:sup_isotone}; replacing $\mathsf{enh}$ with a degradation $\mathsf{deg} \colon S \to T$ yields a degradation $h_Y (S) \to h_Y (T)$.
			\end{proof}
			\begin{corollary}\label{lem:monotone_extensions}
				Given partial monotones $f_W,h_Y \colon (R, \freeconv) \to \ordreals$ defined on upward (downward) closed subsets of $R$ respectively, the two maps
				\begin{equation}\label{eq:root_enh_deg}
					\begin{split}
						\fWmax &\colon \bigl( \mathcal{P}(R), \enhconv \bigr) \to \bigl(\reals , \geq \bigr) \\
						h_Y\textup{-}\mathsf{min} &\colon \bigl( \mathcal{P}(R), \degconv \bigr) \to \bigl(\reals , \geq \bigr)
					\end{split}
				\end{equation}
				given by the composites $\mathrm{sup} \circ f_W$ and $\mathrm{inf} \circ h_Y$ respectively are both monotones.
			\end{corollary}
			The two functions from (\ref{eq:root_enh_deg}) constitute candidate root monotones for constructions that translate monotones from one resource theory to another.
			
			In the case of a trivial preorder on $R$, under which any two distinct elements are incomparable, the corresponding enhancement and degradation preorders are just $\supseteq$ and $\subseteq$ respectively.
			Moreover, every subset of $R$ is both downward and upward closed.
			Thus, we get the following special case.
			\begin{corollary}\label{lem:function_extensions}
				Given a partial function $f_W \colon R \to \reals$, the two maps 
				\begin{equation}
					\begin{split}
						\fWmax &\colon \bigl( \mathcal{P}(R), \supseteq \bigr) \to \bigl(\reals , \geq \bigr) \\
						\fWmin &\colon \bigl( \mathcal{P}(R), \subseteq \bigr) \to \bigl(\reals , \geq \bigr)
					\end{split}
				\end{equation}
				are both monotones.
			\end{corollary}

		\subsubsection{Mediating Order-Preserving Maps Between Individual Resources and Sets of Resources}
		\label{sec:Mediating map from a resource theory to itself}
			
			We also need to find mediating order-preserving maps in order to be able to use both $\fmax$ and $\fmin$ as root monotones in the \ref{broad scheme}.
			\begin{example}[copying is order-preserving]
				\label{ex:copy}
				Consider the map $\mathsf{Copy}_2 \colon \ordres \to (\mathcal{P}(R), \succeq_{\rm enh})$ defined by
				\begin{equation}
					\mathsf{Copy}_2(r) \coloneqq r \boxtimes r.		
				\end{equation}
				Given that $s \in R_{\rm free} \boxtimes r$ implies $s  \boxtimes s \in (R_{\rm free} \boxtimes r) \boxtimes (R_{\rm free}  \boxtimes r) = R_{\rm free} \boxtimes (r \boxtimes r)$, it follows that 
					$s \succeq r$ implies $s\boxtimes s \succeq r\boxtimes r$, and consequently $\mathsf{Copy}_2$ is an order-preserving map.
				The same works for the map $\mathsf{Copy}_n \colon R \to \mathcal{P}(R)$ defined by $\mathsf{Copy}_n(r) = r^{\boxtimes n}$ (i.e., the $n$-fold universal combination). 
			\end{example}
			
			\begin{example}[adding a catalyst is order-preserving]
				\label{ex:augmentation}
				Given a resource $c \in R$, consider the map $\mathsf{Aug}_{c} \colon \ordres \to (\mathcal{P}(R), \succeq_{\rm enh})$ termed the 
					\textbf{augmentation by $\bm{c}$}\footnotemark{} that is defined for any \mbox{resource $r$} by
				\footnotetext{We can think of the action of $\mathsf{Aug}_{c}$ as combination of resources with a catalyst $c$ \cite[Section 3.2.4]{Gonda2021}.}%
				\begin{equation}
					\mathsf{Aug}_{c} (r) \coloneqq c \boxtimes r.
				\end{equation}
				This map is always order-preserving, since the implication $r \succeq s  \implies  c \boxtimes r \succeq c \boxtimes s$ follows from the definition of a resource theory and the resource 
					ordering for any $r,s,c \in R$.
				The same construction also works when $c$ is not just a single resource, but a set of resources.
			\end{example}
			
			Both $\mathsf{Aug}_{c}$ and $\mathsf{Copy}_n$ are thus examples of order-preserving maps from $\ordres$ to $(\mathcal{P}(R), \succeq_{\rm enh})$ and can be 
				used to obtain monotones $\fmax \circ \mathsf{Aug}_{c}$ and $\fmax \circ \mathsf{Copy}_n$ for any monotone $f$, any resource $c$ and any integer $n$.
			These monotone constructions differ from the ones we have seen in Section~\ref{sec:yield and cost} in that the optimization is generally restricted to range over a much smaller set of resources.
			This is a consequence of the monotonicity of $f$, which allows us to use the root monotones for $\mathcal{P}(Q)$ from \Cref{lem:monotone_extensions} in order to remove the free image and free preimage maps as compared to cost and yield constructions.
			
			As a mild generalization of \Cref{ex:augmentation} from using a single catalyst $c$ to a set $U$, we obtain the following result.
			\begin{proposition}
				For any subset $U$ of the set $R$ of all resources, the $U$-image map
				\begin{equation}
					\downwrt{U} \colon \bigl( \mathcal{P}(R), \enhconv \bigr) \to \bigl( \mathcal{P}(R), \enhconv \bigr)
				\end{equation}
				as introduced in \Cref{def:D-image map} is order-preserving.
			\end{proposition}
			\begin{proof}
				Using \Cref{lem:enh_down}, this statement reduces to that of \Cref{thm:Compatibility of star and order} (compatibility of $\boxtimes$ and $\succeq$).
			\end{proof}
			In particular, restricting $\downwrt{U}$ to individual resources provides a valid mediating map for the translation of monotones via the enhancement ordering.
			That is, for any set of resources $U$ and any partial monotone $f_W$ with an upward closed domain $W$, we get a target monotone via the composition
			\begin{equation}\label{eq:preim_monotone}
				\begin{tikzcd}
					(R, \freeconv) \ar[r, "\downwrt{U}"] 	& 	\bigl( \mathcal{P}(R), \freeconv_{\rm enh} \bigr) \ar[rr, "\fWmax"] 	&& 	\ordreals
				\end{tikzcd}
			\end{equation}
			
			
			There is a similar result for preimage maps, which can be used to generate mediating maps for the degradation ordering.
			\begin{proposition}\label{prop:preim_isotone}
				For any subset $U$ of the set $R$ of all resources, the $U$-preimage map
				\begin{equation}
					\upwrt{U} \colon \bigl( \mathcal{P}(R), \degconv \bigr) \to \bigl( \mathcal{P}(R), \degconv \bigr)
				\end{equation}
				as introduced in \Cref{def:D-image map} is order-preserving.
			\end{proposition}
			\begin{proof}
				By \Cref{lem:deg_up}, we have
				\begin{equation}
					S \degconv T   \iff    \up (S) \subseteq  \up (T)   \implies  \upwrt{U} \circ \up (S) \subseteq  \upwrt{U} \circ \up (T).
				\end{equation}
				Note that for arbitrary sets of resources $V$ and $W$ the respective preimage maps commute,
				\begin{equation}
					\upwrt{V} \circ \upwrt{W} = \upwrt{W} \circ \upwrt{V},
				\end{equation}
				since $\boxtimes$ is commutative.
				In fact, they are both equal to $\upwrt{V \boxtimes W}$ by \Cref{thm:composing image maps}.
				Therefore, choosing $V$ to be the $U$ from the proposition statement and choosing $W = R_{\rm free}$ gives
				\begin{equation}
					\up \circ \upwrt{U} (S) \subseteq  \up \circ \upwrt{U} (T).
				\end{equation}
				By \Cref{lem:deg_up} again, this implies $\upwrt{U} (S) \degconv \upwrt{U} (T)$, which is what we aimed to establish.
%
			\end{proof}
			
			As before, restricting $\upwrt{U}$ to individual resources yields a mediating map for the translation of monotones via the degradation ordering.
			That is, for any set of resources $U$ and any partial monotone $f_W$ with a downward closed domain $W$, we get a target monotone via the composition
			\begin{equation}
				\begin{tikzcd}
					(R, \freeconv) \ar[r, "\upwrt{U}"] 	& 	\bigl( \mathcal{P}(R), \freeconv_{\rm deg} \bigr) \ar[rr, "\fWmin"] 	&& 	\ordreals
				\end{tikzcd}
			\end{equation}
			
			We will see further examples of this kind in Section~\ref{sec:Monotones From Contractions in General}, but the source and target resource theories differ there.
			

	\subsection{Translating Monotones from a Resource Theory of Distinguishability}
	\label{sec:Monotones from Information Theory}

		Many resource theories of interest either have an information-theoretic flavour or are explicitly about informational resources.
		It is no surprise then, that in these resource theories, measures of information often crop up as monotones or as building blocks for resource monotones.
		As we will see below, to any resource theory $\RT{R}$, it is possible to associate an information theory where, roughly speaking, the set $R$ of resources constitutes the alphabet for the encoding of a 
			classical message.
		This association can then be used to understand such results in greater generality.

		We now provide a couple of examples of monotone constructions based on contractions that we aim to understand and generalize here.

		Consider a quantum resource theory of states, where $R$ contains all quantum states, and $R_{\rm free}$ contains states considered to be free (we need not stipulate which 
			maps of states are considered free operations in the resource theory because the construction will work regardless of this choice).
		A contraction in such a resource theory is a real-valued function $f$ of pairs of quantum states that satisfies the data processing inequality
			\begin{equation}
				\label{eq:data processing inequality}
				f(\rho, \sigma) \geq f \bigl(\Phi(\rho), \Phi(\sigma) \bigr)
			\end{equation}
			for all states $\rho$ and $\sigma$ (of the same system) and all CPTP maps $\Phi$.

		\begin{example}[monotones as distance from the free set]
			\label{ex:Monotones from contractions}
			Given a contraction, it is well-known that one can obtain a monotone by minimizing over the set of all free states in one of its arguments.
			That is, the function $M \colon R \to \reals$ given by
			\begin{align}
				\label{eq:monotone from contraction}
				M (\rho) = \inf \Set*[\big]{f(\rho, \sigma)  \given  \sigma \in R_{\rm free} }
			\end{align}
			is a resource monotone.
			Various monotones based on distance measures such as the trace distance, relative R\'{e}nyi entropies~\cite{Petz1986}, and many others arise in this way.
			An extensive overview of these kinds of monotones can be found in~\cite{Chitambar2018}.
		\end{example}
		
		\begin{example}[monotones from operations that commute with the free operations]
			\label{ex:Monotones from contractions 2}
			
			
			Consider the quantum resource theory of asymmetry with respect to a symmetry group $G$ introduced in \Cref{ex:rt_quantum} \ref{ex:rt_asymmetry}.
			It is one where the free processes are those that are covariant with respect to a group action of $G$ on quantum states.
			Given a contraction $f$ and a twirling map $\mathcal{G}_p$ weighted by a distribution $p$ over the group, that is, 
			\begin{equation}
				\mathcal{G}_p(\cdot)= \int \mathrm{d}g \; p(g)\; U_g (\cdot) U^{\dag}_g,
			\end{equation}
			the following function is a monotone~\cite{Marvian2012}:
			\begin{equation}
				M(\rho) \coloneqq f \bigl( \rho, \mathcal{G}_p (\rho) \bigr).
			\end{equation}  
			The proof of monotonicity relies on the fact that any free operation mapping $\rho$ to $\sigma$ also maps $\mathcal{G}_p (\rho)$ to $\mathcal{G}_p (\sigma)$.
			This is because a twirling map commutes with each of the free operations in this resource theory.
			It is worth considering some special cases of this monotone.
			If $p$ is a point distribution on some nonunit element of the group, $g_0$, so that $M(\rho)= f(\rho, U_{g_0} \rho U^{\dag}_{g_0})$, then the monotone quantifies how distinguishable $\rho$ is from its image (or ``rotation'') under $g_0$.
			If $p$ is a uniform distribution over the group, so that $M(\rho)= f(\rho, \int \mathrm{d}g  U_g  \rho U^{\dag}_g)$, then the monotone quantifies how distinguishable $\rho$ is from its ``uniformly twirled'' counterpart.
			One can then understand the monotonicity of these functions intuitively as the statement that more asymmetric states are more distinguishable from their rotated and uniformly twirled counterparts.
			Note, furthermore, that one can define a monotone from a contraction and a \emph{pair} of distributions, $p$ and $q$, given by $M(\rho) = f ( \mathcal{G}_q (\rho), \mathcal{G}_p (\rho))$~\cite{Marvian2012}.
		\end{example}
		
		\subsubsection{Resource Theory of Distinguishability}
		\label{sec:Encodings}
			
			In order to understand monotones constructed from contractions (Examples~\ref{ex:Monotones from contractions} and \ref{ex:Monotones from contractions 2}) as special cases of the \ref{broad scheme} and to thereby generalize them, we study a resource theory of tuples of resources in which contractions are resource monotones.
			In the context of \emph{quantum states}, a related notion has been studied recently (for the special case of pairs of states), namely, the resource theory of asymmetric distinguishability~\cite{Wang2019}.

			\begin{definition}
				\label{def:Information Theory}
				Let $k$ be a natural number and let $\mathcal{R} \equiv (R_{\rm free}, R, \boxtimes)$ be a resource theory in which $r \boxtimes s$ is a singleton set of resources for all $r,s \in R$; i.e., 
					resources can be combined in exactly one way so that we can think of $\boxtimes$ as a binary operation $R \times R \to R$.
				The \textbf{resource theory of (unconstrained) $\bm{k}$-distinguishability associated to $\mathcal{R}$} is a resource theory $(\R{k}_{\rm cons}, \R{k}, \boxtimes)$, where
				\begin{enumerate}
					\item $\R{k} = R \times R \times \ldots \times R$ is the set of all $k$-tuples of resources from $R$,
					\item $\R{k}_{\rm cons} \subseteq \R{k}$ is the set of all constant $k$-tuples; i.e., those of the form $(r, r, \ldots, r)$; and
					\item the composition of $k$-tuples is given by
						\begin{equation}
							(r_1, r_2, \ldots, r_k) \boxtimes (s_1, s_2, \ldots, s_k) \coloneqq (r_1 \boxtimes s_1, r_2 \boxtimes s_2, \ldots, r_k \boxtimes s_k).
						\end{equation}
				\end{enumerate}
			\end{definition}
			One can construct a similar resource theory even if one relaxes the assumption that any two resources can be combined in exactly one way.
			However, the construction becomes more complicated and therefore we forego considering it here in order to avoid obscuring the main ideas.
			The more general construction of a resource theory of $k$-distinguishability can be found in \cite{Gonda2021}. 
			All the results mentioned here also have their counterparts in the more general case.
			
			A natural way to view the resource theory $(\R{k}_{\rm cons}, \R{k}, \boxtimes)$ is to think of the $k$-tuples as encodings of a classical hypothesis.
			Namely, if $H_k$ is a set of cardinality $k$ representing a classical hypothesis, a $k$-tuple of resources $(r_1,\dots, r_k)$ can be conceptualized as a function $H_k \to R$ taking each value $h\in H_k$ to a resource $r_h\in R$.
			The free $k$-tuples are the constant ones because they can be constructed with no information about the value of $H_k$.
			That is, if the $k$-tuple in question is from $\R{k}_{\rm cons}$, so that every value of $H_k$ is associated to the \emph{same} resource, then learning the identity of the resource teaches one nothing about 
				the value of $H_k$.
			Note that the resource theory $(\R{k}_{\rm cons}, \R{k}, \boxtimes)$ doesn't distinguish between valuable and free resources in the original theory $\mathcal{R}$.
			It is purely about the information content of the encodings.
				
			If we now look back at the definition of a contraction for quantum states via data processing inequality, we can see that monotones in the resource theory $(\R{2}_{\rm cons}, \R{2}, \boxtimes)$ provide a 
				suitable generalization of this notion.
			We thus refer to monotones in the resource theory $(\R{2}_{\rm cons}, \R{2}, \boxtimes)$ as \textbf{contractions}. 
			Analogously, a monotone in the resource theory $(\R{k}_{\rm cons}, \R{k}, \boxtimes)$ is termed a $\bm{k}$\textbf{-contraction}.
			
			\begin{remark}\label{rem:cons_free}
				In several monotone constructions throughout Section~\ref{sec:Translating Monotones}, we will also make use of a variant of resource theories of unconstrained $k$-distinguishability associated to $\mathcal{R}$.
				It is one that combines the two restrictions---one given by $R_{\rm free} \subseteq R$ and the other by $\R{k}_{\rm cons} \subseteq \R{k}$.
				More specifically, we define 
				\begin{equation}
					\label{eq:free distinguishability}
					\Rconsfree{k} \coloneqq \Rcons{k} \cap \Rfreebasic{k},
				\end{equation}
				and consequently obtain a resource thery $(\Rconsfree{k}, \R{k}, \boxtimes)$ termed the \textbf{resource theory of constrained $\bm{k}$-distinguishability}.
				Note that this notion subsumes the one from Definition~\ref{def:Information Theory} if one takes $R_{\rm free} \coloneqq R$ so that $\mathcal{R}$ is just $(R,R,\boxtimes)$.
			\end{remark}
			
			\begin{example}[yield applied to pairs of resources]
				\label{ex:cost for pairs}
				Generalized channel divergences \cite{Leditzky2018} arise from the generalized yield construction when thinking about the resource theory of pairs of resources, i.e., a resource theory of distinguishability \cite{Wang2019}.
				More details on pairs (and other tuples) of resources and in what way they constitute a resource theory can be found in \Cref{sec:Encodings}.
			\end{example}

		\subsubsection{Monotones from Functions That Commute with the Free Operations}
		\label{sec:monotones from twirling generalized}
		
			We now describe a generalization of the monotone construction in Example \ref{ex:Monotones from contractions 2} to a general resource theory $\mathcal{R} = \RT{R}$. 
			Let $(\R{2}_{\rm cons}, \R{2}, \boxtimes)$ be the resource theory of 2-distinguishability associated to $\mathcal{R}$, as in Definition~\ref{def:Information Theory}. 
			We take its resource ordering $(\R{2}, \succeq)$ to be the mediating preordered set in the \ref*{broad scheme}.
			A root monotone is therefore a particular contraction $f \colon \R{2} \to \reals$, i.e., a monotone in the resource theory $(\R{2}_{\rm cons}, \R{2}, \boxtimes)$.
			
			Consider a function $\Phi \colon R \to R$ which commutes with the free operations in the sense that for any $t \in R_{\rm free}$ and $r \in R$, we have 
			\begin{equation}
				\Phi(t \boxtimes r) = t \boxtimes \Phi(r).
			\end{equation}
			Take the mediating order-preserving map in the \ref*{broad scheme} to be $r \mapsto (r,\Phi(r))$. 
			The proof that the latter map is indeed order-preserving is strightforward.
			If $r \succeq s$, so that there is a free resource $t$ satisfying $s= t \boxtimes r$, then we have
			\begin{equation}
				(t,t) \boxtimes \bigl( r,  \Phi(r) \bigr) = \bigl( t \boxtimes r, t \boxtimes \Phi(r) \bigr) = \bigl( s, \Phi(s) \bigr)
			\end{equation}
			and therefore the images of $r$ and $s$ are ordered in the resource theory of $2$-distinguishability:
			\begin{equation}
				(r,  \Phi(r)) \succeq (s, \Phi(s)).
			\end{equation}

			It then follows that we can construct a monotone for resources in $R$, for any given contraction $f$ and any function $\Phi$ that commutes with free operations, as follows:
			\begin{equation}\label{eq:monotone_from_free_commuting}
				M(r) \coloneqq f \bigl( r, \Phi(r) \bigr)
			\end{equation}
			This construction clearly applies to resource theories which are not universally combinable, since the argument does not make use the commutativity of the $\boxtimes$ operation (see \cite{Gonda2021} for more details).			
			Moreover, it also applies to any $f$ that is merely a monotone in a resource theory of unconstrained $k$-distinguishability from \Cref{rem:cons_free}.
			
			Besides twirling operations in the resource theory of asymmetry, there are functions in other resource theories, which commute with the free operations.
			\begin{example}\label{ex:dist_from_thermal}
				In the resource theory of athermality introduced in \Cref{ex:rt_classical} \ref{ex:rt_athermality}, one can take $\Phi$ to be the discarding map followed by a preparation of the thermal state for the system just discarded.
				Monotones constructed in this way describe the thermo-majorization order \cite{Horodecki2013}, a special case of relative majorization first introduced in \cite{Veinott1971}.
				They arise by translating the matrix majorization order (which, in the case of a binary hypothesis, boils down to zonotope inclusion \cite{Dahl1999}) through the mediating map $r \mapsto (r,\Phi(r))$.
				For the special case of resource theory of nonuniformity, if we choose the root monotone to be the relative entropy (a.k.a.\ Kullback-Leibler divergence), we obtain the Shannon nonuniformity from \Cref{ex:monotones} \ref{ex:monotones_nonuniformity}.
			\end{example}
			\begin{example}\label{ex:dist_from_product}
				A similar example arises in resource theories wherein the free operations are \emph{local}, i.e., those which are of tensor product form with respect to some prespecified partition.
				The map that sends every state to the tensor product of its marginals commutes with the free operations and thus can be used in the above monotone construction.
				If $\rho$ denotes a bipartite state on the joint system $\alpha \otimes \beta$ and $\rho_{\alpha} \otimes \rho_{\beta}$ is the tensor product of its projections (or ``marginals''), the resulting monotone is
				\begin{equation}
					M(\rho) =  f ( \rho , \rho_{\alpha} \otimes \rho_{\beta} ).
				\end{equation}
				Choosing $f$ to be the relative Von Neumann entropy, for instance, gives rise to the mutual information $I(\alpha \mathbin{;} \beta)$ as a measure of correlation strength.
			\end{example}

		\subsubsection{Minimal Distinguishability as a Monotone}
		\label{sec:Monotones From Contractions in General}
		
		
			We now turn to the question of generalizing the monotone construction in Example \ref{ex:Monotones from contractions}, which expresses the distance of a resource from the set of free resources in terms of the smallest distinguishability between the given resources and a free resource. 
			Because of the optimization over all free resources in equation~\eqref{eq:monotone from contraction}, which occurs at the level of \emph{sets} of $2$-tuples of resources, we use the power set $\mathcal{P}(\R{2})$ as the mediating set in order to cast this construction as an instance of the \ref{broad scheme}.
			This is in contrast to the mediating set $\R{2}$ that is sufficient for our discussion of monotones obtained from functions that commute with free operations.
				
			Given any contraction $f$, we can obtain two root monotones for $\mathcal{P}(\R{2})$ using the scheme described in Section~\ref{sec:Monotones for Sets of Resources}. 
			They are $\fmax$ and $\fmin$ as defined in \Cref{lem:monotone_extensions}.
			The root monotone $\fmax$ can be used for translating monotones from a resource theory of distinguishability by virtue of mediating maps that generalize the functions that commute with free operations from \Cref{sec:monotones from twirling generalized}.
			These mediating maps are called \emph{oplax equivariant} maps in \cite{Gonda2021}.
			On the other hand, the root monotone $\fmin$ is the one that allows us to understand the monotone construction of Example~\ref{ex:Monotones from contractions} as an instance of the \ref*{broad scheme}. 
			The latter is the one we focus on hereafter.
			
			According to the \ref*{broad scheme}, in order to generate a resource monotone by pulling back $\fmin$, we need an order-preserving map from $\ordres$ to $(\mathcal{P}(\R{2}), \succeq_{\rm deg})$. 
			Consider the map $\mathcal{E}_{R_{\rm free}}$ defined by
			\begin{equation}\label{ERfree}
				\mathcal{E}_{R_{\rm free}} (r) \coloneqq \Set*[\big]{ (r, s)  \given  s \in R_{\rm free} }.
			\end{equation}
			This map is a valid candidate for the mediating order-preserving map to pull $\fmin$ through.
			The fact that it is order-preserving follows from a general result that we prove in Section~\ref{sec:General Translating} (in particular, see Lemmas~\ref{thm:Sufficient condition for order-preserving maps 3} 
				and~\ref{thm:Generating more order-preserving maps}).

			Following the \ref*{broad scheme}, the monotone $\fmin$ on $(\mathcal{P}(\R{2}), \succeq_{\rm deg})$ obtained from the $2$-contraction $f$ on $\R{2}$, can be pulled back to a monotone 
				$M$ on $(R, \succeq)$ via $M \coloneqq \fmin \circ \mathcal{E}_{R_{\rm free}}$.  
			Given the definitions of $\fmin$ and $\mathcal{E}_{R_{\rm free}}$, this can be unpacked to
			\begin{equation}\label{eq:GeneralMonotoneFromContraction}
				\begin{split}
					M(r) &= \fmin \left( \mathcal{E}_{R_{\rm free}} (r) \right) \\
						&= \inf f \bigl( \Set{ (r, s)  \given  s \in R_{\rm free} } \bigr)  \\
						&= \inf \Set*[\big]{ f(r,s)  \given  s \in R_{\rm free} }.
				\end{split}
			\end{equation}

			This expression is clearly very close to the construction of monotones from contractions in Example~\ref{ex:Monotones from contractions}. 
			The only difference that remains is that the optimization in equation~\eqref{eq:monotone from contraction} is specific to quantum \emph{states} while the one given here implements an optimization over resources $s$ that can be of any type.
			This is relevant because in the proof that $\fmin$ is order-preserving, in Lemma~\ref{lem:monotone_extensions}, it is assumed that the domain of $f$ is downward closed and thus it typically includes resources 
				\emph{of all types}. 
			Nevertheless, we can still use a type-specific contraction in the \ref*{broad scheme} as long as the mediating order-preserving map preserves the types of resources. 
			In the case of Example~\ref{ex:Monotones from contractions}, we need to ensure that the mediating map $\mathcal{E}_{\rm free}$ sends \emph{states} to sets of pairs of \emph{states}. 
			The \ref*{broad scheme} then generates a target monotone, which is non-trivial for states only.\footnotemark{}
			These particular issues regarding resource types disappear in the context of quantale modules~\cite{Gonda2021}.
						
			We have shown how to recast the monotone construction of Example~\ref{ex:Monotones from contractions} as an instance of the \ref*{broad scheme}. 
			We now turn to the question of how this abstract perspective on monotone constructions based on contractions leads to generalizations of the types of such monotone constructions considered previously. 

			The \ref*{broad scheme} stipulates that \emph{any} order-preserving map from $\ordres$ to $(\mathcal{P}(\R{2}), \succeq_{\rm deg})$ generates a monotone for $\ordres$ from $\fmin$. 
			A natural alternative to the order-preserving map $\mathcal{E}_{R_{\rm free}}$ defined by equation~\eqref{ERfree} is one of the same form but where the free set $R_{\rm free}$ is replaced by 
				any downward closed set $R_{\rm dc} \in \mathcal{DC}(R)$; namely, the map $\mathcal{E}_{R_{\rm dc}}$ defined by
			\begin{equation}
				\mathcal{E}_{R_{\rm dc}} (r) \coloneqq \Set*[\big]{ (r, s) \given  s \in R_{\rm dc} }.
			\end{equation}
			The order-preserving property of this function for the relevant domain and codomain follows by the same argument as for $\mathcal{E}_{R_{\rm free}}$.
			The reason why $R_{\rm dc}$ has to be downward closed is essentally the same as the reason why the set $Y$ in \Cref{lem:monotone_extensions} has to be downward closed.
			Consequently, we obtain a monotone
			\begin{equation}\label{eq:min_dis_downset}
				\fmin \circ \mathcal{E}_{R_{\rm dc}} (r) = \inf \Set*[\big]{ f(r,s)  \given  s \in R_{\rm dc} }.
			\end{equation}
			
			The upshot of this discussion is that for every $2$-contraction $f \colon \R{2} \to \reals$ and every downward closed subset $R_{\rm dc}$ of $R$, the function 
				$\fmin \circ \mathcal{E}_{R_{\rm dc}}$ is a monotone on $(R, \succeq)$. 
			For every concrete example of a downward closed set distinct from the free set, therefore, one obtains a corresponding variation on a monotone expressing the minimal ``$f$-distance'' from the free set---one 
				that quantifies the minimal distinguishability from the chosen downward closed set according to $f$.
				
			\begin{example}[examples of monotones quantifying the distance to a downward closed set]\label{examplesMonotonesDCsets}\hspace{0pt}
				\begin{compactenum}
					\item Consider the resource theory of bipartite quantum entanglement from \Cref{ex:rt_quantum} \ref{ex:rt_entanglement}.
					For any contraction $f$ of quantum states, the $f$-distinguishability between a state and the set of separable states is a popular entanglement monotone. 
					Recall that the set of states with entanglement rank at most $k$ is obtained from a downward closed set in this resource theory by intersecting with the set of states 
						(as noted in \Cref{examplesDCsets} \ref{ex:entrankDCsets}).\footnotemark{}
					\footnotetext{This result can be found in \cite[Theorem 6.23]{watrous2018theory} and it has been first shown in \cite{Terhal2000} where entanglement rank has been introduced under the name ``Schmidt number''.}%
					Thus, for every $k$, the $f$-distinguishability between a state and the set of states with entanglement rank at most $k$ is also an entanglement monotone.
					
					\item Resource theories of quantum coherence \cite{streltsov2017colloquium} describe resources of ``quantum superposition'' relative to a given basis of the Hilbert space, so that the free states are those given by diagonal density matrices.
						Besides these, there are downsets of bounded ``coherence number''. 
						In particular, a pure state is $k$-incoherent if it is block-diagonal with blocks of size no larger than $k \times k$.
						For general states, we can take the convex hull of these to obtain all the $k$-incoherent states \cite{johnston2018evaluating}.
						The coherence number $k$ introduced in \cite{sperling2015convex} plays the analogous role of Schmidt number from entanglement theory \cite{killoran2016converting,regula2018converting}.
						Indeed, $k$-incoherent states form a downward closed set for any $k$ and can be thus used to generate monotones via \cref{eq:min_dis_downset}.
					
					\item In the resource theory of \emph{tripartite} quantum entanglement, the set of states that are \emph{not} intrinsically $3$-way entangled form a downward closed set 
						(as noted in \Cref{examplesDCsets} \ref{ex:3-wayentangledDCsets}).
					Consequently, for any contraction $f$ of quantum states, the $f$-distance between a state and this set is an entanglement monotone that quantifies intrinsic $3$-way entanglement.
					
					\item Resource theories of $G$-asymmetry provide another illustrative example. 
					In particular, the set of states that are symmetric under a subgroup $H$ of $G$ form a downward closed set (as noted in \Cref{examplesDCsets} \ref{ex:symmetricDCsets}). 
					Therefore, the minimal $f$-distinguishability between a state and the states symmetric under $H$ is also an asymmetry monotone. 
					Roughly speaking, of all the ways that a state may break $G$-symmetry, the extent to which it does so by breaking $H$-symmetry is quantified by this monotone.\footnotemark{}
					\footnotetext{If the contraction $f$ is the relative entropy, then this monotone becomes $S(\rho||\mathcal{G}_H(\rho))$ and is equivalent to the Holevo asymmetry monotone $S(\mathcal{G}_H(\rho)) - S(\rho)$ associated to the uniform twirling $\mathcal{G}_H$ over $H$. 
						This equivalence follows from \cite[Proposition 2]{Gour2009}. 
						Note that the simplest case of such a monotone, $S(\mathcal{G}(\rho)) - S(\rho)$, was introduced in \cite{Vaccaro2008}.}%
				\end{compactenum}
			\end{example}
	
		\subsubsection{Monotones from k-Contractions in General}
		\label{sec:Monotones From k-Contractions in General}
		
			We can repeat the construction from Section~\ref{sec:Monotones From Contractions in General} for $f$ that is a $k$-contraction instead of a $2$-contraction.
			Specifically, given a $k$-contraction $f$, we have monotones
			\begin{align}
				\fmax &\colon (\mathcal{P}(\R{k}), \succeq_{\rm enh}) \to \ordreals,  &  \fmin &\colon (\mathcal{P}(\R{k}), \succeq_{\rm deg}) \to \ordreals,
			\end{align}
			which can be used as root monotones in the \ref{broad scheme}.
				 
			Notice that we can view $\mathcal{E} \equiv \mathcal{E}_{R}$ from Section~\ref{sec:Monotones From Contractions in General} as mapping $r$ to its preimage under the projection 
				$\Pi_1 \colon \R{2} \to R$ given by $(r,s) \mapsto r$.
			When we replace pairs of resources with $k$-tuples, we have $k$ such projections $\Pi_i \colon \R{k} \to R$ given by $(r_1, r_2, \ldots, r_k) \to r_i$, one for each $i \in \{1,2,\ldots, k\}$.
			Let us describe the case of $i=1$ here.
			The preimage map $\mathcal{E} \coloneqq \Pi_1^{-1}$ of the projection $\Pi_1$ is given by
			\begin{equation}
				\mathcal{E}(r) = \Set*[\big]{ (r, r_2, r_3, \ldots, r_k)  \given  r_j \in R \text{ for all } j \in \{2,3, \ldots, k\} },
			\end{equation}
			and it is order-preserving as a map $\ordres \to (\mathcal{P}(\R{k}), \succeq_{\rm deg})$, which follows from Lemma~\ref{thm:Sufficient condition for order-preserving maps 3}. 
			
			For any fixed downward closed set $W_{\rm dc}$ of $\R{k}$, we can then restrict each $\mathcal{E}(r)$ to its intersection with $W_{\rm dc}$ and retain the order-preserving property with respect to 
				$\succeq_{\rm deg}$.
			In particular, the map given by
			\begin{equation}
				\label{eq:EWdc}
				\begin{split}
					\mathcal{E}_{W_{\rm dc}} \colon \ordres &\to (\mathcal{P}(\R{k}), \succeq_{\rm deg}) \\
					r &\mapsto \mathcal{E}(r)  \cap  W_{\rm dc}
				\end{split}
			\end{equation}
			is order-preserving, which follows from Lemma~\ref{thm:Generating more order-preserving maps}. 
			Writing $\mathcal{E}_{W_{\rm dc}} (r)$ in this way, as an intersection of $\mathcal{E}(r)$ and $W_{\rm dc}$, is instructive because it more readily connects to the general results from Section~\ref{sec:General Translating}.
			Note that we can recover the map $\mathcal{E}_{R_{\rm dc}}$ from the previous section via the choice of $W_{\rm dc} = R \times R_{\rm dc}$.
			The fact that $R \times R_{\rm dc}$ is downward closed follows from the following lemma.
			\begin{lemma}[product of downward closed sets is downward closed]
				\label{thm:Product of DC is DC}
				For any family of $k$ sets of resources $\{S_i\}_{i=1}^{k}$, downward closed in the resource theory $\RT{R}$, the set
				\begin{equation}
					\label{eq:set product}
					S \coloneqq S_1 \times S_2 \times \dots \times S_k
				\end{equation} 
				is downward closed in $(\Rconsfree{k},\R{k},\boxtimes)$.
				Similarly, if each $S_i$ is upward closed in $\RT{R}$, then $S_1 \times S_2 \times \ldots \times S_k$ is upward closed in $(\Rconsfree{k},\R{k},\boxtimes)$.
			\end{lemma}
			\begin{proof}
				Note that the set of constant tuples consisting of free resources, $\Rconsfree{k}$, is a subset of all tuples that consist of free resources, $\Rfreebasic{k}$.
				Therefore, we have
				\begin{equation}
					\begin{split}
						S \boxtimes \Rconsfree{k} &\subseteq S \boxtimes \Rfreebasic{k} \\
							&= (S_1 \boxtimes R_{\rm free}) \times (S_2 \boxtimes R_{\rm free}) \times \dots \times (S_k \boxtimes R_{\rm free}) \\
							&\subseteq S_1 \times S_2 \times \dots \times S_k,
					\end{split}
				\end{equation}
				so that $S$ is indeed downward closed if each $S_i$ is downward closed in the original resource theory.
				
				Similarly, if each $S_i$ is upward closed, then we have	
				\begin{equation}
					\begin{split}
						\up (S_1 \times S_2 \times \dots \times S_k) &\subseteq \up (S_1) \times \up (S_2) \times \dots \times \up (S_k) \\
							&\subseteq S_1 \times S_2 \times \dots \times S_k.
					\end{split}
				\end{equation}
				It follows that $S$ is a subset of $\R{k}$ that is upward closed in $(\Rconsfree{k},\R{k},\boxtimes)$.
			\end{proof}
			However, not all the downward closed sets in the resource theory of $k$-distinguishability are of this kind.
			For example, there is generally no family $\{S_i\}$ of subsets of $\mathcal{R}$, downward closed or not, such that $\Rcons{k}$ (or indeed $\Rconsfree{k}$) is equal to 
				$S_1 \sqcuptimes S_2 \sqcuptimes \ldots \sqcuptimes S_k$.
				
			Sections~\ref{sec:Monotones From Contractions in General} and~\ref{sec:Monotones From k-Contractions in General} can be thus summarized by the following theorem.
			\begin{theorem}[generalized construction of monotones from $k$-contractions]
				\label{thm:contractions}
				Consider a resource theory $\RT{R}$ and let $f \colon \R{k} \to \reals$ be a $k$-contraction.
				For any subset $W_{\rm dc}$ of $\R{k}$, downward closed in $(\Rconsfree{k},\R{k},\boxtimes)$, the function $\fmin \circ \mathcal{E}_{W_{\rm dc}}$ is a monotone in the resource theory $\RT{R}$. 
			\end{theorem}
			\begin{proof}
				Note that we have the following equality
				\begin{equation}
					\fmin \circ \mathcal{E}_{W_{\rm dc}} = f_{W_{\rm dc}}\textup{-}\mathsf{min} \circ \mathcal{E} 
				\end{equation}
				where $f_{W_{\rm dc}}$ denotes the partial function with domain $W_{\rm dc}$ which coincides with $f$, whenever it is defined.
				The theorem then follows by combining \Cref{lem:monotone_extensions}, which establishes that $f_{W_{\rm dc}}\textup{-}\mathsf{min}$ is a monotone, and \Cref{thm:Sufficient condition for order-preserving maps 3}, which shows that the map $\mathcal{E}$ is order-preserving. 
			\end{proof}
		
		\subsubsection{Resource Weight and Robustness as Monotones Obtained from a $3$-Contraction}
		\label{sec:weight and robustness}

			As an example of how the generalized construction of monotones from $k$-contractions (Theorem~\ref{thm:contractions}) appears in a more concrete setting, we examine arguably two of the most 
				ubiquitous monotones---resource weight~\cite{Barrett2006,Abramsky2017} and resource robustness~\cite{Chitambar2018}---within the context of resource theories with a linear structure.
			By connecting them to a monotone in the resource theory of distinguishability, we complement the results of~\cite{takagi2019general}, \cite{Skrzypczyk2019}, and \cite{Ducuara2019}.
			In the first two, robustness measures are connected to discrimination tasks, while the latter article describes a similar connection between the weight measure and the state exclusion task.
		
			Let's consider a resource theory $\RT{R}$ with a convex-linear structure on $R$ that is preserved by~$\boxtimes$.
			The elements of $R$ can thus be represented as vectors, and convex combinations are preserved by the composition of resources $\boxtimes$.
			Furthermore, just like in the previous section, we assume that $r \boxtimes s$ is a single resource for all $r,s \in R$, so that $\boxtimes$ is a bilinear map $R \times R \to R$.
			A more general scenario corresponding to a broader idea of convex-linear resource theories is treated in 
				\cite{Gonda2021}.
			
			We can construct the resource theory of $3$-tuples $(\R{3}_{\rm cons}, \R{3}, \boxtimes)$ as described in Section~\ref{sec:Encodings} and define the following function.
			\begin{definition}
				\label{def:cva}
				The \textbf{convex alignment} is a function $\mathsf{cva} \colon \R{3} \to \reals$ defined by
				\begin{equation}
					\mathsf{cva} (r,s,t) \coloneqq 
						\begin{cases}
							\lambda & \text{if $r = \lambda s + (1-\lambda) t$ for $\lambda \in [0,1]$.} \\
							1 & \text{otherwise.}
						\end{cases}
				\end{equation}
			\end{definition}
			\begin{lemma}[convex alignment is a $3$-contraction]
				\label{lem:Alignment}
				Let $\RT{R}$ be a resource theory with a convex-linear structure as described at the start of this section.
				The convex alignment, $\mathsf{cva}$, is a monotone in the resource theory $(\R{3}_{\rm cons}, \R{3}, \boxtimes)$.
			\end{lemma}			
			\begin{proof}
				Let $(r, s, t) \in \R{3}$ and $(u,u,u) \in \Rcons{3}$.
				We aim to show that 
				\begin{equation}
					\mathsf{cva}(r,s,t) \geq \mathsf{cva}(r \boxtimes u, s \boxtimes u, t \boxtimes u)
				\end{equation}
				holds for all $r,s,t,u \in R$.
				
				If $\mathsf{cva}(r,s,t) = 1$, then its value cannot increase.
				Otherwise, if $\mathsf{cva}(r,s,t) = \lambda$ is strictly less than 1, then $r = \lambda s + (1-\lambda) t$.
				By the convex-linearity of $\boxtimes$, we have
				\begin{equation}
					\bigl( \lambda s + (1-\lambda) t \bigr) \boxtimes u = \lambda s \boxtimes u + (1-\lambda) t \boxtimes u.
				\end{equation}
				Therefore, $\mathsf{cva}(r,s,t) = \mathsf{cva}(r \boxtimes u, s \boxtimes u, t \boxtimes u)$ whenever $\mathsf{cva}(r,s,t) < 1$.
				Consequently, convex alignment is a $3$-contraction.
			\end{proof}
			
			Now we can use the generalized construction of monotones from $k$-contractions (Theorem~\ref{thm:contractions}) to get monotones for $\RT{R}$ by optimizing the convex alignment in various ways.
			Let us focus on a construction of $\calmin \circ \mathcal{E}_{W_{\rm dc}}$ with $W_{\rm dc}$ that is of form as in (\ref{eq:set product}).
			Specifically, we use $W_{\rm dc} = S_1 \times S_2 \times S_3 \in \mathcal{DC}(\R{3})$, where each $S_i$ is itself a downward closed subset of $R$. 
			There are many downward closed sets one could use for each $S_i$, but here we restrict our attention to the two most obvious choices---$R_{\rm free}$ and $R$.
			Even with this restriction, one can obtain 12 constructions of the form $\calmin \circ \mathcal{E}_{W_{\rm dc}}$.
			Specifically, there are three possible choices of the projection $\Pi_i$, which then determines $\mathcal{E}$.
			For each of them we let the corresponding $S_i$ be $R$ without loss of generality, which leaves 4 choices for the other two downward closed sets.
			Out of these 12 constructions in total, eight produce a constant monotone and are therefore uninteresting.
			The other four are the following.
			\begin{enumerate}
				\item The \textbf{resource weight} (also known as the resource fraction) $M_{\rm w} \colon R \to \reals$ is defined as $\calmin \circ \mathcal{E}_{W_{\rm dc}}$ for 
					$\mathcal{E} = \Pi_1^{-1}$ and $W_{\rm dc} = R \times R \times R_{\rm free}$.
					Explicitly, its value for any resource $r \in R$ is
					\begin{equation}\label{eq:res_weight}
						\begin{split}
							M_{\rm w} (r) &\coloneqq \inf \Set*[\big]{ \mathsf{cva} (r,s,t)  \given  s \in R,\; t \in R_{\rm free} } \\
								&= \inf \Set*[\big]{ \lambda  \given  r \in \lambda R + (1-\lambda) R_{\rm free} }.
						\end{split}
					\end{equation}
					It corresponds to the smallest weight of a resource that can be used to form $r$ by convex mixture with some free resource.
				
				\item The \textbf{resource robustness} $M_{\rm rob} \colon R \to \reals$ is defined as $\calmin \circ \mathcal{E}_{W_{\rm dc}}$ for $\mathcal{E} = \Pi_3^{-1}$ and 
					$W_{\rm dc} = R_{\rm free} \times R \times R$.
					Explicitly, its value for any resource $t \in R$ is
					\begin{equation}\label{eq:res_robustness}
						\begin{split}
							M_{\rm rob} (t) &\coloneqq \inf \Set*[\big]{ \mathsf{cva} (r,s,t)  \given  r \in R_{\rm free},\; s \in R } \\
								&= \inf \Set*[\big]{ \lambda  \given  \lambda s + (1-\lambda) t \in R_{\rm free} ,\; s \in R }.
						\end{split}
					\end{equation}
					It is the smallest weight of a resource that one needs to convexly mix with $t$ in order to obtain a free resource.
					
				\item The \textbf{free robustness} $M_{\rm f.\,rob} \colon R \to \reals$ is defined as $\calmin \circ \mathcal{E}_{W_{\rm dc}}$ for $\mathcal{E} = \Pi_3^{-1}$ and 
					$W_{\rm dc} = R_{\rm free} \times R_{\rm free} \times R$.
					Explicitly, its value for any resource $t \in R$ is
					\begin{equation}
						\begin{split}
							M_{\rm f.\,rob} (t) &\coloneqq \inf \Set*[\big]{ \mathsf{cva} (r,s,t)  \given  r \in R_{\rm free},\; s \in R_{\rm free} } \\
								&= \inf \Set*[\big]{ \lambda  \given  \lambda s + (1-\lambda) t \in R_{\rm free} ,\; s \in R_{\rm free} }.
						\end{split}
					\end{equation}
					It is the smallest weight of a free resource that one needs to convexly mix with $t$ in order to obtain another free resource.
				
				\item The \textbf{resource non-convexity} $M_{\rm nc} \colon R \to \reals$ is defined as $\calmin \circ \mathcal{E}_{W_{\rm dc}}$ for $\mathcal{E} = \Pi_1^{-1}$ and 
					$W_{\rm dc} = R \times R_{\rm free} \times R_{\rm free}$.
					Explicitly, its value for any resource $r \in R$ is
					\begin{equation}
						\begin{split}
							M_{\rm nc} (r) &\coloneqq \inf \Set*[\big]{ \mathsf{cva} (r,s,t)  \given  s \in R_{\rm free},\; t \in R_{\rm free} } \\
								&= \inf \Set*[\big]{ \lambda  \given  r \in \lambda R_{\rm free} + (1-\lambda) R_{\rm free} }.
						\end{split}
					\end{equation}
					It is trivial if all the sets of free resources happen to be convex. 
					Otherwise it tells us about the ordering of resources that are within the convex hull of the free resources, but are not free themselves. 
					It quantifies the relative distance of a resource from the set of free resources in terms of its convex decompositions into free resources. 
					Its value is set to 0 if the resource in question is free itself, and 1 if it is outside of the convex hull of the free resources.
			\end{enumerate}
								
			\begin{figure}[htb!]
				\ctikzfig{weight_rob}
				\caption{A pictorial depiction of the optimal convex decompositions for each of the four monotones mentioned in this section: (a) resource weight $M_{\rm w}$, (b) resource robustness $M_{\rm rob}$, (c) free robustness $M_{\rm f.\,rob}$, and (d) resource non-convexity $M_{\rm nc}$.
				Grey disc represents the set $R$ of all resources, while the yellow ``hourglass'' witin represents the free resources among them.
				In order to illustrate each of the four optimal decompositions, we select a distinct resource (element of $X$), depicted by a green node.
				These demopositions are given by the three points along one of the line segments with an orange and purple portion.
				The value of each of the monotones for these; $M_{\rm w}(a), M_{\rm rob}(b), M_{\rm f.\,rob}(c)$, and $M_{\rm nc}(d)$; can be read off as the length of the respective orange segment divided by the total lenth of the orange and purple segments combined.}
				\label{fig:weight_rob}
			\end{figure}
			
			As a consequence of Lemma~\ref{thm:Product of DC is DC}, Theorem~\ref{thm:contractions} and Lemma~\ref{lem:Alignment}, all four functions above are monotones. 
			However, being able to prove the monotonicity of these four functions is not where the value of the generalized construction of monotones from contractions lies.
			What they provide is an understanding of the assumptions required in order for these functions to be monotones.
			Furthermore, they give us a unified picture, within which we can adjust various elements of the monotone constructions according to the question we are interested in.
			In this case, there are many more monotones one can obtain from $\mathsf{cva}$ in this way, since $R$ or $R_{\rm free}$ in the optimization can be replaced by any other downward 
				closed set.

	\subsection{General Ways of Translating Monotones Between Resource Theories}
	\label{sec:General Translating}
	
		In Section~\ref{sec:Translating Monotones from a Resource Theory to Itself} we investigated how one can translate monotones from a resource theory $\mathcal{Q}$ given by $(Q_{\rm free}, Q, 
			\boxtimes_{Q})$ to a resource theory $\mathcal{R} = \RT{R}$ when the two are in fact identical. 
		Then, in Section~\ref{sec:Monotones from Information Theory}, we looked at the choice of $\mathcal{Q}$ in the form of a resource theory of distinguishability. 
		Here, we would like explore what can be said in general.
		Can the methods introduced in Sections~\ref{sec:Translating Monotones from a Resource Theory to Itself} and~\ref{sec:Monotones from Information Theory} be extended to the case of arbitrary $\mathcal{Q}$?
		
		We consider two choices of the mediating preordered set $(\mathcal{A}, \succeq_{\mathcal{A}})$: $(\mathcal{P}(Q), \succeq_{\rm enh})$ and ${(\mathcal{P}(Q), \succeq_{\rm deg})}$.
		For any monotone $f$ on $\mathcal{Q}$, we again have corresponding root monotones $\fmax$ and $\fmin$ as introduced in \Cref{lem:monotone_extensions}.
		In order to find out which maps can be used as the mediating order-preserving map ${\ordres \to (\mathcal{P}(Q), \succeq_{\rm enh})}$, we can use the following sufficient conditions.
		\begin{lemma}[mediating maps for $\succeq_{\rm enh}$]
			\label{thm:Sufficient condition for order-preserving maps 2}
			Let $\RT{R}$ and $(Q_{\rm free}, Q, \boxtimes_{Q})$ be resource theories and let $F \colon R \to \mathcal{P}(Q)$ be a function 
				with an extension $F \colon \mathcal{P}(R) \to \mathcal{P}(Q)$ obtained from the original $F$ by requiring that it commutes with unions.\footnotemark{}
			\footnotetext{The extension $F \colon \mathcal{P}(R) \to \mathcal{P}(Q)$ maps a set $S$ to the union of images of elements of $S$ under $F \colon R \to \mathcal{P}(Q)$.
			It is the unique extension of $F \colon R \to \mathcal{P}(Q)$ to a suplattice homomorphism $\mathcal{P}(R) \to \mathcal{P}(Q)$.}%
			If for all $r \in R$ we have
			\begin{equation}
				\label{eq:down commute}
				F(R_{\rm free} \boxtimes r) \subseteq Q_{\rm free} \boxtimes_{Q} F(r),
			\end{equation}
			i.e., if $F \bigl( \down (r) \bigr) \subseteq \down \bigl( F(r) \bigr)$ holds, then $F \colon \ordres \to (\mathcal{P}(Q), \succeq_{\rm enh})$ is order-preserving.
		\end{lemma}
		\begin{proof}
			We need to show that for $F$ as above, the implication $r \succeq s \implies F(r) \succeq_{\rm enh} F(s)$ holds for any $r,s \in R$.
			This fact can be broken down as follows:
			\begin{equation}
				\begin{split}
					r \succeq s &\iff s \in R_{\rm free} \boxtimes r \\
						&\implies F(s) \subseteq F(R_{\rm free} \boxtimes r) \\
						&\implies F(s) \subseteq Q_{\rm free} \boxtimes_{Q} F(r) \\
						&\iff F(r) \succeq F(s),
				\end{split}
			\end{equation}
			where the second implication follows from property (\ref{eq:down commute}).
			The statement of the lemma then follows by recognizing that $\succeq$ and $\succeq_{\rm enh}$ are identical as preorders on $\mathcal{P}(Q)$.
		\end{proof}
		\begin{lemma'}{thm:Sufficient condition for order-preserving maps 2}[\hypertarget{thm:Sufficient condition for order-preserving maps 2 v2}{mediating maps for $\succeq_{\rm enh}$}]
			Alternatively, if we have
			\begin{subequations}
				\label{eq:Sufficient condition for order-preserving maps 2 v2 (1)}
				\begin{align}
					\label{eq:star-morphism 0}
					F (r \boxtimes s) &= F(r) \boxtimes_{Q} F(s)  \quad \forall r, s \in R, \text{ and } \\
					\label{eq:Free-preserving 0}
					F(R_{\rm free}) &\subseteq Q_{\rm free} \boxtimes_{Q} F(0),
				\end{align}
			\end{subequations}
			then $F \colon \ordres \to (\mathcal{P}(Q), \succeq_{\rm enh})$ is order-preserving.
			Finally, if instead $F$ satisfies
			\begin{subequations}
				\label{eq:Sufficient condition for order-preserving maps 2 v2 (2)}
				\begin{align}
					\label{eq:star-morphism 1}
					F (r \boxtimes s) &\subseteq F(r) \boxtimes_{Q} F(s)  \quad \forall r, s \in R, \text{ and } \\
					\label{eq:Free-preserving 1}
					F(R_{\rm free}) &\subseteq Q_{\rm free},
				\end{align}
			\end{subequations}
			then $F \colon \ordres \to (\mathcal{P}(Q), \succeq_{\rm enh})$ is order-preserving.
		\end{lemma'}
		\begin{proof}
			Conditions~(\ref{eq:Sufficient condition for order-preserving maps 2 v2 (1)}) imply condition~(\ref{eq:down commute}) via
			\begin{equation}
				F(R_{\rm free} \boxtimes r) = F(R_{\rm free}) \boxtimes_{Q} F(r) 
					\subseteq Q_{\rm free} \boxtimes_{Q} F(0) \boxtimes_{Q} F(r) 
					= Q_{\rm free} \boxtimes_{Q} F(0 \boxtimes r) 
					= Q_{\rm free} \boxtimes_{Q} F(r),
			\end{equation}
			so that the first part of Lemma~\hyperlink{thm:Sufficient condition for order-preserving maps 2 v2}{56'} follows from Lemma~\ref{thm:Sufficient condition for order-preserving maps 2}.
			
			Conditions~(\ref{eq:Sufficient condition for order-preserving maps 2 v2 (2)}) imply condition~(\ref{eq:down commute}) via
			\begin{equation}
				F(R_{\rm free} \boxtimes r) \subseteq F(R_{\rm free}) \boxtimes_{Q} F(r) \subseteq Q_{\rm free} \boxtimes_{Q} F(r),
			\end{equation}
			so that the second part of Lemma~\hyperlink{thm:Sufficient condition for order-preserving maps 2 v2}{56'} also follows from Lemma~\ref{thm:Sufficient condition for order-preserving maps 2}.
		\end{proof}
		\begin{example'}{ex:augmentation}[adding a catalyst is order-preserving]
			The augmentation map $\mathsf{Aug}_{C} \colon R \to \mathcal{P}(R)$ from Example~\ref{ex:augmentation}, defined for any $C \subseteq R$ by
			\begin{equation}
				\mathsf{Aug}_{C} (r) \coloneqq C \boxtimes r,
			\end{equation}
			satisfies condition~(\ref{eq:down commute}), since we have
			\begin{equation}
				\mathsf{Aug}_{C} (R_{\rm free} \boxtimes r) = C \boxtimes R_{\rm free} \boxtimes r = R_{\rm free} \boxtimes \mathsf{Aug}_{C}(r).
			\end{equation}
			Lemma~\ref{thm:Sufficient condition for order-preserving maps 2} thus provides a way to prove that the function $\mathsf{Aug}_{C}$ is order-preserving as a map of type ${\ordres \to (\mathcal{P}(R),\succeq_{\rm enh})}$.
			However, it satisfies neither condition~(\ref{eq:star-morphism 0}) nor condition~(\ref{eq:Free-preserving 1}) in general.
		\end{example'}
		\begin{example'}{ex:copy}[copying is order-preserving]
			The copy map, $\mathsf{Copy}_{n} \colon R \to \mathcal{P}(R)$ was defined in Example~\ref{ex:copy} as the combination of $n$ copies of a resource,
			\begin{equation}
				\mathsf{Copy}_{n} (r) \coloneqq r \boxtimes r \boxtimes \ldots \boxtimes r \equiv r^{\boxtimes n}.
			\end{equation}
			The image of a set of resources $S$ by $\mathsf{Copy}_{n}$ cannot in general be expressed as $S^{\boxtimes n}$.
			Nevertheless, one can show that $\mathsf{Copy}_{n} (r \boxtimes s) \subseteq \mathsf{Copy}_{n}(r) \boxtimes \mathsf{Copy}_{n}(s)$ and $\mathsf{Copy}_{n}(R_{\rm free}) \subseteq R_{\rm free}$, 
				which corresponds to conditions~(\ref{eq:Sufficient condition for order-preserving maps 2 v2 (2)}).
		\end{example'}
		However, the map $\mathcal{E}$ defined in Section~\ref{sec:Monotones From k-Contractions in General} doesn't satisfy these conditions. 
		In general, only conditions~(\ref{eq:star-morphism 0}) and~(\ref{eq:star-morphism 1}) hold for $\mathcal{E}$.
		In particular, $\mathcal{E}$ is \emph{not} order-preserving as a function $\ordres \to (\mathcal{P}(Q), \succeq_{\rm enh})$.
		The function $\mathcal{E}$ is, nonetheless, an example of an order-preserving map $\ordres \to (\mathcal{P}(Q), \succeq_{\rm deg})$.
		How could we generalize this fact?
		Recall that for $k=2$, $\mathcal{E}$ maps $r$ to its preimage under the projection $\Pi_1 \colon (r,s) \mapsto r$.
		The following lemma provides sufficient conditions for such functions to be order-preserving in general.
		\begin{lemma}[mediating maps for $\succeq_{\rm deg}$]
			\label{thm:Sufficient condition for order-preserving maps 3}
			Let $\RT{R}$ and $(Q_{\rm free}, Q, \boxtimes_{Q})$ be resource theories and let $F \colon R \to \mathcal{P}(Q)$ be a function. 
			If there exists a map $G \colon Q \to R$ satisfying
			\begin{subequations}
				\label{eq:Sufficient condition for order-preserving maps 3}
				\begin{align}
					\label{eq:pre image map}
					F(r) &= G^{-1}(r) , \\
					\label{eq:star-morphism 2}
					G (p \boxtimes_{Q} q) &\supseteq G(p) \boxtimes G(q)  \quad \forall p, q \in Q, \text{ and } \\
					\label{eq:free surjective}
					G(Q_{\rm free}) &\supseteq R_{\rm free},
				\end{align}
			\end{subequations}
			then $F \colon \ordres \to (\mathcal{P}(Q), \succeq_{\rm deg})$ is order-preserving.
		\end{lemma}
		\begin{proof}
			We want to show that for any $r,s \in R$ such that $r \succeq s$, there exists a degradation $D \colon F(r) \to F(s)$.
			Firstly, note that the fact that the image of $Q_{\rm free}$ under $G$ contains $R_{\rm free}$ (property (\ref{eq:free surjective})) means that there is function 
				$G^{\dagger} \colon R_{\rm free} \to Q_{\rm free}$ such that 
			\begin{equation}
				G \circ G^{\dagger} = \Id_{R_{\rm free}},
			\end{equation}
			where $\Id_{R_{\rm free}}$ is the canonical embedding of $R_{\rm free}$ in $R$.
			That is, $G^{\dagger}$ is a partial right inverse of $G$.
			
			If $r \succeq s$ holds, then there is an $x \in R_{\rm free}$ such that $s \in r \boxtimes x$.
			For all $s \in F(r) = G^{-1}(r)$, we then have 
			\begin{equation}
				s \in r \boxtimes x = G(s) \boxtimes G(G^{\dagger}(x)) = G( s \boxtimes_{Q} G^{\dagger}(x)),
			\end{equation}
			so that there exists a resource $t$ in the set $s \boxtimes_{Q} G^{\dagger}(x)$ such that $t$ is also in $G^{-1}(s) = F(s)$.
			If we let $D \colon F(r) \to F(s)$ be defined by $D(s) \coloneqq t$, then $D$ is clearly a degradation since $G^{\dagger}(x)$ is an element of $Q_{\rm free}$.
		\end{proof}
		An alternative way to prove Lemma~\ref{thm:Sufficient condition for order-preserving maps 3} would be to show that conditions~(\ref{eq:Sufficient condition for order-preserving maps 3}) imply 
		\begin{equation}
			F \bigl( \up (q) \bigr) \subseteq \up \bigl( F(q) \bigr) \quad \forall\; r \in Q,
		\end{equation}
		which, by an argument analogous to the proof of Lemma~\ref{thm:Sufficient condition for order-preserving maps 2}, is a sufficient condition for the function
			$F \colon \ordres \to (\mathcal{P}(Q), \succeq_{\rm deg})$ to be order-preserving.
		One can check that $\mathcal{E}$ indeed satisfies conditions~(\ref{eq:Sufficient condition for order-preserving maps 3}) if $G$ is chosen to be the projection $\Pi_i \colon \R{k} \to R$.
		In fact,~(\ref{eq:star-morphism 2}) becomes an equality in this case.
			
		Lemma~\ref{thm:Sufficient condition for order-preserving maps 3} cannot be used, however, to show that the map $\mathcal{E}_{R_{\rm dc}}$ introduced in 
			Section~\ref{sec:Monotones From Contractions in General} is order-preserving.
		As we have seen explicitly in the proof of \Cref{thm:contractions}, proving this fact is not necessary when all we care about is the resulting monotone obtained from the root monotone $\fmin$ by the \ref{broad scheme}.
		This is because we can incorporate the restriction of the image of $\mathcal{E}$ to a downward closed set such as $W_{\rm dc} = R \times R_{\rm free}$ into the root monotone by restricting the domain of the contraction $f$.
		\Cref{lem:monotone_extensions} can then still be used to show that the function $f_{W_{\rm dc}}\textup{-}\mathsf{min}$ is a monotone.
		
		For completenes, we nevertheless include the following lemma which, when combined with the other results in this section, shows that the map $\mathcal{E}_{R_{\rm dc}}$ is itself order-preserving.
		\begin{lemma}[mediating maps with intersections]
			\label{thm:Generating more order-preserving maps}
			Let $\RT{R}$ and $(Q_{\rm free}, Q, \boxtimes_{Q})$ be resource theories and let $W_{\rm uc} \in \mathcal{UC}(Q)$ and 
				 $W_{\rm dc} \in \mathcal{DC}(Q)$ be upward and downward closed subsets of $Q$, respectively.
			\begin{enumerate}
				\item If $F \colon \ordres \to (\mathcal{P}(Q), \succeq_{\rm enh})$ is an order-preserving map, then the map $F_{W_{\rm uc}}$ defined by
					\begin{equation}
						F_{W_{\rm uc}} (r) \coloneqq F(r) \cap W_{\rm uc}
					\end{equation}
					is also order-preserving as a map $\ordres \to (\mathcal{P}(Q), \succeq_{\rm enh})$.
					 
				\item If $F \colon \ordres \to (\mathcal{P}(Q), \succeq_{\rm deg})$ is an order-preserving map, then the map $F_{W_{\rm dc}}$ defined by
					\begin{equation}
						F_{W_{\rm dc}} (r) \coloneqq F(r) \cap W_{\rm dc}
					\end{equation}
					is also order-preserving as a map $\ordres \to (\mathcal{P}(Q), \succeq_{\rm deg})$.
			\end{enumerate}
		\end{lemma}
		\begin{proof}
			This lemma is basically saying that the maps
			\begin{align*}
				(\mathcal{P}(Q), \succeq_{\rm enh}) &\to (\mathcal{P}(Q), \succeq_{\rm enh})  &  (\mathcal{P}(Q), \succeq_{\rm deg}) &\to (\mathcal{P}(Q), \succeq_{\rm deg}) \\
				S &\mapsto S \cap W_{\rm uc}  &   S &\mapsto S \cap W_{\rm dc} 
			\end{align*}
			are both order-preserving.
			The first one is order-preserving, because for any $S,T \in \mathcal{P}(R)$, we have
			\begin{equation}
				\begin{split}
					S \succeq_{\rm enh} T  &\iff  \down (S) \supseteq \down (T) \\
						&\,\implies  W_{\rm uc} \cap \down (S) \supseteq W_{\rm uc} \cap \down (T) \\
						&\,\implies  \down \bigl( W_{\rm uc} \cap \down (S) \bigr) \supseteq \down \bigl( W_{\rm uc} \cap \down (T) \bigr) \\
						&\iff \down (W_{\rm uc} \cap S) \supseteq \down (W_{\rm uc} \cap T) \\
						&\iff  W_{\rm uc} \cap S \succeq_{\rm enh} W_{\rm uc} \cap T.
				\end{split}
			\end{equation}
			The first equivalence follows from Lemma~\ref{lem:enh_down} and the penultimate one is a consequence of Lemma~\ref{thm:removing arrows} presented below.
			All in all, this concludes the proof of the first part of Lemma~\ref{thm:Generating more order-preserving maps}.
			
			The second part can be shown in an analogous way.
			In particular, the fact that the map $S \mapsto S \cap W_{\rm dc}$ is order-preserving follows by
			\begin{equation}
				\begin{split}
					S \succeq_{\rm deg} T  &\iff  \up (S) \subseteq \up (T) \\
						&\,\implies  W_{\rm dc} \cap \up (S) \subseteq W_{\rm dc} \cap \up (T) \\
						&\,\implies  \up \bigl( W_{\rm dc} \cap \up (S) \bigr) \subseteq \up \bigl( W_{\rm dc} \cap \up (T) \bigr) \\
						&\iff \up (W_{\rm dc} \cap S) \subseteq \up (W_{\rm dc} \cap T) \\
						&\iff  W_{\rm dc} \cap S \succeq_{\rm deg} W_{\rm dc} \cap T,
				\end{split}
			\end{equation}
			The first equivalence follows from Lemma~\ref{lem:deg_up} and the penultimate one is again a consequence of Lemma~\ref{thm:removing arrows}.
		\end{proof}
		
		\begin{lemma}
			\label{thm:removing arrows}
			Let $(\mathcal{A}, \succeq)$ be a preordered set and let $S$ and $T$ be two subsets of $\mathcal{A}$.
			Then we have
			\begin{subequations}
				\label{eq:removing arrows}
				\begin{align}
					\label{eq:down_up_down}
					\down \bigl( \up (S) \cap \down (T) \bigr) &= \down \bigl( \up (S) \cap T \bigr),\\
					\label{eq:up_down_up}
					\up \bigl( \down (S) \cap \up (T) \bigr) &= \up \bigl( \down (S) \cap T \bigr).
				\end{align}
			\end{subequations}
		\end{lemma}
		
		\begin{proof}
			First, let us prove equation~(\ref{eq:down_up_down}).
			The set on the right hand side is clearly a subset of the one on the left, so let's argue why also $\down \bigl( \up (S) \cap \down (T) \bigr) \subseteq \down \bigl( \up (S) \cap T \bigr)$ holds.
			For any $x \in \down \bigl( \up (S) \cap \down (T) \bigr)$, there exists a $y \in \up (S) \cap \down (T)$ such that $y \succeq x$.
			Therefore, there is a $t \in T$, such that $t \succeq y$ and $t \in \up (S)$. 
			Since $t \succeq x$, this means that $x \in  \down \bigl( \up (S) \cap T \bigr)$, thus proving equation~(\ref{eq:down_up_down}).
			
			Equation~(\ref{eq:up_down_up}) is the dual statement to~(\ref{eq:down_up_down}) and therefore it follows by reversing $\succeq$.
		\end{proof}
		
		Given the choice of $(Q_{\rm free}, Q, \boxtimes_{Q}) \coloneqq (\R{k}_{\rm cons}, \R{k}, \boxtimes)$ and $F \coloneqq \mathcal{E}$ in Lemma~\ref{thm:Generating more order-preserving maps}, we thus recover 
			the fact that $\mathcal{E}_{W_{\rm dc}} \colon \ordres \to (\mathcal{P}(\R{k}), \succeq_{\rm deg})$ is order-preserving. 
			
		\begin{corollary}[translating monotones between resource theories]
			\label{cor:translating monotones}
			Consider two resource theories $(Q_{\rm free}, Q, \boxtimes_{Q})$ and $\RT{R}$
			Given 
			\begin{itemize}
				\item a monotone $f \colon Q \to \reals$,
				\item a function $F \colon R \to \mathcal{P}(Q)$ that satisfies condition~(\ref{eq:down commute}), and 
				\item an upward closed set $W_{\rm uc} \in \mathcal{UC}(Q)$,
			\end{itemize}
			we get a monotone $\fmax \circ F_{W_{\rm uc}} \colon R \to \reals$ given by
			\begin{equation}
				\fmax \circ F_{W_{\rm uc}} (r) = \sup \Set*[\big]{ f(q)  \given  q \in F(r) \cap W_{\rm uc} }.
			\end{equation}
			
			Similarly, given 
			\begin{itemize}
				\item a monotone $f \colon Q \to \reals$,
				\item a function $F \colon R \to \mathcal{P}(Q)$ that satisfies conditions~(\ref{eq:Sufficient condition for order-preserving maps 3}), and 
				\item a downward closed set $W_{\rm dc} \in \mathcal{DC}(Q)$,
			\end{itemize}
			we get a monotone $\fmin \circ F_{W_{\rm dc}} \colon R \to \reals$ given by
			\begin{equation}
				\fmin \circ F_{W_{\rm dc}} (r) = \inf \Set*[\big]{ f(q)  \given  q \in F(r) \cap W_{\rm dc} }.
			\end{equation}
		\end{corollary}
		
		\begin{remark}
			Note that the generalized yield and cost constructions (Theorem~\ref{thm:yield and cost}) can be applied in succession with those from Corollary~\ref{cor:translating monotones}.
			However, neither of these commute in general.
			For example, for generic $f$ and $F$, composing the constructions that use infima give
			\begin{equation}
				\begin{split}
					\cost{(\fmin \circ F_{W_{\rm uc}})}{}(q) &= \inf f\bigl( W_{\rm uc} \cap F (\up (q))\bigr) \\
						&\neq  \inf f \bigl( \up (W_{\rm uc} \cap F (q))\bigr) \\
						&= \xmin{(\cost{f}{})} \circ F_{W_{\rm uc}} (q).
				\end{split}
			\end{equation}
		\end{remark}

	\section{Assessing Informativeness of Monotones in General Resource Theories}
\label{sec:Ordering monotones}

	The general monotone constructions (Theorem~\ref{thm:yield and cost} and Corollary~\ref{cor:translating monotones}) have several inputs that need to be specified to obtain a single resource monotone.
	In this section, we would like to address the question of which choices of these input parameters are good in the sense that they yield a useful resource monotone.
	In order to assess the usefulness of monotones as far as characterizing a preordered set $(\mathcal{A},\succeq_{\mathcal{A}})$ is concerned, we define a preorder $\sqsupseteq_{\mathcal{A}}$ on the set of monotones itself.
	We denote this set by $\mathsf{Mon}(\mathcal{A})$.
	It is just the collection of all order-preserving maps from $(\mathcal{A},\succeq_{\mathcal{A}})$ to $\ordreals$.
	In this context, we consider a monotone $f$ to be more ``useful'' than a monotone $g$ if it contains all of the information about $(\mathcal{A},\succeq_{\mathcal{A}})$ that $g$ does and possibly more.	
	We now formalize what we mean by the amount of inforamtion a monotone has about a preordered set.
	
	A function $f \colon \mathcal{A} \to \reals$ is a monotone if and only if for all pairs $(a,b) \in \mathcal{A} \times \mathcal{A}$, the following implication holds:
	\begin{equation}
		\label{eq:monotone inference}
		f(a) < f(b) \implies a \not \succeq_{\mathcal{A}} b.
	\end{equation}
	That is, monotones contain information about the order relation $\succeq_{\mathcal{A}} \subseteq \mathcal{A} \times \mathcal{A}$ insofar as they witness when pairs of elements of $\mathcal{A}$ are \emph{not} related by 
		$\succeq_{\mathcal{A}}$.
	Of course, if $f(a) \geq f(b)$, the implication above doesn't let us learn anything about the order relation $\succeq_{\mathcal{A}}$.
	Given a monotone $f$, a pair $(a,b)$ is henceforth called \mbox{$\bm{f}$-\textbf{interesting}} if $f(a) < f(b)$ holds. 
	The $f$-interesting pairs are those, for which we can learn that $a \not \succeq_{\mathcal{A}} b$ holds from the fact that $f$ is a monotone.
	
	The set of all $f$-interesting pairs for a monotone $f$ is denoted by 
	\begin{equation}
		\label{eq:f-interesting relation for monotones}
		\Interesting{f}{\mathcal{A}, \succeq_{\mathcal{A}}} \coloneqq \Set*[\big]{(a,b) \in \mathcal{A} \times \mathcal{A}  \given  f(a) < f(b) }.  
	\end{equation}
	We also refer to $\Interesting{f}{\mathcal{A}, \succeq_{\mathcal{A}}}$ as the $f$-interesting relation on $\mathcal{A}$.
	
	\begin{definition}
		\label{def:Monotone preorder}
		Let $(\mathcal{A},\succeq_{\mathcal{A}})$ be a preordered set and let $\mathsf{Mon}(\mathcal{A})$ be the set of order-preserving maps $(\mathcal{A},\succeq_{\mathcal{A}}) \to \ordreals$.
		We define a preorder $\sqsupseteq_{\mathcal{A}}$ on $\mathsf{Mon}(\mathcal{A})$ by
		\begin{equation}
			\label{eq:Monotone preorder}
			f \sqsupseteq_{\mathcal{A}} g  \quad\iff\quad  \Interesting{f}{\mathcal{A}, \succeq_{\mathcal{A}}} \supseteq \Interesting{g}{\mathcal{A}, \succeq_{\mathcal{A}}}
		\end{equation}
		and we say that $f$ is \textbf{more informative about} $\bm{(\mathcal{A},\succeq_{\mathcal{A}})}$ than $g$ is if $f \sqsupseteq_{\mathcal{A}} g$ holds.
		If the preordered set is $\ordres$ arising from some resource theory $\RT{R}$, we denote the informativeness order relation by $\sqsupseteq$ instead of 
		$\sqsupseteq_{R}$.
	\end{definition}
	For functions $f$ and $g$ which are monotones, we can express $f \sqsupseteq_{\mathcal{A}} g$ also as
	\begin{subequations}
		\label{eq:Informative relation for monotones}
		\begin{align}
			\label{eq:Informative relation for monotones 1}
			f \sqsupseteq_{\mathcal{A}} g  &\quad\iff\quad  \forall a,b \in \mathcal{A} : g(a) < g(b) \implies f(a) < f(b) \\
			\label{eq:Informative relation for monotones 2}
				&\quad\iff\quad  \forall a,b \in \mathcal{A} : f(a) \geq f(b) \implies g(a) \geq g(b).
		\end{align}
	\end{subequations}

	We would like to compare the constructions of monotones appearing in Sections~\ref{sec:yield and cost} and~\ref{sec:Translating Monotones} in terms of how useful they are depending on the input elements thereof.
	One of the input elements for cost and yield constructions is a partial function $f_W \colon R \to \reals$.
	Although it need not be a monotone on its domain $W$, it can still be understood as witnessing nonconvertibility between some resources within $W$. 
	In Proposition~\ref{thm:yield and cost from more informative functions} below we prove that whenever a partial function $f_W$ witnesses all the pairs of nonconvertible resources that $g_{W'}$ does, 
		then $f_W$ is at least as useful as $g_{W'}$ is, when thought of as an input to the generalized yield and cost constructions (Theorem~\ref{thm:yield and cost}).
	That is, in such case $\yield{f_W}{D}$ is more informative about $\ordres$ than $\yield{g_{W'}}{D}$ is and likewise for the cost construction.
	In order to make these kinds of statements more precise, we now formalize the notion of the amount of resource nonconvertibility that a partial function witnesses.
	
	Let $f_W \colon R \to \reals$ be a partial function with domain $W$.
	We say that $f_W$ witnesses the nonconvertibility of a pair of resources $(r,s)$ if both $f(r) < f(s)$ and $r \not \succeq s$ hold.
	As far as this property is concerned, we call such a pair of resources $(r,s)$ \mbox{$\bm{f_W}$-\textbf{interesting}}.
	
	The set of all $f_W$-interesting pairs for a partial function $f_W$ is denoted by 
	\begin{equation}
		\Interesting{f_W}{R, \succeq} \coloneqq \Set*[\big]{ (r,s) \in W \times W  \given  f(r) < f(s) \;\land\;  r \not \succeq s }
	\end{equation}
	We also refer to $\Interesting{f_W}{R, \succeq}$ as the $f_W$-interesting relation on $R$.
	Note that this definition coincides with the $f$-interesting relation for a monotone $f$ given by equation~(\ref{eq:f-interesting relation for monotones}), whenever $f_W$ is indeed a monotone.
	That is why we use the same notation for both of these relations.	

	\begin{definition}
		\label{def:Function preorder}
		Let $\ordres$ be a preordered set and let $f_W, g_{W'} \colon R \to \reals$ be partial functions with domains $W$ and $W'$ respectively.
		We say that $f_W$ \textbf{witnesses more resource nonconvertibility in} $\bm{\ordres}$ than $g_{W'}$ does if $f_W \sqsupseteq g_{W'}$ holds, where
		\begin{equation}
			\label{eq:Function preorder}
			f_W \sqsupseteq g_{W'}  \quad\iff\quad  \Interesting{f_W}{R,\succeq} \supseteq \Interesting{g_{W'}}{R,\succeq}.
		\end{equation}
	\end{definition}
	
	\begin{proposition}[more informative monotones from more informative functions]
		\label{thm:yield and cost from more informative functions}
		Let $\RT{R}$ be a resource theory with an associated preordered set $\ordres$ and let $D$ be a downward closed subset of $R$.
		Furthermore, let $f_W \colon R \to \reals$ and $g_{W'} \colon R \to \reals$ be two partial functions with domains $W$ and $W'$ respectively.
		
		If $f_W$ witnesses more resource nonconvertibility in $\ordres$ than $g_{W'}$ does, then $\yield{f_W}{D}$ is more informative about $\ordres$ than $\yield{g_{W'}}{D}$ is and also 
			$\cost{f_W}{D}$ is more informative about $\ordres$ than $\cost{g_{W'}}{D}$ is.
		That is, we have
		\begin{subequations}
			\label{eq:yield and cost from more informative functions}
			\begin{align}
				\label{eq:yield from more informative functions}
				f_W \sqsupseteq g_{W'} &\implies \yield{f_W}{D} \sqsupseteq \yield{g_{W'}}{D}, \\
				\label{eq:cost from more informative functions}
				f_W \sqsupseteq g_{W'} &\implies \cost{f_W}{D} \sqsupseteq \cost{g_{W'}}{D}.
			\end{align}
		\end{subequations}
		
		Moreover, if $f_W$ and $g_{W'}$ are monotones on their respective domains and their domains coincide; i.e., $W = W'$; and if $D = R_{\rm free}$, then the converse of both implications holds as well. 
		That is, in such case we have
		\begin{subequations}
			\label{eq:yield and cost from more informative monotones}
			\begin{align}
				\label{eq:yield from more informative monotones}
				f_W \sqsupseteq g_{W} &\iff \yield{f_W}{D} \sqsupseteq \yield{g_{W}}{D}, \\
				\label{eq:cost from more informative monotones}
				f_W \sqsupseteq g_{W} &\iff \cost{f_W}{D} \sqsupseteq \cost{g_{W}}{D}.
			\end{align}
		\end{subequations}
	\end{proposition}
	
	\begin{proof}
		In order to prove claim~(\ref{eq:yield from more informative functions}), we need to show that $\yield{g_{W'}}{D}(r) < \yield{g_{W'}}{D}(s)$ implies $\yield{f_W}{D}(r) < \yield{f_W}{D}(s)$ for all 
			$r, s \in R$ such that $r \not\succeq s$. 
		This follows via
		\begin{subequations}
			\begin{align}
				\yield{g_{W'}}{D}(r) &< \yield{g_{W'}}{D}(s) \\ 
					\label{eq:yield from more informative functions 2}
					&\iff  \forall r_2 \in W' \cap \downwrt{D} (r) ,\; \exists s_2 \in W' \cap \downwrt{D} (s) \;:\; g_{W'}(r_2) < g_{W'}(s_2) \text{ and } r_2 \not \succeq s_2 \\
					\label{eq:yield from more informative functions 3}
					&\,\implies  \forall r_2 \in W \cap \downwrt{D} (r) ,\; \exists s_2 \in W \cap \downwrt{D} (s) \;:\; f_W(r_2) < f_W(s_2) \text{ and } r_2 \not \succeq s_2 \\
					&\iff \yield{f_W}{D}(r) < \yield{f_W}{D}(s).
			\end{align}
		\end{subequations}
		In the first (and last) equivalence, we could restrict $s_2$ to be such that $r_2 \not \succeq s_2$ because $r_2 \succeq s_2$ (together with $r_2 \in \downwrt{D} (r)$) implies that $s_2$ is an element of 
			$\downwrt{D} (r)$, which is a subset of $\down \circ \downwrt{D} (r)$ whenever $D$ is downward closed, as we prove in Lemma~\ref{thm:composing image maps}.
		This in turn implies that $g_{W'}(s_2)$ is bounded above by $\yield{g_{W'}}{D}(r)$.
		Since $\yield{g_{W'}}{D}(s)$ is strictly larger than $\yield{g_{W'}}{D}(r)$, there must be such $s_2$ outside $\downwrt{D} (r)$.
		
		The implication \mbox{(\ref{eq:yield from more informative functions 2}) $\implies$ (\ref{eq:yield from more informative functions 3})} follows from the assumption that $f_W$ witnesses more resource nonconvertibility in 
			$\ordres$ than $g_{W'}$ does.
		
		In order to prove claim~(\ref{eq:cost from more informative functions}), we need to show the analogous statement for cost monotones.
		\begin{subequations}
			\begin{align}
				\cost{g_{W'}}{D}(r) &< \cost{g_{W'}}{D}(s) \\ 
					\label{eq:cost from more informative functions 2}
					&\iff  \forall s_2 \in W' \cap \upwrt{D} (s) ,\; \exists r_2 \in W' \cap \upwrt{D} (r) \;:\; g_{W'}(r_2) < g_{W'}(s_2) \text{ and } r_2 \not \succeq s_2 \\
					\label{eq:cost from more informative functions 3}
					&\,\implies  \forall s_2 \in W \cap \upwrt{D} (s) ,\; \exists r_2 \in W \cap \upwrt{D} (r) \;:\; f_W(r_2) < f_W(s_2) \text{ and } r_2 \not \succeq s_2 \\
					&\iff \cost{f_W}{D}(r) < \cost{f_W}{D}(s).
			\end{align}
		\end{subequations}
		In the first (and last) equivalence, we can again restrict $r_2$ to be such that $r_2 \not \succeq s_2$ because $r_2 \succeq s_2$ (together with $s_2 \in \upwrt{D} (s)$) implies that $r_2$ is an element of 
			\mbox{$\upwrt{D} (s) \subseteq \up \circ \upwrt{D} (s)$} as we show in Lemma~\ref{thm:composing image maps}.
		In turn, this implies that $g_{W'}(r_2)$ is bounded below by $\cost{g_{W'}}{D}(s)$.
		Since $\cost{g_{W'}}{D}(r)$ is strictly smaller than $\cost{g_{W'}}{D}(s)$, there must be such $r_2$ outside $\downwrt{D} (s)$.
		The implication \mbox{(\ref{eq:cost from more informative functions 2}) $\implies$ (\ref{eq:cost from more informative functions 3})} follows from the assumption that $f_W \sqsupseteq g_{W'}$ holds.
		This concludes the proof of the first half of Proposition~\ref{thm:yield and cost from more informative functions}.
		
		Finally, in order to obtain claim~(\ref{eq:yield and cost from more informative monotones}), we can show that $f_W \not\sqsupseteq g_W$ implies both $\yield{f_W}{} \not\sqsupseteq \yield{g_{W}}{}$ and 
			$\cost{f_W}{} \not\sqsupseteq \cost{g_{W}}{}$, under the assumption that $f_W$ and $g_W$ are monotones on $W$. 
		The statement $f_W \not\sqsupseteq g_W$ can in such case be expressed as
		\begin{equation}
			\label{eq:f not above g}
			\exists r,s \in W \,:\, g_W(r) < g_W(s) \text{ and } f_W(r) \geq f_W(s).
		\end{equation}
		By Proposition~\ref{prop:Extensions of monotones are the same on the original domain} proved below, the values of $\yield{f_W}{D}$ and $\cost{f_W}{D}$ coincide with the value of $f_W$ on $W$, and similarly for $g_W$. 
		Therefore, $f_W \not\sqsupseteq g_W$ implies the following two statements
		\begin{align}
			\exists r,s \in W &\,:\, \yield{g_W}{D}(r) < \yield{g_W}{D}(s) \text{ and } \yield{f_W}{D}(r) \geq \yield{f_W}{D}(s) \\
			\exists r,s \in W &\,:\, \cost{g_W}{D}(r) < \cost{g_W}{D}(s) \text{ and } \cost{f_W}{D}(r) \geq \cost{f_W}{D}(s).
		\end{align}
		Since the yields and costs are also monotones, these imply that $\yield{f_W}{D} \not\sqsupseteq \yield{g_{W}}{D}$ and \mbox{$\cost{f_W}{D} \not\sqsupseteq \cost{g_{W}}{D}$}.
		Consequently, the proof of the second half of Proposition~\ref*{thm:yield and cost from more informative functions} is also complete.
	\end{proof}
		
	
		\begin{proposition}
			\label{prop:Extensions of monotones are the same on the original domain}
			Let $f_W \colon W \to \mathbb{R}_{+}$ be a monotone. Then for all $r \in W$, we have
			\begin{equation}
				\yield{f_W}{}(r) = f_W(r) = \cost{f_W}{}(r).
			\end{equation}
		\end{proposition}
		
		\begin{proof}
			Since $r \in \down (r)$ and $r \in \up(r)$, we have $\cost{f_W}{}(r) \geq f_W(r) \geq \yield{f_W}{} (r)$ for all $r \in W$. 
			On the other hand, $f$ being a monotone on $W$ implies that for each $s,t \in W$ such that $t \preceq r \preceq s$, we have $f_W(t) \leq f_W(r) \leq f_W(s)$. 
			Performing a supremum of the left inequality over all $t \in W \cap \down(r)$ yields $\cost{f_W}{}(r) \leq f_W(r)$, while taking the infimum of the right inequality over all $s \in W \cap \up(r)$ gives 
				$f_W(r) \leq \yield{f_W}{} (r)$, so that the result follows.
		\end{proof}
	
%
	
	\begin{corollary}
		As a consequence of Proposition~\ref{thm:yield and cost from more informative functions}, sufficient and necessary conditions for the ordering (by $\sqsupseteq$) of generalized yields and costs relative to 
			$R_{\rm free}$\footnotemark{} are given by the ordering (by $\sqsupseteq$) of their restrictions to $W \cup W'$.
		\footnotetext{In fact, the same result holds for generalized yields and costs relative to any downward closed set $D$.
		However, this does not follow directly from Proposition~\ref{thm:yield and cost from more informative functions}.
		One needs to use the fact that $\yield{f_W}{D}$ and $\cost{f_W}{D}$ preserve not only the convertibility relation with respect to $R_{\rm free}$ (i.e., $\succeq$), but also the convertibility relation with respect 
			to $D$.
		Note that the latter relation is not in general transitive, since $D$ may not be closed under $\boxtimes$.}%
		These facts can also be expressed in terms of the order relation with respect to informativeness about $(W \cup W', \succeq)$ as:
		\begin{subequations}
			\label{eq:Informativeness on domain is equivalent to informativeness on R}
			\begin{align}
				\label{eq:yield from more informative functions 4}
				\yield{f_W}{} \sqsupseteq_{W \cup W'} \yield{g_{W'}}{} &\iff \yield{f_W}{} \sqsupseteq \yield{g_{W'}}{} \\
				\label{eq:cost from more informative functions 4}
				\cost{f_W}{} \sqsupseteq_{W \cup W'} \cost{g_{W'}}{} &\iff \cost{f_W}{} \sqsupseteq \cost{g_{W'}}{}
			\end{align}
		\end{subequations}
	\end{corollary}
	
	Therefore, if one wishes to characterize the resource preorder $\ordres$ by virtue of monotones generated by the generalized yield and cost constructions (Theorem~\ref{thm:yield and cost}), then using functions 
		$W \to \reals$ which are more informative about $\ordres$ according to $\sqsupseteq$ should be preferred.
		
	A function $f_W$ cannot witness more resource nonconvertibility than a complete set of monotones because such a set captures all the information in the preordered set $(W, \succeq)$.
	Nonetheless, a function $f_W$ can witness more resource nonconvertibility than \emph{any single} monotone. 
	The simplest example is provided by $W$ with $4$ elements, two pairs of which are ordered as in the following Hasse diagram.
	\begin{equation}
		\label{eq:preorder example}
		\begin{tikzcd}
			r_1 \arrow[d]	&	s_1 \arrow[d] \\
			r_2				&	s_2
		\end{tikzcd}
	\end{equation}
	If we let $f_W$ be defined as follows 
	\begin{equation}
		\begin{split}
			f_W(r_1) &= 0  \\
			f_W(r_2) &= 1	
		\end{split}
		\qquad
		\begin{split}
			f_W(s_1) &= 0 \\
			f_W(s_2) &= 1
		\end{split}
	\end{equation}
	then it clearly fails to be a monotone.
	Note that $f_W$ witnesses nonconvertibility for the two pairs of resources, $(r_1,s_2)$ and $(s_1,r_2)$, while no single monotone can do so simultaneously.
	The proof of the latter claim is that the order-preserving property of a monotone implies that it must satisfy $M(r_1) \geq M(r_2)$ and $M(s_1) \geq M(s_2)$. 
	If it witnesses the nonconvertibility of the pair $(r_1,s_2)$, then $M(r_1) < M(s_2)$ holds and these three inequalities together imply that $M(s_1)$ is greater than $M(r_2)$, so that $M$ then \emph{cannot} witness the 
		nonconvertibility of the pair $(s_1,r_2)$. 
	The function $f_W$ is capable of witnessing the nonconvertibility of both pairs of resources \emph{precisely because} it fails to be order-preserving.

	\begin{remark}
		Note that this function has another interesting property in that both $\yield{f_W}{}$ and $\cost{f_W}{}$ are constant; i.e., they are least informative about $(W, \succeq)$ among all monotones $W \to \reals$.
		$f_W$ can thus serve as a counterexample to conjectures regarding conditions under which the yield and cost constructions generate useful monotones.
	\end{remark}
	
	\begin{example}[chains admit a most informative function]
		If $W$ is a chain; that is, a totally ordered subset of $R$; then there \emph{is} a single monotone $f_W$ that forms a complete set of monotones by itself. 
		Therefore, it is more informative about $(W, \succeq)$ than any other function $W \to \reals$.
		As a consequence, for each downward closed set $D \in \mathcal{DC}(R)$, there are unique most informative yield and cost monotones with respect to $W$, namely $\yield{f_W}{D}$ and $\cost{f_W}{D}$.
		Given the choice of $D = R_{\rm free}$, these correspond to currencies defined in~\cite{Kraemer2016}, as we mentioned earlier.
	\end{example}

	\section{Conclusions}
\label{sec:Conclusions}

	To summarize, in this manuscript we introduced a somewhat minimal framework for describing (universally combinable) resource theories, within which we investigated various ways of constructing monotones through the lens of 
		the \ref{broad scheme}.
		
	Firstly, we looked at generalized resource yield and generalized resource cost constructions.
	An extensive (but definitely not exhaustive) list of examples of such constructions in the literature has been provided.
	We also showed how the generalized constructions of this kind can be used to obtain monotones that cannot be conceptualized as standard yield and cost constructions.
	
	Secondly, we looked at monotones which can be seen as arising from another monotone by virtue of a translation via a mediating order-preserving map.
	After introducing a resource theory of $k$-distinguishability, we described a translation of measures of distinguishability to other resource theories, which generalizes the familiar concept of constructing monotones 
		from contractions.
	As an application of this construction, we unified resource weight and resource robustness as arising from a single root contraction for $3$-tuples of resources. 
	Moreover, by varying the parameters in the general construction, we showed how one can obtain other related monotones from the same measure of $3$-distinguishability.
	General methods to translate monotones between resource theories were then presented, culminating with Corollary~\ref{cor:translating monotones} that summarizes the results on translation of monotones.

	The two main themes of the paper, generalizing yield and cost constructions and generalizing the translation of monotones, are intricately linked.
	We investigate these connections further in Appendix~\ref{sec:Proofs}, by showing how the corresponding mediating preordered sets are related to each other.
	One of the main features that distinguishes the two kinds of constructions is that the partial function $f$, a starting point in both cases, is assumed to have different properties.
	The scheme for translating monotones is targeted to functions $f$ which are monotones themselves, while yield and cost constructions work for functions $f$ which are not monotones, for example by virtue of a 
		restricted domain of applicability.
	One could think that the generalized yield and cost constructions are therefore superior.
	However, their disadvantage is the optimization over \emph{all} free resources inherent in the construction.
	Moreover, even though $f$ has to be a monotone if we want to translate it, a seemingly insipid one (like the convex alignment) can still generate interesting target monotones (like weight and robustness).
		
	Finally, we also explored the structure of the set of all resource monotones.
	The monotone constructions presented here are very general and widely applicable, but using them in practice as a method for generating monotones involves several choices.
	For example, a priori it is not clear which choices of the downward closed set of resources $D$, the valuation function $f$, and its domain $W$ in the generalized yield and cost constructions are the best ones.
	These are the kinds of questions we made progress on by ``assessing informativeness of monotones''. 
	In particular, we compared them in terms of how good they are in capturing the resource ordering.
	With this criterion, we investigated what are the best ways to use the monotone costructions introduced earlier in order to get the most informative monotones.
	
	Our work advances the studies of general structures appearing in resource theories and has potential applications to any area where the resource-theoretic point of view is of some use.
	These include the study of information theory, both quantum and classical, but also of thermodynamics, of renormalization, and of various other parts of physics where resource-theoretic questions are tackled.
	We believe, however, that similar questions in more distant fields can also be analyzed with the resource-theoretic mindset, which is one of the main reasons why we choose to work in a framework that does not presuppose 
		the resources to be quantum processes.
	Indeed, one of the benefits of working within an abstract framework for resource theories is that there is the potential for cross-fertilization of ideas between very different fields of study. 
	This was one of the motivations for previous attempts at abstract formalisms for resources theories~\cite{Coecke2016,Fritz2017}, which can describe situations well beyond the scope of physics.
	To name a few, we can use them to study the theory of chemical reactions, but also a kind of proof theory wherein the free operations are compass and straight-edge and the nontrivial resources are 
		geometrical constructions that cannot be achieved by compass and straight-edge. 
	In this vein, Fritz has further shown how the framework of ordered commutative monoids has interesting applications in fields as diverse as graph theory and game theory~\cite{Fritz2017}.

	There are many possible future directions for extending this work.
	\begin{enumerate}
		\item One might aim to determine how the mathematical structures presented here relate to other mathematical structures used in physics, mathematics, and computer science.
			The study of their relation to some of the other mathematical frameworks for resource theories can be found in \cite{Gonda2021}, but other connections are yet to be developed.
			
		\item One might try to devise general techniques for constructing monotones by considering resource theories that have more structure than we have presumed here.
			One way to do so would be by strenghtening the assumptions of our central results in order to arrive at stronger conclusions. 
			There are many possibilities in this direction, one of which is to assume a linear or convex structure of resources as we did when we studied the weight and robustness measures here.
			It is clear that these results will then be connected to ideas from convex geometry~\cite{Regula2017} and convex optimization~\cite{uola2019quantifying}. \\
			Related to this is the aim of reexpressing the results presented here in a framework which is closer to the structure of a resource theory that one would use in practical applications.
			A framework like that would capture partitioned process theories with a restriction on the allowed types of resources for example, which, as we argued in Example~\ref{ex:counterexample}, cannot be expressed 
				as a universally combinable resource theory in the sense of Definition~\ref{def:resource theory}.
			This is what we do in~\cite{Gonda2021}. 
			
		\item Last, but not least, one would hope to be able to not only unify and generalize existing concrete results about resource theories as we have done here, 
				but also to find novel applications of monotones with the help of the conceptual clarity arising from the abstract point of view.
	\end{enumerate}
	

	\newpage
	
	\appendix

	\section{Ordering Sets of Resources}
\label{sec:Proofs}

	In this appendix, we show the following isomorphisms of preordered sets:
	\begin{align}
		\label{eq:Partial order of downward closed sets}
		\newfaktor{(\mathcal{P}(R), \succeq)}[0.6]{\sim} \simeq \newfaktor{(\mathcal{P}(R), \succeq_{\rm enh})}[0.6]{\sim_{\rm enh}} &\simeq (\mathcal{DC}(R), \supseteq)  \\
		\label{eq:Partial order of upward closed sets}
		\newfaktor{(\mathcal{P}(R), \succeq_{\rm deg})}[0.6]{\sim_{\rm deg}} &\simeq (\mathcal{UC}(R), \subseteq)
	\end{align}
	When restricted to singletons and the free images/preimages of individual resources, all five partially ordered sets are isomorphic.
	In fact, it is immediate from the definitions that $\{r\} \succeq_{\rm enh} \{s\}$ if and only if $\{r\} \succeq_{\rm deg} \{s\}$, whence the preorders $\succeq$, $\succeq_{\rm enh}$ and $\succeq_{\rm deg}$ are 
		themselves identical when restricted to singletons.

	\begin{theorem}[first isomorphism theorem for preordered sets \cite{Gratzer2008}]
		\label{thm:First isomorphism theorem for prerdered sets}
		Let $(\mathcal{A}, \succeq_{\mathcal{A}})$ and $(\mathcal{B}, \succeq_{\mathcal{B}})$ be two preordered sets and let $\phi \colon \mathcal{A} \to \mathcal{B}$ be an order-preserving map.
		The kernel of $\phi$ is an equivalence relation $\sim_{\phi}$ on $\mathcal{A}$ defined by 
		\begin{equation}
			a \sim_{\phi} a'  \iff  \phi(a) \sim_{\mathcal{B}} \phi(a'),
		\end{equation}
		where $\sim_{\mathcal{B}}$ is the standard equivalence relation on $\mathcal{B}$ induced by $\succeq_{\mathcal{B}}$.
		Then, there is a canonical isomorphism
		\begin{equation}
			\tilde{\phi} \colon \newfaktor{(\mathcal{A}, \succeq_{\mathcal{A}})}{\sim_{\phi}} \to \newfaktor{(\phi(\mathcal{A}), \succeq_{\mathcal{B}})}{\sim_{\mathcal{B}}}.
		\end{equation}
	\end{theorem}
	
	\begin{proof}
		$(\mathcal{A}, \succeq_{\mathcal{A}}) / {\sim_{\phi}}$ consists of the set of equivalence classes $\mathcal{A} / {\sim_{\phi}}$ and the corresponding order relation $\succeq_{\mathcal{A}}$ defined as
		\begin{equation}
			[a_1]_{\sim_{\phi}} \succeq_{\mathcal{A}} [a_2]_{\sim_{\phi}}  \iff  a_1 \succeq_{\mathcal{A}} a_2
		\end{equation}
		for any $a_1, a_2 \in \mathcal{A}$, where $[a_1]_{\sim_{\phi}}$ is the equivalence class of $a_1$ with respect to $\sim_{\phi}$.
		Since $\phi$ is an order-preserving map, $\succeq_{\mathcal{A}}$ is a well-defined partial order.
		
		We can then construct $\tilde{\phi}$ as
		\begin{equation}
			\tilde{\phi} ([a]_{\sim_{\phi}}) \coloneqq  [\phi(a)]_{\sim_{\mathcal{B}}},
		\end{equation}
		for any $a \in \mathcal{A}$. 
		Again, $\tilde{\phi}$ is well-defined and order-preserving because $\phi$ is order-preserving.
		Furthermore, the fact that it is injective follows from the definition of $\sim_{\phi}$.
	\end{proof}
	
	\begin{lemma}
		\label{thm:Kernel of free image map}
		The kernel of $\down \colon (\mathcal{P}(R), \succeq_{\rm enh}) \to (\mathcal{DC}(R), \supseteq)$ is $\sim_{\rm enh}$ and consequently the kernel of
			$\down \colon \ordres \to (\mathcal{DC}(R),\supseteq)$ is $\sim$.
	\end{lemma}
	
	\begin{proof}
		Follows directly from Lemma~\ref{lem:enh_down}.
	\end{proof}
	
	\begin{corollary}
		The partially ordered sets $(\mathcal{P}(R), \succeq_{\rm enh}) / {\sim_{\rm enh}}$ and $(\mathcal{DC}(R), \supseteq)$ are isomorphic.
	\end{corollary}
	
	\begin{lemma}
		\label{thm:Kernel of free preimage map}
		The kernel of $\up \colon (\mathcal{P}(R), \succeq_{\rm deg}) \to (\mathcal{UC}(R), \supseteq)$ is $\sim_{\rm deg}$ and consequently the kernel of
			$\up \colon \ordres \to (\mathcal{UC}(R),\supseteq)$ is $\sim$.
	\end{lemma}
	
	\begin{proof}
		Follows directly from Lemma~\ref{lem:deg_up}.
	\end{proof}
	
	\begin{corollary}
		The partially ordered sets $(\mathcal{P}(R), \succeq_{\rm deg}) / {\sim_{\rm deg}}$ and $(\mathcal{UC}(R), \subseteq)$ are isomorphic.
	\end{corollary}
	
	\newpage
	\section{Overview of Monotone Constructions}\label{sec:monotone_overview}
	
	\begingroup
		\renewcommand{\arraystretch}{1.21} 
		\setlength{\tabcolsep}{12pt} 
		\begin{table}[!h]{}
			\makebox[\linewidth]{
			\begin{tabular}{c p{0.46\textwidth} c}
				\textbf{Name} & \centering \textbf{Yield and cost for $f$ defined on $\res$} & \textbf{Text reference} \\
				\hline
					$\yield{f}{}(r)$ & \centering $\sup \Set{ f(s)  \given  s \in R_{\rm free} \boxtimes r }$ & \Cref{eq:yield_basic} \\
					$\cost{f}{}(r)$ & \centering $\inf \Set{ f(s)  \given  r \in R_{\rm free} \boxtimes s }$ & \Cref{eq:cost_basic} \\
						\cdashline{2-2}
					for $f \colon \res \to \reals$ & \tabitem e.g.\ $f$ $=$ channel dimension in a RT of communication. & \Cref{ex:dim_cost}  \\
				\hline
				\hline
				  & \centering \textbf{Yield and cost, $f$ defined on a subset} &   \\
				\hline
					$\yield{f_W}{}(r)$ & \centering $\sup \Set{ f_W(s)  \given  s \in R_{\rm free} \boxtimes r, \; s \in W }$ & \Cref{eq:yield_subset} \\
					$\cost{f_W}{}(r)$ & \centering $\inf \Set{ f_W(s)  \given  r \in R_{\rm free} \boxtimes s ,\; s \in W }$ & \Cref{eq:cost_subset} \\ 
						\cdashline{2-2}
					where $f_W \colon \res \to \reals$ & \centering \textbf{Currencies \cite{Kraemer2016} for $W$ a chain:}  & \\ 
					has domain $W \subseteq \res$ & \tabitem $W$ $=$ the set of $n$-fold products of e-bits in the RT of bipartite entanglement \cite{Horodecki2009}. & \Cref{ex:currencies} \\
					 & \tabitem $W$ $=$ the set of sharp states in the RT of nonuniformity \cite{Gour2015}. & \Cref{ex:currencies} \\
					 & \tabitem$W$ $=$ a chain of boxes between the PR box and a free box in the RT of nonclassicality of common-cause boxes \cite{Wolfe2019}. & \Cref{ex:qcorrelations} \\
						\cdashline{2-2}
					 & \centering \textbf{$W$ as convexly extremal resources:}  & \\ 
					 & \tabitem $W$ $=$ pure quantum states in the RT of entanglement. & \Cref{ex:ent_rank_cost} \\
						\cdashline{2-2}
					 & \centering \textbf{$W$ as processes of particular type:}  & \\ 
					 & \tabitem $W$ $=$ the set of states of arbitrary dimension in any RT of channels \cite{gour2019quantify}. & \Cref{ex:Changing type} \\
				\hline
				\hline
				  & \centering \textbf{Yield and cost w.r.t.\ a downset $D$} &  \\
				\hline
					$\yield{f_W}{D}(r)$ & \centering $\sup \Set{ f_W(s)  \given  s \in D \boxtimes r, \; s \in W }$ & \Cref{eq:yield_downset} \\
					$\cost{f_W}{D}(r)$ & \centering $\inf \Set{ f_W(s)  \given  r \in D \boxtimes s ,\; s \in W }$ & \Cref{eq:cost_downset} \\ 
						\cdashline{2-2}
					where $D \subseteq \res$ is & \centering \textbf{Downsets with $D \boxtimes D = D$:}  & \\ 
					s.t.\ $D \boxtimes \resfree = D$ & \tabitem $D$ $=$ separable operations in the RT of LOCC-entanglement &  \\
					 & \tabitem $D$ $=$ Gibbs-preserving operations in the RT of athermality \cite{Brandao2013}. &  \\
					 & \tabitem $D$ $=$ processes covariant w.r.t.\ a subgroup of $G$ in a resource theory of $G$-asymmetry \cite{Marvian2014}. & \Cref{examplesDCsets} \ref{ex:symmetricDCsets} \\
						\cdashline{2-2}
					 & \centering \textbf{Adding non-free resources to $\resfree$:}  & \\ 
					 & \tabitem $D = \resfree \boxtimes \rho$ for a state $\rho$ and for $W$ the set of states, thus extending a monotone $f_W$ from states to other processes. & \Cref{ex:Advantage of generalized yield} \\
					 & \tabitem $D = \resfree \boxtimes S$ for $S$ the states with bounded entanglement rank. & \Cref{examplesDCsets} \ref{ex:entrankDCsets} \\
					 & \tabitem $S$ $=$ the set of states that are not intrinsically 3-way entangled. & \Cref{examplesDCsets} \ref{ex:3-wayentangledDCsets} \\
					 & \tabitem $S$ = one-way quantum communication channels in either direction in the RT of LOCC-entanglement. &  \\
					 & \tabitem $S$ = one-way classical communication channels in either direction in the RT of LOSR-entanglement \cite{schmid2020understanding}. &  
			\end{tabular}
			}
		\end{table}

		\begin{table}[!ht]{}
			\makebox[\linewidth]{
			\begin{tabular}{c p{0.46\textwidth} c}
				\textbf{Name} & \centering \textbf{Monotones from 2-contractions} & \textbf{Text reference} \\
				\hline
					$M_{f} (r)$ & \centering $\inf \Set*[\big]{ f(r,s)  \given  s \in R_{\rm free} }$ & \Cref{eq:GeneralMonotoneFromContraction} \\
					$M_{f,D} (r)$ & \centering $\inf \Set*[\big]{ f(r,s)  \given  s \in D }$ & \Cref{eq:min_dis_downset} \\
						\cdashline{2-2}
					for $f \colon \R{2} \to \reals$ & \tabitem $D$ is any downset as above. & \Cref{examplesMonotonesDCsets} \\
					a 2-contraction & \tabitem For $f$ = relative entropy in the RT of entanglement, $M_f$ is the relative entropy of entanglement \cite{vedral1997quantifying}. &   \\
				\hline
				\hline
				 & \centering \textbf{Monotones from functions that commute with $\resfree$} &  \\
				\hline
					$M_{\Phi, f} (r)$ & \centering $f \bigl( r, \Phi(r) \bigr)$ & \Cref{eq:monotone_from_free_commuting} \\
						\cdashline{2-2}
					for $\Phi \colon \res \to \res$ s.t.\ & \tabitem $\Phi$ $=$ twirling map in a RT of asymmetry. & \Cref{ex:Monotones from contractions 2} \\
					$\Phi(t \boxtimes r) = t \boxtimes \Phi(r)$ & \tabitem $\Phi$ $=$ constant map to the thermal state in RT of athermality.  & \Cref{ex:dist_from_thermal}  \\
					for all $t \in \resfree$ & \tabitem $\Phi$ $=$ a map sending a bipartite state to the product of its marginals in the RT of correlations where $\resfree$ consists of local processes w.r.t.\ the bipartition. & \Cref{ex:dist_from_product} \\
				\hline
				\hline
				 & \centering \textbf{Monotones from 3-contractions} &  \\
				\hline
					$M_{f,D_1,D_2} (r)$ & \centering $\inf \Set*[\big]{ f(r,s_1,s_2)  \given  s_i \in D_i }$ & \Cref{thm:contractions} \\
						\cdashline{2-2}
					for $f \colon \R{3} \to \reals$ & e.g.\ $D_1 = \res$, $D_2 = \resfree$ and $f = \mathsf{cva}$ gives & \Cref{def:cva} \\
					a 3-contraction, and & \tabitem the resource weight $M_{\rm w} (r)$ given by &  \\
					downsets $D_1$, $D_2$ & \centering $\inf \Set*[\big]{ \lambda  \given  r \in \lambda \res + (1-\lambda) \resfree },$  & \Cref{eq:res_weight} \\
					 & while choosing $f(r,s_1,s_2) = \mathsf{cva}(s_2,s_1,r)$ instead gives  &  \\
					 & \tabitem the robustness $M_{\rm rob} (r)$ given by &  \\
					 & \centering $\inf \Set*[\big]{ \lambda  \given  \lambda s + (1-\lambda) r \in \resfree ,\: s \in \res }$. & \Cref{eq:res_robustness} 
			\end{tabular}
			} \: 
		\end{table}
	\endgroup

	\clearpage
	\bibliographystyle{plainnat}
	\bibliography{references}{}

\end{document}